\documentclass[12pt]{article}
\makeatletter
\RequirePackage[english]{babel}
\RequirePackage[shortlabels]{enumitem}
\RequirePackage[mathscr]{euscript}
\RequirePackage[margin=1.25in]{geometry}
\RequirePackage[small]{titlesec}
\RequirePackage[dvipsnames]{xcolor}
\RequirePackage{adjustbox,amsbsy,amsmath,amssymb,amsthm,bibunits,bm,booktabs,commath,chngcntr,dsfont,econometrics,etoolbox,fancyhdr,float,gensymb,graphicx,IEEEtrantools,longtable,marginnote,mathrsfs,mathtools,mdframed,natbib,needspace,parskip,pgf,placeins,setspace,siunitx,subfigure,tabularx,textcomp,threeparttable,tikz}

\RequirePackage[bookmarks=false,%
                colorlinks=true,%
                citecolor=ucladarkblue,%
                linkcolor=pastelred,%
                pdftoolbar=false,%
                pdfmenubar=true]{hyperref}
\RequirePackage{cleveref}

\AddToHook{cmd/appendix/after}{%
  \crefalias{section}{appendix}%
  \crefalias{subsection}{subappendix}%
  \crefalias{subsubsection}{subsubappendix}%
}

\allowdisplaybreaks
\titlespacing*{\section}{0pt}{2ex plus .5ex minus .2ex}{1ex}
\titlespacing*{\subsection}{0pt}{1.5ex plus .3ex minus .2ex}{.8ex}

\newenvironment{tabnotes}
  {\begin{tablenotes}[flushleft]\footnotesize\item[]\emph{Notes:}\ }
  {\end{tablenotes}}

\newcommand{\floatnotes}[2][\linewidth]{%
  \par\medskip\parbox{#1}{\footnotesize\emph{Notes:}\ #2}}

\newcommand\@shorttitle{}
\newcommand\shorttitle[1]{\renewcommand\@shorttitle{#1}}
\newcommand\@supplementtitle{}
\newcommand\@supplementauthor{}
\newcommand\@supplementdate{}
\newcommand\@supplementand{}
\let\@setupmaketitle\maketitle
\renewcommand{\maketitle}{%
  \global\let\@supplementtitle\@title
  \global\let\@supplementauthor\@author
  \global\let\@supplementdate\@date
  \global\let\@supplementand\and
  \@setupmaketitle
}
\renewcommand{\sectionmark}[1]{\markboth{#1}{}}

\renewcommand{\E}{\mathbb{E}}
\renewcommand{\P}{\mathbb{P}}
\newcommand{\ind}{\mathds{1}}

\DeclareMathOperator*{\argmin}{arg\,min}
\DeclareMathOperator{\Cov}{Cov}
\DeclareMathOperator{\supp}{supp}
\DeclareMathOperator{\sign}{sign}
\DeclareMathOperator{\Var}{Var}

\definecolor{uclablue}{HTML}{2774AE}
\definecolor{ucladarkblue}{RGB}{0,85,135}
\definecolor{pastelred}{HTML}{d62d0e}

\theoremstyle{definition}

\newtheorem*{example*}{Example}

\newtheorem*{remark*}{Remark}

\theoremstyle{plain}
\newtheorem{assumption}{Assumption}
\crefname{assumption}{Assumption}{Assumptions}
\Crefname{assumption}{Assumption}{Assumptions}
\newtheorem*{assumption*}{Assumption}
\newtheorem{corollary}{Corollary}
\newtheorem*{corollary*}{Corollary}

\newtheorem*{definition*}{Definition}
\newtheorem{lemma}{Lemma}
\newtheorem*{lemma*}{Lemma}
\newtheorem{prop}{Proposition}
\newtheorem*{prop*}{Proposition}
\newtheorem{theorem}{Theorem}
\newtheorem*{theorem*}{Theorem}

\newif\if@supplementbibunit
\newcommand{\resetsupplementcounters}{%
  \setcounter{section}{0}%
  \setcounter{subsection}{0}%
  \setcounter{subsubsection}{0}%
  \setcounter{paragraph}{0}%
  \setcounter{subparagraph}{0}%
  \setcounter{page}{1}%
  \setcounter{equation}{0}%
  \setcounter{figure}{0}%
  \setcounter{table}{0}%
  \setcounter{footnote}{0}%
  \setcounter{assumption}{0}%
  \setcounter{corollary}{0}%
  \setcounter{definition}{0}%
  \setcounter{example}{0}%
  \setcounter{lemma}{0}%
  \setcounter{prop}{0}%
  \setcounter{remark}{0}%
  \setcounter{theorem}{0}%
}
\newcommand{\supplement}{%
  \clearpage
  \appendix
  \counterwithout{equation}{section}%
  \counterwithout{corollary}{section}%
  \counterwithout{example}{section}%
  \counterwithout{lemma}{section}%
  \counterwithout{prop}{section}%
  \resetsupplementcounters
  \renewcommand{\thesection}{S\arabic{section}}%
  \renewcommand{\thepage}{S\arabic{page}}%
  \renewcommand{\theequation}{S\arabic{equation}}%
  \renewcommand{\thefigure}{S\arabic{figure}}%
  \renewcommand{\thetable}{S\arabic{table}}%
  \renewcommand{\theassumption}{S\arabic{assumption}}%
  \renewcommand{\thecorollary}{S\arabic{corollary}}%
  \renewcommand{\thedefinition}{S\arabic{definition}}%
  \renewcommand{\theexample}{S\arabic{example}}%
  \renewcommand{\thelemma}{S\arabic{lemma}}%
  \renewcommand{\theprop}{S\arabic{prop}}%
  \renewcommand{\theremark}{S\arabic{remark}}%
  \renewcommand{\thetheorem}{S\arabic{theorem}}%
  \@ifundefined{theHpage}{}{\renewcommand{\theHpage}{supp.\arabic{page}}}%
  \@ifundefined{theHsection}{}{\renewcommand{\theHsection}{supp.\arabic{section}}}%
  \@ifundefined{theHsubsection}{}{\renewcommand{\theHsubsection}{supp.\arabic{section}.\arabic{subsection}}}%
  \@ifundefined{theHsubsubsection}{}{\renewcommand{\theHsubsubsection}{supp.\arabic{section}.\arabic{subsection}.\arabic{subsubsection}}}%
  \@ifundefined{theHequation}{}{\renewcommand{\theHequation}{supp.equation.\arabic{equation}}}%
  \@ifundefined{theHfigure}{}{\renewcommand{\theHfigure}{supp.figure.\arabic{figure}}}%
  \@ifundefined{theHtable}{}{\renewcommand{\theHtable}{supp.table.\arabic{table}}}%
  \@ifundefined{theHassumption}{}{\renewcommand{\theHassumption}{supp.assumption.\arabic{assumption}}}%
  \@ifundefined{theHcorollary}{}{\renewcommand{\theHcorollary}{supp.corollary.\arabic{corollary}}}%
  \@ifundefined{theHdefinition}{}{\renewcommand{\theHdefinition}{supp.definition.\arabic{definition}}}%
  \@ifundefined{theHexample}{}{\renewcommand{\theHexample}{supp.example.\arabic{example}}}%
  \@ifundefined{theHlemma}{}{\renewcommand{\theHlemma}{supp.lemma.\arabic{lemma}}}%
  \@ifundefined{theHprop}{}{\renewcommand{\theHprop}{supp.proposition.\arabic{prop}}}%
  \@ifundefined{theHremark}{}{\renewcommand{\theHremark}{supp.remark.\arabic{remark}}}%
  \@ifundefined{theHtheorem}{}{\renewcommand{\theHtheorem}{supp.theorem.\arabic{theorem}}}%
  \gdef\@extra@b@citeb{.supp}%
  \gdef\@extra@binfo{.supp}%
  \@supplementbibunittrue
  \begin{bibunit}%
}
\newcommand{\supplementtitlepage}[1][Supplemental Appendix to ``\@supplementtitle'']{%
  \begingroup
  \renewcommand\thefootnote{\@fnsymbol\c@footnote}%
  \let\thanks\@gobble
  \let\and\@supplementand
  \thispagestyle{plain}%
  \null
  \vskip 2em%
  \begin{center}
    {\LARGE #1\par}%
    \vskip 1.5em%
    {\large
      \lineskip .5em%
      \begin{tabular}[t]{c}%
        \@supplementauthor
      \end{tabular}\par}%
    \vskip 1em%
    {\large \@supplementdate}%
  \end{center}%
  \endgroup
  \par\bigskip
}
\newcommand{\startmainbibliography}{%
  \if@filesw
    \immediate\openout\@bibunitaux bu.aux
    \immediate\write\@mainaux{\string\@input{bu.aux}}%
  \fi
}
\newcommand{\finishmainbibliography}{%
  \if@filesw
    \immediate\closeout\@bibunitaux
  \fi
}
\newcommand{\startsupplementbibliography}{%
  \begingroup
  \renewcommand{\bibsection}{\section*{\refname}}%
  \@startbibunitorrelax
}
\newcommand{\finishsupplementbibliography}{%
  \endgroup
  \end{bibunit}%
  \gdef\@extra@b@citeb{}%
  \gdef\@extra@binfo{}%
  \@supplementbibunitfalse
}
\AtEndDocument{\if@supplementbibunit\end{bibunit}\fi}

\defaultbibliographystyle{chicago}

\newcommand{\En}{\E_n}
\newcommand{\Ebar}{\overline{\E}}
\DeclareMathOperator{\rank}{rank}
\DeclareMathOperator{\diver}{div}
\makeatother
\numberwithin{equation}{section}

\title{Choosing the Dictionary and Penalty for IV-LASSO\thanks{An \texttt{R} package implementing IV-LASSO with the first-stage implementation guided by this paper can be found on \href{https://github.com/mnavjeev/ivamse}{GitHub}. This package, along with the proofs of several technical lemmas in this paper, were developed with the help of Claude Fable. We take responsibility for all errors.}}
\author{Yukun Ma\thanks{Yukun Ma: yma69@ur.rochester.edu. University of Rochester}
\and Manu Navjeevan\thanks{Manu Navjeevan: mnavjeevan@tamu.edu. Texas A\&M University}
\and Bohdan Salahub\thanks{Bohdan Salahub: bsalahub@g.ucla.edu. University of California, Los Angeles}}
\date{November 8, 2026}

\AddToHook{env/example/begin}{\crefalias{section}{example}}
\AddToHook{env/remark/begin}{\crefalias{section}{remark}}
\AddToHook{env/assumption/begin}{\crefalias{section}{assumption}}
\AddToHook{env/corollary/begin}{\crefalias{section}{corollary}}
\AddToHook{env/definition/begin}{\crefalias{section}{definition}}
\AddToHook{env/lemma/begin}{\crefalias{section}{lemma}}
\AddToHook{env/prop/begin}{\crefalias{section}{prop}}
\AddToHook{env/theorem/begin}{\crefalias{section}{theorem}}

\begin{document}

\maketitle

\begin{abstract}
Estimating an instrumental variables (IV) first stage with the least absolute shrinkage and selection operator (LASSO) requires choosing a dictionary and penalty. Existing rules target first-stage objectives rather than precision of the structural parameter estimate. We derive a feasible approximate mean squared error (AMSE) criterion for full-sample IV-LASSO with one endogenous regressor and homoskedastic Gaussian errors. For a prespecified list of dictionary--penalty pairs, the criterion trades off approximation of the optimal instrument against a complexity cost that depends on the first two moments of the LASSO effective dimension and on the degree of endogeneity. The selected pair attains the smallest candidate-specific AMSE component up to vanishing relative error. In simulations calibrated to \citet{PRSW-2015}, the criterion improves substantially on plug-in rules at low endogeneity and on cross-validation at high endogeneity, with nearly identical results under non-Gaussian errors. A supplementary appendix formally extends analysis to a class of non-Gaussian errors.

\end{abstract}

\begin{center}
\end{center}
{}

\section{Introduction}
\label{sec:introduction}

When empirical researchers have access to a large set of technical instruments, LASSO makes estimation of an IV first-stage feasible but leaves researchers to choose how heavily to penalize the fit and which dictionary of technical instruments to use.%
\footnote{See \citet{PRSW-2015},
\citet{gilchrist-glassberg-2016},
\citet{ReimersWaldfogel-2021},
\citet{BaldChynHastingsMachelett-2022},
\citet{Derenoncourt-2022},
\citet{RyanSudarshan-2022},
\citet{AganDoleacHarvey-2023},
\citet{BazziEtAl-2023},
\citet{BerajaEtAl-2023},
\citet{DoJacoby-2025},
\citet{BazziEtAl-2026}, and
\citet{CuddyCurrie-2026} for prominent applications of IV-LASSO and IV-post-LASSO methods in empirical economics.} %
Under standard sparse-first-stage conditions, \citet{BCCH-2012} show that all sufficiently accurate LASSO first-fits yield IV estimators with the same first-order limiting distribution. Standard asymptotic theory therefore does not rank admissible implementations, even though their finite-sample performance can differ. 

\Cref{fig:bcch-paths-intro} illustrates this issue in the eminent-domain application of \citet{BCCH-2012}. Across three nested first-stage dictionaries, the IV estimate changes materially as the penalty parameter varies. Moreover, different methods of selecting the first-stage penalty parameter can change whether the conventional pointwise 95\% confidence interval for the resulting structural parameter estimate contains zero. For example, for the Case--Shiller home price index outcome, using cross-validation to select the penalty yields a significant result on the Restricted dictionary while using Stein's Unbiased Risk Estimate (SURE) does not. Given this variation, how should an applied researcher choose among possible first-stage implementations?

\begin{figure}[!htbp]
\centering
\includegraphics[width=\textwidth]{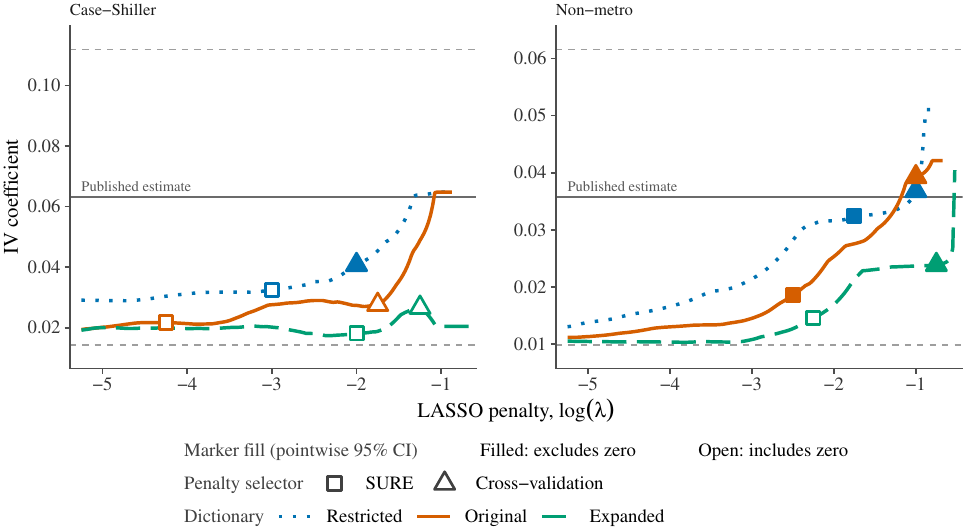}
\caption{IV Estimates and Conventional Significance across First-Stage
Implementations}
\label{fig:bcch-paths-intro}
\floatnotes{Lines trace IV-LASSO estimates over LASSO penalties within three nested dictionaries. Original is the published dictionary; Restricted removes cross-products, and Expanded adds all nondegenerate second-order
products. Filled markers denote pointwise 95\% heteroskedasticity robust confidence intervals that exclude zero.  The solid gray line and dashed bounds show the published IV-post-LASSO estimate and confidence interval.  The panels report the two outcomes for which that interval excludes zero (183 Case--Shiller and 110 non-metro circuit-year observations).}
\end{figure}

Existing first-stage tuning rules are not designed to target accuracy of the resulting structural parameter estimates. The calibrated penalty rule of \citet{BCCH-2012} and the bootstrap-after-cross-validation method of \citet{ChetverikovSorensen-2025} are designed to control first-stage estimation error, while SURE and ordinary cross-validation target first-stage prediction
\citep{TibshiraniTaylor-2012,ChetverikovLiaoChernozhukov-2021}. Neither of these approaches necessarily yield high-quality structural parameter estimates. A LASSO fit with a lighter penalty or richer dictionary may better approximate the true first stage, but it can also fit more of the first-stage noise. When the first-stage noise is correlated with the structural error, this in turn introduces a bias into the resulting structural parameter estimate whose size grows with the degree of endogeneity. The optimal first-stage implementation therefore depends not only on goodness of first-stage fit and instrument strength, but also on the level of endogeneity.

This paper develops second-order asymptotic theory for the full-sample IV-LASSO estimator and uses it to provide a principled method for choosing the first-stage implementation. We focus on the leading case of a single endogenous regressor and work in a homoskedastic Gaussian model. As in \citet{DonaldNewey-2001}, who build on the higher-order expansions of \citet{Nagar-1959} to choose among instrument sets, we take a finite, prespecified list of first-stage LASSO candidates, each consisting of a dictionary and a penalty, and rank them by the approximate mean squared error (AMSE) of the IV-LASSO estimators they produce.  The candidate specific part of the AMSE can be decomposed into a first-stage approximation cost and a complexity cost, with complexity measured by a collinearity-robust analogue of the number of instruments selected by the LASSO. The size of this complexity cost rises with the squared covariance between the first-stage and structural errors, reflecting the discussion above.

This characterization leads to a feasible criterion that can be computed from each candidate’s fitted first stage and its effective dimension, a measure of first-stage complexity closely related to the number of variables selected by the LASSO fit \citep{TibshiraniTaylor-2012}. Under our regularity conditions, the pair that minimizes this criterion has an AMSE as small as that of the best candidate in the list, up to vanishing relative error. This guarantee is stated relative to the second-order components that distinguish the candidates, rather than the first-order risk that all of them share. In a supplementary appendix, we extend these results to a class of non-Gaussian error distributions under some further regularity conditions, showing that the same feasible criterion minimizes AMSE in this more general model.

To derive the AMSE of the IV-LASSO estimator, we need to account for both the shrinkage and the variable selection performed by the first-stage LASSO. Although the LASSO fitted values are continuous in the endogenous variable, they need not be differentiable at points where the selected set of instruments changes. This lack of full differentiability complicates higher-order analysis, which typically relies on a smooth expansion of the first-stage estimator. Moreover, since our assumptions allow for coefficients near the selection threshold, we cannot avoid this difficulty by assuming that the selected model is eventually fixed. We address this by using the geometric results of \citet{TibshiraniTaylor-2012} and the second-order Stein identities of \citet{BellecZhang-2021} to characterize moments involving the LASSO fit and structural errors. The first-stage LASSO fit is Lipschitz continuous and thus differentiable almost everywhere, allowing these calculations despite its nonsmoothness. 

Constructing a feasible criterion presents a further difficulty. Because the LASSO first-stages may converge slowly, estimation error in the feasible criterion can asymptotically dominate the candidate-specific AMSE terms. For our proposed feasible selection rule, we show that the leading part of this error is common to all candidates and therefore cancels when they are compared. After this common centering, the remaining estimation error is uniformly negligible relative to the candidate-specific AMSE components. Estimating the complexity cost requires fluctuations in the realized LASSO effective dimension to be small relative to its expectation. Using similar arguments to \citet{BellecZhang-2021}, we obtain this stability under a partial beta-min condition on a growing group of coefficients large enough for their instruments to be selected with probability approaching one. This group must account for a nonvanishing share of the nonzero coefficients used to approximate the first stage. The condition also makes the AMSE approximation error small relative to the terms used to compare candidates, playing a similar role in our analysis to the growing basis condition of \citet{DonaldNewey-2001}.

To evaluate the performance of the proposed criterion, we conduct a simulation study calibrated to the credit-supply and export data of \citet{PRSW-2015}. We hold the original instruments fixed and vary the strength of identification and degree of endogeneity. At low levels of endogeneity, cross-validation and the proposed criterion perform comparably. As endogeneity rises, however, the criterion selects less complex first stages and achieves up to 50\% lower mean squared error than cross-validation in designs with the highest levels of endogeneity. Compared to the cross-fit IV-LASSO estimator, the proposed criterion has substantially lower mean squared error when endogeneity is low and performs comparably when it is high. Interestingly, the relative performance of the criterion is nearly identical in the Gaussian designs and in the two non-Gaussian designs we consider, suggesting that the AMSE in more general models is well approximated by the AMSE in a Gaussian model.

Our results extend the literature on selecting and regularizing instruments according to the AMSE of the structural parameter estimator \citep{DonaldNewey-2001,KuersteinerOkui-2010,Okui-2011,Carrasco-2012} to high-dimensional LASSO first stages, where each dictionary--penalty candidate itself selects instruments using the observed endogenous variable. Unlike a least-squares first stage on a prespecified instrument basis, fixing the dictionary and penalty does not fix the selected model. Our analysis accounts for this selection uncertainty as well as the shrinkage introduced by the penalty.

Our analysis also relates to a growing literature on the use of high-dimensional and machine-learning methods in econometrics \citep{BelloniChernozhukovHansen-2014,CCDDHNR-2018,WagerAthey-2018,FarrellLiangMisra-2021} and an established literature studying higher-order asymptotic properties of econometric estimators \citep{Sargan1976,Rothenberg1984,NeweySmith2004,HahnHausmanKuersteiner2004}. \citet{Velez-2026} develops second-order approximations for double/debiased machine learning estimators to study the choice of the number of cross-fitting folds. His higher-order analysis assumes a stochastic linear expansion of the first-stage estimators. We instead account directly for LASSO shrinkage and variable selection to choose the dictionary and penalty.

\Cref{sec:framework} presents the model and candidate first stages.
\Cref{sec:results} states the main results, and \Cref{sec:empirical-application}
reports the simulations. Proofs are in the main appendix, while a supplementary appendix extends our analysis to a class of non-Gaussian errors. 

\emph{Notation.}  For \(m\in\mathbb N\), let \([m]:=\{1,\ldots,m\}\).  For
an array \(\{a_i\}_{i\in[n]}\), define
\begin{equation*}
    \En[a_i]:=\frac1n\sum_{i\in[n]}a_i,
    \qquad
    \Ebar[a_i]:=\En[\E[a_i]].
\end{equation*}
{}
\section{Model and Setup}
\label{sec:framework}

To simplify exposition, we present the analysis assuming that any exogenous regressors have already been partialled out of both the outcome and the endogenous regressor. Our analysis assumes there is only a single endogenous variable.\footnote{A formal treatment of the model with controls can be found in the arXiv version of this paper \citep{MaNavjeevanSalahub-2026}. An extension to multiple endogenous variables is available on request.} The structural equation and first stage are given by
\begin{equation}
    \label{eq:model}
    \begin{split}
        y_i &= \beta x_i + \varepsilon_i \\[-0.5pt]
        x_i &= \Pi_i + v_i
    \end{split}
\end{equation}
We assume the researcher observes \(n\) independent observations \((y_i,x_i,z_i)\), where \(y_i,x_i\in\SR\) and \(z_i\in\SR^p\). In \eqref{eq:model}, we let \(\Pi_i=\E[x_i\mid z_i]\) denote the conditional expectation of the endogenous variable given the instruments and assume that the structural error is conditional-mean independent of the instruments, so that \(\E[\varepsilon_i\mid z_i]=\E[v_i\mid z_i]=0\). Throughout the theoretical analysis, we condition on the realized instruments and write \(x,\Pi,\varepsilon,v\in\SR^n\) for the corresponding stacked vectors.\footnote{That is, \(x = (x_1,\dots,x_n)\in \SR^n\).} All expectations and probability statements below are therefore understood as conditional on this realization, i.e., \(\Pi_i=\E[x_i]\).
\begin{assumption}[Conditional Gaussian Sampling]
\label{assm:gaussian} Assume
\begin{enumerate}[(i)] 
    \item The errors \(\{(\varepsilon_i, v_i)\}_{i=1}^n\) are independent of the instruments and independently and identically distributed according to 
    \begin{equation*}
            \begin{pmatrix}\varepsilon_i\\v_i\end{pmatrix}
        \sim N(0,\Omega), \qquad 
        \Omega :=\begin{pmatrix} \sigma_\varepsilon^2&\sigma_{\varepsilon v}\\ \sigma_{\varepsilon v}&\sigma_v^2
       \end{pmatrix},
    \end{equation*}
    where \(\Omega\) is a fixed nonsingular matrix.
    \item For constants \(0<\underline H<\overline H<\infty\),
    \begin{equation}
        \label{eq:strength-endogeneity}
        \underline H\le H:=\En[\Pi_i^2]\le\overline H, \qquad \sigma_{\varepsilon v}\ne0.
    \end{equation}
\end{enumerate}
\end{assumption}
\Cref{assm:gaussian}(i) presents the conditionally homoskedastic Gaussian model, while \Cref{assm:gaussian}(ii) assumes the structural parameter is strongly identified and that there is non-negligible endogeneity problem, i.e., the structural parameter \(\beta\) cannot be recovered from ordinary least squares (OLS).

Under conditional homoskedasticity, the conditional mean \(\Pi_i\) is the optimal instrument: among IV estimators based on a nonrandom instrument, using \(\Pi_i\) itself minimizes the asymptotic variance of the resulting structural parameter estimate.\footnote{For nonrandom \(g\in\SR^n\) with \(\En[g_i^2]\) bounded and \(\En[g_i\Pi_i]\) bounded away from zero, the leading variance is \(\sigma_\varepsilon^2\En[g_i^2]/\En[g_i\Pi_i]^2\ge\sigma_\varepsilon^2/H\), with equality when \(g\) is proportional to \(\Pi\). The bound \(\sigma_\varepsilon^2/H\) is the fixed-design counterpart of the usual optimal-IV efficiency bound; see \citet{Chamberlain-1987} and \citet{Newey-1990}.} Since \(\Pi_i\) is unknown, feasible estimation replaces it with a first-stage estimate. Any IV estimator \(\widehat\beta\) based on a sufficiently well-estimated first stage (\Cref{sec:results} states sufficient conditions) admits the same expansion,
\begin{equation}
    \label{eq:first-order}
    \sqrt n(\widehat\beta-\beta)
    =\frac{\sqrt n\,\En[\Pi_i\varepsilon_i]}{H}+o_p(1),
\end{equation}
whose leading term does not involve the first-stage estimate and, conditional on the
instruments, is distributed exactly as \(N(0,\sigma_\varepsilon^2/H)\) at every sample size. Comparing first-stage estimators thus requires analyzing terms hidden in the remainder.

\subsection{Candidate First Stages and the IV Estimator}

We select from a prespecified list of \(J\) candidates, where \(J\) is fixed. Each \(c\in[J]\) consists of a dictionary of technical instruments and a nonrandom penalty \(\lambda_c>0\). The dictionary maps \(z_i\) to \(z_{c,i}:=(z_{c,i1},\dots,z_{c,ip_c})'\in\SR^{p_c}\), whose entries are fixed functions of \(z_i\). The dictionary dimension \(p_c\) may exceed \(n\), and both the dictionary and the penalty may depend on \(n\). Write
\begin{equation*}
    Z_c:=(z_{c,1},\dots,z_{c,n})'\in\SR^{n\times p_c},
    \qquad
    Z_{c,B}:=(Z_c)_{\cdot,B}\in\SR^{n\times |B|}
    \quad\text{for }B\subseteq[p_c].
\end{equation*}

For each candidate implementation \(c\in[J]\), we define the LASSO coefficient
estimate and fitted first stage \citep{Tibshirani-1996} by
\begin{equation}
    \widehat\pi_c
    \in\argmin_{\pi\in\mathbb R^{p_c}}
    \left\{
       \frac12\En[(x_i-z_{c,i}'\pi)^2]
       +\lambda_c\lVert\pi\rVert_1
    \right\},
    \qquad
    \widehat\Pi_c:=Z_c\widehat\pi_c.
    \label{eq:lasso}
\end{equation}
When the columns of \(Z_c\) are linearly dependent, the coefficient
minimizer need not be unique, but the fitted vector \(\widehat\Pi_c\) remains
unique \citep{TibshiraniTaylor-2012,Tibshirani-2013}. Using the fitted values
\(\widehat\Pi_{c,i}:=z_{c,i}'\widehat\pi_c\) as instruments, we obtain the
IV-LASSO estimator
\begin{equation}
    \label{eq:iv-estimator}
    \widehat\beta_c
    \coloneqq
    \begin{cases}
       \displaystyle
       \frac{\En[\widehat\Pi_{c,i}y_i]}
            {\En[\widehat\Pi_{c,i}x_i]},
       &\En[\widehat\Pi_{c,i}^2]>0,\\[2mm]
       0,
       &\En[\widehat\Pi_{c,i}^2]=0.
    \end{cases}
\end{equation}

A natural criterion for comparing the candidate estimators
\(\widehat\beta_c\) is their scaled mean squared error (MSE), \(\E[n(\widehat\beta_c-\beta)^2].\)
The exact MSE of an IV estimator, however, need not exist in finite samples.
To see this, consider the infeasible estimator that uses the optimal instrument
\(\Pi_i\) in place of \(\widehat\Pi_{c,i}\) in \eqref{eq:iv-estimator}. This
estimator satisfies
\begin{equation*}
    \widehat\beta-\beta
    =\frac{\En[\Pi_i\varepsilon_i]}{\En[\Pi_ix_i]}.
\end{equation*}
Under \Cref{assm:gaussian}, the numerator and denominator on the right-hand
side are jointly Gaussian. Because \(\Omega\) is nonsingular, the numerator
is not an exact linear function of the denominator, and their ratio does not have a finite second moment. More generally, even when the exact MSE exists, it
may be neither tractable nor consistently estimable.

Given this restriction, we instead follow the approach of
\citet{DonaldNewey-2001}. Under suitable
conditions, we show that the scaled squared error of each estimator \(\widehat\beta_c\) admits a decomposition of the form
\begin{equation}
    \label{eq:decomp}
    n (\widehat\beta_c - \beta)^2 = C_0 + Q_c + r_c.
\end{equation}
In the above, \(C_0\) is common to all candidates and captures the component of
squared error corresponding to the leading term in \eqref{eq:first-order}.
The second term, \(Q_c\),  captures the second-order contribution of candidate \(c\)'s first-stage procedure, while \(r_c\) is asymptotically negligible in the sense that \( \max_{c\in[J]}|r_c|/\E[Q_c]\to_p 0.\)

We define \(\E[C_0+Q_c]\) as the AMSE of the candidate IV-LASSO estimator \(\widehat\beta_c\). For the infeasible estimator based on the
optimal instrument, the expectation of the candidate-specific component is
zero. Hence, \(\bar Q_c:=\E[Q_c]\) measures the excess AMSE of candidate
\(c\) relative to that benchmark. Because \(C_0\) is common across candidates, ranking candidates by AMSE is
equivalent to ranking them by \(\bar Q_c\). In \Cref{sec:results}, we derive a
closed-form asymptotic approximation to \(\bar Q_c\), construct a consistent
feasible counterpart, and use it to select among the first-order-equivalent
candidate estimators.

\subsection{The Effective Dimension}
\label{sec:effective-dimension}

A key element of our analysis is the effective dimension of the LASSO fit, which will provide a measure of first-stage complexity. As discussed above, the number of selected instruments is not always well defined. With collinear dictionary columns, LASSO coefficient solutions can have different supports even though their fitted values are identical. The effective dimension therefore measures first-stage complexity using the rank of the equicorrelation set.

Letting \(u_c:=x-\widehat\Pi_c\) denote the LASSO residual from candidate implementation \(c\in [J]\), define the equicorrelation set by
\begin{equation}
    E_c
    :=\left\{
       j\in[p_c]:
       \left|\En[z_{c,ij}u_{c,i}]\right|=\lambda_c
     \right\}.
    \label{eq:equicorrelation}
\end{equation}
This is the set of dictionary columns for which the LASSO Karush--Kuhn--Tucker (KKT) inequality binds. Because the residual is unique, \(E_c\) does not depend on which coefficient solution is used, and it contains the support of every solution.
We define the effective dimension of the first-stage fit as
\begin{equation}
    d_c:=\rank(Z_{c,E_c}).
    \label{eq:realized-rank}
\end{equation}
The quantity \(d_c\) is observable and invariant to the choice of coefficient solution. When the LASSO solution is unique, \(d_c\) agrees with the number of selected instruments almost surely under our sampling assumption.\footnote{A sufficient condition for uniqueness is that the dictionary columns are in general position, which holds with probability one when the dictionary matrix has a joint continuous distribution \citep{Tibshirani-2013}.} As shown by \citet{TibshiraniTaylor-2012}, the degrees of freedom of the LASSO fit are given by \(\E[d_c]\).
{}
\section{Main Results}
\label{sec:results}

After stating the conditions required for our analysis, we characterize the candidate-specific AMSE components (\Cref{thm:main}) and show that the selected candidate attains the smallest of them up to vanishing relative error (\Cref{thm:feasible}).

\subsection{Conditions on the Candidate First Stages}
\label{sec:primitive}

Along with the Gaussian sampling and strong identification conditions in \Cref{assm:gaussian}, we impose conditions for first-stage prediction consistency and a partial beta-min condition that is used to establish concentration of the LASSO effective dimension but does not require exact support recovery.

We use the approximately sparse representation of \citet{BCCH-2012}, in which for each candidate dictionary \(Z_c\) there is a nonrandom coefficient vector \(\pi_c^0\in\SR^{p_c}\) and a prediction error \(\xi_{c,i}\) such that
\begin{equation}
    \Pi_i=z_{c,i}'\pi_c^0+\xi_{c,i},
    \qquad
    T_c:=\supp(\pi_c^0),
    \qquad
    s_c:=|T_c|,
    \label{eq:sparse-approximation}
\end{equation}
where \(\supp(b):=\{j:b_j\ne0\}\), \(s_c\ge1\), and \(\xi_c=(\xi_{c,1},\ldots,\xi_{c,n})'\). The rate restrictions on \(s_c\) and \(\xi_c\) are stated below.

For a vector \(\delta\) and an index set \(T\), write \(\delta_T:=(\delta_j)_{j\in T}\)
and \(\lVert\delta\rVert_0:=|\supp(\delta)|\).  Fix a constant
\(c_\lambda>1\), common to all candidates, and set
\(L_\lambda:=(c_\lambda+1)/(c_\lambda-1)\).
Following
\citet{BCCH-2012}, define the restricted eigenvalue \citep{BickelRitovTsybakov-2009,VandeGeerBuhlmann-2009}
\begin{equation}
    \kappa_c
    :=\min_{\substack{\delta\ne0\\
             \lVert\delta_{T_c^c}\rVert_1
             \le L_\lambda\lVert\delta_{T_c}\rVert_1}}
       \frac{\sqrt{s_c\En[(z_{c,i}'\delta)^2]}}
            {\lVert\delta_{T_c}\rVert_1},
    \label{eq:restricted-eigenvalue}
\end{equation}
and the upper and lower \(m\)-sparse eigenvalues
\begin{equation}
    \label{eq:sparse-eigenvalue}
    \phi_{\max,c}(m)
    :=\max_{\substack{\delta\ne0\\
                     \lVert\delta\rVert_0\le m}}
       \frac{\En[(z_{c,i}'\delta)^2]}{\lVert\delta\rVert_2^2},
    \qquad
    \phi_{\min,c}(m)
    :=\min_{\substack{\delta\ne0\\
                     1\le\lVert\delta\rVert_0\le m}}
       \frac{\En[(z_{c,i}'\delta)^2]}{\lVert\delta\rVert_2^2}.
\end{equation}
These conditions control the design matrix only on sparse sets and the cone in \eqref{eq:restricted-eigenvalue}, so they do not require the full dictionary to have linearly independent columns and, as such, do not rule out cases where \(p_c\gg n\).

\begin{assumption}[First-Stage Convergence]
\label{assm:primitive}
    There are constants \(0\le C_a<\infty\), \(\alpha>0\),
    \(0<\kappa_0\le1\),
    \(0<\underline\phi\le1\le\overline\phi<\infty\), and
    \(\vartheta>4L_\lambda^2\overline\phi/\kappa_0^2\), common to all candidates,
    such that the following hold for each \(c \in [J]\):
    \begin{enumerate}[(i)]
        \item The dictionary columns are normalized such that \(\En[z_{c,ij}^2]=1\) for each \(j\in[p_c]\).
        \item The approximation error satisfies
            \[
                \En[\xi_{c,i}^2]^{1/2}
                \le C_a\sqrt{\frac{s_c}{n}}.
            \]
        \item The penalty satisfies
            \[
                \lambda_c
                \ge c_\lambda\sqrt{2}\sigma_v
                \sqrt{
                   \frac{\log(2p_c)+\alpha\log\log(p_c\vee n)}{n}
                }.
            \]
        \item The design satisfies
        \begin{equation*}
            \kappa_c\ge\kappa_0,
            \qquad
            \phi_{\max,c}(\lceil\vartheta s_c\rceil\wedge p_c) \le\overline\phi,
            \qquad
            \phi_{\min,c}(\lceil(\vartheta+1)s_c\rceil\wedge p_c) \ge\underline\phi.
        \end{equation*}
        \item The sparsity index, \(s_c\), and penalty \(\lambda_c\) satisfy
        \begin{equation*}
            \frac{s_c}{\sqrt n}\to 0,
            \qquad
            \sqrt{s_c}\lambda_c\to 0.
        \end{equation*}
    \end{enumerate}
\end{assumption}

Versions of each condition in \Cref{assm:primitive} appear in \citet{BCCH-2012}, who also use them to establish convergence of first-stage LASSO estimates. Together with parts~(i)--(iv) of \Cref{assm:primitive}, \(\sqrt{s_c}\lambda_c\to0\) yields first-stage prediction consistency, while \(s_c/\sqrt n\to0\) controls terms that arise in the higher-order expansion. Combined with \Cref{assm:gaussian}, these conditions ensure that \eqref{eq:first-order} holds for every candidate. Beyond guaranteeing suitable first-stage convergence of the candidate first-stage estimators, these conditions do little to restrict the candidate list, so dictionaries and penalties may vary across candidates \citep{BickelRitovTsybakov-2009,VandeGeerBuhlmann-2009,BCCH-2012,ChetverikovSorensen-2025}.

The next condition ensures that the effective dimension grows and that its fluctuations become small relative to its mean. This allows us to estimate the complexity cost from the observed effective dimension, as described in \Cref{sec:feasible}, and to show that the terms omitted from the AMSE approximation are negligible relative to the candidate-specific component.

\begin{assumption}[Partial Beta-Min]
\label{assm:strong-core}
For each candidate, let
\[
    \underline b_{n,c}
    :=\frac{4}{\sqrt{\underline\phi}}
      \left\{
        \frac{(1+c_\lambda^{-1})\lambda_c\sqrt{s_c}}{\kappa_0}
        +C_a\sqrt{\frac{s_c}{n}}
      \right\},
\]
where \(\kappa_0\) and \(\underline\phi\) are as in
Assumption~\ref{assm:primitive}(iv). Define the set of  strong coordinates by
\(A_c:=\{j\in T_c:|\pi_{c,j}^0|\ge\underline b_{n,c}\}\) with
\(k_c:=|A_c|\).  There is a fixed constant \(C_s\ge1\) such that,
for every \(c\in[J]\),
\begin{equation}
    s_c\le C_sk_c,
    \qquad
    k_c\to\infty,
    \qquad
    \frac{\log(p_c/k_c)}{k_c}
    \to 0.
    \label{eq:strong-core-growth}
\end{equation}
\end{assumption}

The strong coordinates \(A_c\) are selected with probability approaching one. Unlike full beta-min conditions for support recovery \citep{ZhaoYu-2006,Wainwright-2009}, \Cref{assm:strong-core} does not impose a lower bound on the coefficients in \(T_c\setminus A_c\), nor does it require recovery of the entire support. The requirement that \(k_c\to\infty\) parallels the growing-basis condition in \citet{DonaldNewey-2001} and, together with the bound \(s_c\le C_sk_c\), underlies the relative concentration of \(d_c\) in \Cref{prop:ratio}. Since, under \Cref{assm:primitive}(v), \(\underline b_{n,c}\to0\), the coefficients in the strong core may shrink with \(n\).

The partial beta-min condition may be particularly well motivated under strong identification. If all first-stage coefficients are close to zero, the instruments may not induce enough exogenous variation in the endogenous variable to consistently estimate the structural parameter.

\subsection{The AMSE Criterion}
\label{sec:leading-result}

This subsection formally establishes the decomposition in \eqref{eq:decomp} and presents the infeasible criterion, a closed-form approximation to the expectation of the candidate-specific term. Exact expressions for the terms in \eqref{eq:decomp} are given in the appendix but are not needed for the characterization below.

Two population quantities enter the infeasible criterion. The first is a measure of first-stage strength,
\begin{equation}
    h_c:=\Ebar[\widehat\Pi_{c,i}\Pi_i],
    \label{eq:h-def}
\end{equation}
the population cross-moment between the fitted values and the optimal instrument. This is the population counterpart of the denominator of the IV estimator in \eqref{eq:iv-estimator}. Under the conditions of \Cref{sec:primitive}, \(h_c\) converges to the second moment of the optimal instrument, \(H\), and so is bounded away from zero in large samples.

The second quantity measures how well the fitted instrument approximates the optimal instrument. Because the IV estimator is unchanged when its instrument is multiplied by a nonzero constant, a fitted instrument proportional to the optimal instrument should register no error. We therefore measure the first-stage approximation error against the best rescaling of the optimal instrument,
\begin{equation}
    \mathcal A_c
    :=\min_{a\in\mathbb R}
      \Ebar\!\left[
       (\widehat\Pi_{c,i}-a\Pi_i)^2
     \right]
     =\Ebar[\widehat\Pi_{c,i}^2]-\frac{h_c^2}{H}.
    \label{eq:amse-approx}
\end{equation}
The minimized value \(\mathcal A_c\) is the mean squared residual from a population regression of the fitted instrument on \(\Pi\). The ratio \(\mathcal A_c/h_c^2\) appearing in the infeasible criterion below is invariant to rescaling of the fitted instrument.

The closed-form approximation to the expectation of the candidate-specific term combines these two quantities with certain error moments and moments of the effective dimension of the LASSO fit defined in \Cref{sec:effective-dimension},
\begin{equation}
    S_c
    :=\frac{\sigma_\varepsilon^2\mathcal A_c}{h_c^2}
     +\frac{\sigma_{\varepsilon v}^2\,\E[d_c^2+d_c]}
            {nh_c^2}.
    \label{eq:closed-form-criterion}
\end{equation}
We refer to \(S_c\) as the infeasible criterion. Each of its components is a population quantity, so the infeasible criterion cannot be computed directly from the data. \Cref{sec:feasible} constructs a consistent feasible counterpart used to select a candidate.
 
\begin{theorem}[Infeasible AMSE Criterion]
\label{thm:main}
Suppose Assumptions~\ref{assm:gaussian}, \ref{assm:primitive}, and
\ref{assm:strong-core} hold, with \(J<\infty\) fixed.  Then there exist random variables \(C_0\), 
    \(\{Q_c\}_{c\in[J]}\), and \(\{r_c\}_{c \in [J]}\) such that
\begin{enumerate}[(i)]
\item For every candidate,
\begin{equation}
    n(\widehat\beta_c-\beta)^2
    =C_0+Q_c+r_c,
    \label{eq:squared-decomposition}
\end{equation}
where \(C_0\) does not depend on the candidate,
\(\E[C_0]=\sigma_\varepsilon^2/H+O(n^{-1})\), and the remainder
satisfies \(\max_{c\in[J]}|r_c|/\E[Q_c]\longrightarrow_p0\).
\item For every candidate, \(\E[Q_c]>0\)
and \(\E[Q_c]=S_c\left(1+o(1)\right).\)
\end{enumerate}
\end{theorem}

\Cref{thm:main}(i) establishes the validity of the decomposition in \eqref{eq:decomp}, while \Cref{thm:main}(ii) shows that \(S_c\) approximates the candidate-specific AMSE component \(\bar Q_c=\E[Q_c]\) with vanishing relative error. Together, these results justify ranking candidates by \(S_c\): minimizing \(S_c\) is asymptotically equivalent to minimizing \(\bar Q_c\), and hence AMSE, over the candidate list. The remainder in part~(i) is controlled in probability, so this result does not establish a ranking by exact finite-sample MSE.

The first term measures how well the candidate LASSO approximates the optimal instrument, while the second depends on the first two moments of the effective dimension. This term originates from the second moment of \(\varepsilon'\widehat\Pi_c\), which we compute using results of \citet{TibshiraniTaylor-2012} and \citet{BellecZhang-2021}.

When the first stage is least squares on a prespecified set of \(K\) linearly independent instruments (\(\lambda = 0\)), the effective dimension always equals \(K\), and \(S_c\) reduces, up to lower-order terms, to the criterion of \citet{DonaldNewey-2001}. With LASSO, however, the effective dimension depends on the realization of the sample, so the complexity cost depends on both its mean and its variance.\footnote{Under the partial beta-min condition in \Cref{assm:strong-core}, the effective dimension grows, so that \(\E[d_c^2]\) asymptotically dominates \(\E[d_c]\). However, both terms appear in the expression for \(\E[(\varepsilon'\widehat\Pi_c)^2]\), and we keep them since they may both be relevant in finite samples.}

A richer dictionary or a lighter penalty may reduce \(\mathcal A_c\) while increasing the effective dimension. To examine this tradeoff, write \(\rho=\operatorname{Corr}(\varepsilon_i,v_i)\), so that \(\sigma_{\varepsilon v}=\rho\sigma_\varepsilon\sigma_v\), and rewrite \eqref{eq:closed-form-criterion} as
\begin{equation}
    S_c
    =\frac{\sigma_\varepsilon^2}{h_c^2}
     \left\{
       \mathcal A_c
       +\rho^2\sigma_v^2\,
        \frac{\E[d_c^2+d_c]}{n}
     \right\}.
    \label{eq:criterion-exchange-rate}
\end{equation}
In this form the second term is weighted by \(\rho^2\sigma_v^2\), which is proportional to the squared covariance between the structural and first-stage errors. When the regressor is close to exogenous, the complexity term receives little weight.\footnote{Exact exogeneity is excluded by Assumption~\ref{assm:gaussian}(ii), and the case \(\sigma_{\varepsilon v}=0\) requires a different expansion.} When the regressor is highly endogenous, the same complexity produces a larger bias. This is the LASSO analogue of the many-instrument bias examined in \citet{Bekker-1994} and \citet{DonaldNewey-2001}.

Cross-validation targets first-stage prediction error rather than the structural AMSE criterion in \eqref{eq:closed-form-criterion}. Its rankings may therefore differ even when the complexity term receives little weight, because \(\mathcal A_c\) allows rescaling and the criterion also depends on \(h_c\). As endogeneity strengthens, the AMSE criterion prefers models with lower complexity, while cross-validation does not take this into consideration. The plug-in penalty rules of \citet{BCCH-2012} and \citet{ChetverikovSorensen-2025} are designed to control first-stage noise as in \Cref{assm:primitive}(iii) and likewise do not account for endogeneity.

\subsection{Feasible Candidate Selection}
\label{sec:feasible}

Since \(S_c\) depends on population quantities and is thus infeasible, we construct an observable criterion that estimates it up to a candidate-independent centering term and therefore preserves its ranking asymptotically. The only additional inputs not obtained from the candidate fits are estimates of the structural error moments \(\sigma_\varepsilon^2\) and \(\sigma_{\varepsilon v}\). We obtain these estimates using residuals from a pilot estimate \(\check\beta\) of \(\beta\).

For each candidate, the observed IV denominator
\begin{equation}
    \label{eq:sample-denominator}
    \widehat h_c:=\En[\widehat\Pi_{c,i}x_i]
\end{equation}
will serve as an estimator for the first-stage cross-moment \(h_c\).  The KKT conditions for the LASSO optimization problem in \eqref{eq:lasso} yield \(\widehat h_c =\En[\widehat\Pi_{c,i}^2]+\lambda_c\lVert\widehat\pi_c\rVert_1\ge0\), so \(\widehat h_c=0\) if and only if \(\widehat\pi_c = 0\).  The pilot \(\check\beta\) can be any estimator with \(\check\beta-\beta=O_p(n^{-1/2})\), such as the IV-LASSO estimate \(\widehat\beta_{\check c}\) from a possibly data-dependent candidate \(\check c\in[J]\).  Given the pilot, let
\begin{equation}
    \label{eq:nuisance-estimators}
    \widehat\sigma_\varepsilon^2
       :=\En[(y_i-\check\beta x_i)^2]
    \qquad \text{and} \qquad
    \widehat\sigma_{\varepsilon v}
       :=\En[(y_i-\check\beta x_i)x_i].
\end{equation}
We put together these estimates to construct the feasible criterion
\begin{equation}
    \label{eq:feasible-score}
    \widehat S_c
    :=\frac{
         \widehat\sigma_\varepsilon^2
         \En[\widehat\Pi_{c,i}^2]
         +n^{-1}
         \widehat\sigma_{\varepsilon v}^{\,2}
         (d_c^2+d_c)
       }{
         \widehat h_c^2
       },
\end{equation}
with the convention \(\widehat S_c:=+\infty\) when \(\widehat h_c=0\). The selected
candidate minimizes the feasible criterion,
\begin{equation*}
    \widehat c
    :=
       \min\left(
          \argmin_{c\in[J]}
             \widehat S_c
       \right),
\end{equation*}
with ties broken by the smallest index.

The feasible criterion \(\widehat S_c\) replaces the moment \(\E[d_c^2+d_c]\) in the infeasible \(S_c\) with the realized value \(d_c^2+d_c\).  This replacement is justified by
the following proposition.
\begin{prop}[Ratio Consistency of the Effective Dimension]
\label{prop:ratio}
Suppose Assumptions~\ref{assm:gaussian}, \ref{assm:primitive}, and
\ref{assm:strong-core} hold, with \(J<\infty\) fixed.  Then, for every
candidate,
\begin{align}
    \frac{d_c}{\E[d_c]}
       &\longrightarrow_p1,
       \label{eq:rank-ratio-main}\\
    \frac{d_c^2+d_c}{\E[d_c^2+d_c]}
       &\longrightarrow_p1.
       \label{eq:rank-factor-ratio-main}
\end{align}
\end{prop}
The growth conditions on the strong core in \Cref{assm:strong-core}, together with a variance bound derived from the second-order Stein identities of \citet{BellecZhang-2021}, yield relative concentration of \(d_c\) and hence of \(d_c^2+d_c\). This does not require the support of the LASSO solution to contain the support \(T_c\) in \eqref{eq:sparse-approximation}.

The numerator in \(\widehat S_c\) contains the fitted second moment \(\E_n[\widehat \Pi_{c,i}^2]\) rather than a direct estimate of the approximation error \(\calA_c\). Applying the decomposition in \eqref{eq:amse-approx} to sample moments, with \(\widehat{\calA}_c:=\min_{a\in\mathbb R}\En[(\widehat\Pi_{c,i}-a\Pi_i)^2]\), splits the first term of \(\widehat S_c\) into two parts,
\[
    \frac{\widehat\sigma_\varepsilon^2\En[\widehat\Pi_{c,i}^2]}{\widehat h_c^2}
    =\frac{\widehat\sigma_\varepsilon^2\widehat{\calA}_c}{\widehat h_c^2}
     +\frac{\widehat\sigma_\varepsilon^2}{H}
      \left(\frac{\En[\widehat\Pi_{c,i}\Pi_i]}{\widehat h_c}\right)^2.
\]
The first part estimates \(\sigma_\varepsilon^2\calA_c/h_c^2\), the first term of \(S_c\). Because each candidate first stage is consistent, we can expand the second part around its value for the optimal instrument itself, \(\widehat\sigma_\varepsilon^2H/\En[\Pi_ix_i]^2\), where \(\En[\Pi_ix_i]=H+\En[\Pi_iv_i]\) and \(\En[\Pi_iv_i]\) is of order \(n^{-1/2}\). This gives
\begin{equation}
    \label{eq:common-term}
    B_n \coloneqq\widehat\sigma_\varepsilon^2
       \left(\frac1H-\frac{2\En[\Pi_iv_i]}{H^2}\right)
\end{equation}
up to terms that are negligible relative to \(\bar Q_c\). \Cref{thm:feasible} below shows that
\[
    \widehat S_c = B_n + \bar Q_c(1 + o_p(1)).
\]
Although \(B_n\) may dominate the candidate-specific AMSE components \(\bar Q_c\), it is common to all candidates and therefore does not affect their feasible-criterion ranking. Minimizing \(\widehat S_c\) is then asymptotically equivalent to minimizing \(\bar Q_c\). This does not require \(B_n\) to be observed.

\begin{theorem}[Feasible AMSE Selection]
\label{thm:feasible}
Suppose Assumptions~\ref{assm:gaussian}, \ref{assm:primitive}, and
\ref{assm:strong-core} hold, with \(J<\infty\) fixed, and that the pilot
satisfies \(\check\beta-\beta=O_p(n^{-1/2})\).
\begin{enumerate}[(i)]
\item Up to the common centering \(B_n\), the feasible criterion
consistently estimates the infeasible criterion,
\begin{equation}
    \max_{c\in[J]}
    \frac{|\widehat S_c-B_n-S_c|}{S_c}
    \longrightarrow_p0,
    \label{eq:feasible-consistency}
\end{equation}
and thus the expectation of the candidate-specific term,
\begin{equation}
    \max_{c\in[J]}
    \frac{|\widehat S_c-B_n-\bar Q_c|}{\bar Q_c}
    \longrightarrow_p0.
    \label{eq:feasible-expansion}
\end{equation}
\item The selected candidate \(\widehat c\) satisfies
\begin{equation}
    \frac{\bar Q_{\widehat c}}
         {\min_{c\in[J]}\bar Q_c}
    \longrightarrow_p1.
    \label{eq:oracle-equivalence}
\end{equation}
\end{enumerate}
\end{theorem}

The supplementary appendix establishes the same AMSE criterion and feasible selection result for a class of non-Gaussian errors. That extension requires the structural error to have a linear conditional mean and constant conditional variance given the first-stage error, along with additional regularity conditions. Under more general departures from Gaussianity or homoskedasticity, AMSE optimality is not established here, although the results in \Cref{sec:application-laplace} suggest that the AMSE in more general models may still be well approximated by the AMSE in the Gaussian model. If, under alternative conditions, all candidates continue to admit the common first-order expansion in \eqref{eq:first-order}, then the selected estimator retains that expansion and first-order inference remains valid.
{}
\section{Simulation Study}
\label{sec:empirical-application}

We compare the proposed criterion to first-stage selection via cross-validation and the plug-in penalty of \citet{BCCH-2012}. We also consider a cross-fitted IV estimator, with the dictionary and penalty chosen within each training sample \citep{CCDDHNR-2018}. The simulations are calibrated to the data of \citet{PRSW-2015}, who study the effect of credit supply on exports. We first consider Gaussian errors in \Cref{sec:application-simulation} and then repeat the exercise in \Cref{sec:application-laplace} with common-scale Laplace and smooth-transformation errors.

\subsection{Calibrated Gaussian Designs}
\label{sec:application-simulation}

We hold the observed exposure instruments fixed and generate the outcome
and endogenous variable from the IV model described in \eqref{eq:model} with \(\beta=1\).  The true first stage \(\Pi\) is proportional to the fitted values from a least-squares regression of credit growth on 17 instruments selected by a cross-validated LASSO from the original instruments and their 45 pairwise products.  We also consider a dense alternative to this main first stage, proportional to the least-squares fitted values on all 55 of these instruments.  We rescale these fitted values to vary identification strength, targeting expected first-stage \(F\) statistics of \(2.7\), \(4.4\), \(7.7\), \(14.5\), or \(28.0\) at the original sample size of \(n=14{,}208\) firm--product--destination export flows. Holding the instruments fixed and only rescaling the two first stages lets us examine how each rule's performance changes with identification strength and endogeneity in an empirically plausible setting.

The errors are generated homoskedastic Gaussian, with
\(\varepsilon_i=\rho v_i+\sqrt{1-\rho^2}\,e_i\) for \(v_i, e_i \overset{\text{iid}}{\sim} N(0,1)\).  The parameter \(\rho\) controls the degree of endogeneity.  We consider a value calibrated to the data, \(\widehat\rho\approx-0.22\), along with its multiples \(\rho=c\,\widehat\rho\) for \(c=0,0.25,\ldots,2.75\) to examine how the performance of each rule changes with the level of endogeneity.\footnote{Assumption~\ref{assm:gaussian}(ii) rules out exact exogeneity, so the designs with \(\rho=0\) lie outside our theory, and we include them only as a limiting case.} Each of the 120 designs is simulated 1,000 times.  We report the average of \(n(\widehat\beta-\beta)^2\), which we refer to as either the risk or the mean squared error.

Each candidate first stage implementation uses one of two nested dictionaries: either the ten original exposure instruments (Original) or a richer dictionary of 114 instruments including their pairwise products and the authors' alternative exposure measures (Expanded).  For each dictionary, we consider 37 multiples of the plug-in penalty of \citet{BCCH-2012} (``BCCH''), giving 74 candidates in total.  This plug-in penalty level uses the known standard deviation of the first-stage error in place of an estimate.  The proposed criterion selects among these candidate dictionary-penalty pairs by minimizing \(\widehat S_c\) in \eqref{eq:feasible-score}, with estimated error moments coming from a pilot estimator of \(\beta\). Cross-validation chooses from the same candidates using ten-fold first-stage prediction error, keeping observations from the same firm in the same fold. The cross-fitting estimator (``Cross-fit'') makes this choice separately for each of ten folds formed by firm, using only the other nine folds, and the resulting predictions for each held-out fold are used as instruments in a single IV estimate. BCCH uses its plug-in penalty on the Original dictionary.  When it selects no instruments, we use the pilot estimate instead.

\begin{figure}[!htp]
\centering
\includegraphics[width=\textwidth]{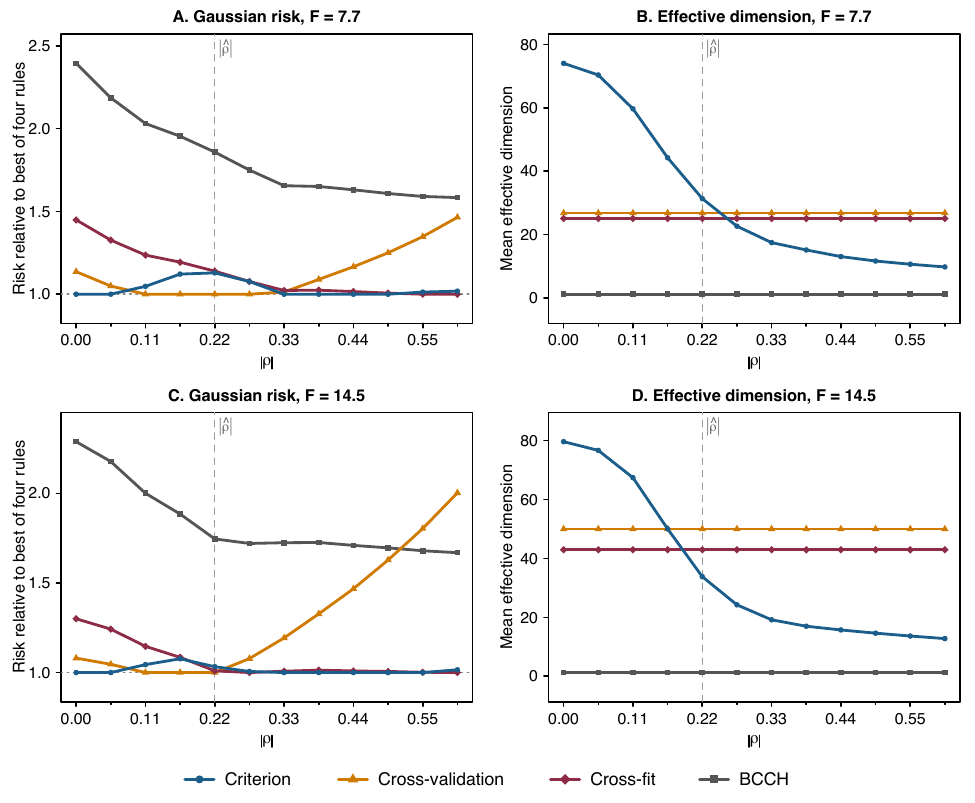}
\caption{Relative risk and effective dimension under Gaussian errors.}
\label{fig:prsw-gaussian-summary}
\floatnotes{Results use the main first stage, \(n=14{,}208\), and 1,000
replications.  Panels A and C report risk relative to the lowest risk among
the four rules; Panels B and D report mean effective dimension, averaged
over folds for cross-fitting.  Panels A--B use \(F=7.7\) and Panels C--D use
\(F=14.5\).  The horizontal axis is the absolute error correlation
\(|\rho|\), and the dashed line marks the calibrated correlation
\(|\widehat\rho|=0.218\).}
\end{figure}

\Cref{fig:prsw-gaussian-summary} shows how the rules respond to endogeneity. Consistent with the increasing complexity cost seen in \eqref{eq:criterion-exchange-rate}, the criterion selects less complex first stages as endogeneity rises, while the other three rules do not account for the level of endogeneity. At both \(F=7.7\) and \(F=14.5\), the criterion's mean effective dimension falls by more than six-fold. Its risk is more than 50\% below that of BCCH at the lowest levels of endogeneity and between 30\% and 50\% below that of cross-validation at the highest level, although cross-validation performs better in some designs with intermediate levels of endogeneity.

\begin{table}[!htp]
\centering
\begin{threeparttable}
\caption{Relative risk, Gaussian errors.}
\label{tab:prsw-compact}
\setlength{\tabcolsep}{4pt}
\begin{tabular}{ll *{3}{S[table-format=1.3]} S[table-format=2.3] S[table-format=1.3]}
\toprule
& & \multicolumn{4}{c}{Feasible rules} & \multicolumn{1}{c}{Infeasible} \\
\cmidrule(lr){3-6}\cmidrule(l){7-7}
Design & \shortstack[l]{Summary\\Measure} & \multicolumn{1}{c}{Criterion} & \multicolumn{1}{c}{\shortstack{Cross-\\validation}} & \multicolumn{1}{c}{Cross-fit} & \multicolumn{1}{c}{BCCH} & \multicolumn{1}{c}{\shortstack{Known\\Moments}} \\
\midrule
\multicolumn{7}{l}{\textit{Panel A: Main signal, by first-stage strength}} \\
\addlinespace
$F=2.7$ & Average & 1.072 & 1.034 & 1.589 & 14.529 & 0.857 \\
& Maximum & 1.199 & 1.143 & 2.045 & 21.117 & 0.976 \\
\addlinespace
$F=4.4$ & Average & 1.051 & 1.074 & 1.240 & 1.813 & 0.924 \\
& Maximum & 1.173 & 1.291 & 1.615 & 2.459 & 0.987 \\
\addlinespace
$F=7.7$ & Average & 1.034 & 1.126 & 1.124 & 1.824 & 0.962 \\
& Maximum & 1.129 & 1.464 & 1.449 & 2.394 & 1.023 \\
\addlinespace
$F=14.5$ & Average & 1.015 & 1.302 & 1.068 & 1.835 & 0.975 \\
& Maximum & 1.077 & 2.002 & 1.300 & 2.288 & 1.009 \\
\addlinespace
$F=28.0$ & Average & 1.036 & 1.301 & 1.038 & 1.908 & 1.014 \\
& Maximum & 1.074 & 2.009 & 1.163 & 2.170 & 1.044 \\
\midrule
\multicolumn{7}{l}{\textit{Panel B: Main signal, by degree of endogeneity}} \\
\addlinespace
$0\le|\rho|\le 0.109$ & Average & 1.009 & 1.048 & 1.422 & 3.490 & 0.865 \\
& Maximum & 1.047 & 1.143 & 2.045 & 9.883 & 1.000 \\
\addlinespace
$0.164\le|\rho|\le 0.273$ & Average & 1.105 & 1.014 & 1.229 & 4.105 & 0.967 \\
& Maximum & 1.199 & 1.107 & 1.725 & 14.566 & 1.007 \\
\addlinespace
$0.327\le|\rho|<0.6$ & Average & 1.026 & 1.304 & 1.098 & 4.966 & 0.976 \\
& Maximum & 1.183 & 2.009 & 1.534 & 21.117 & 1.044 \\
\midrule
\multicolumn{7}{l}{\textit{Panel C: Dense signal, all designs}} \\
\addlinespace
All $F$ and $\rho$ & Average & 1.035 & 1.145 & 1.374 & 4.339 & 0.926 \\
& Maximum & 1.175 & 2.017 & 2.474 & 23.204 & 1.002 \\
\bottomrule
\end{tabular}
{}
\begin{tabnotes}
Risk is the mean of \(n(\widehat\beta-\beta)^2\) over 1,000 replications.
All panels present risk relative to the best of the four feasible rules in each design.  Panels A and B use the main first stage.  Panel A averages or maximizes over the twelve values of \(\rho\) at each \(F\).  Panel B averages or maximizes over the five values of \(F\) and the values of \(|\rho|\) in each range.  Panel C averages or maximizes over all 60 designs with the dense first stage.  Known Moments evaluates the feasible criterion at population error moments and is not one of the four feasible rules, so its relative risk can fall below one.
\end{tabnotes}
\end{threeparttable}
\end{table}

\Cref{tab:prsw-compact} shows that the criterion remains close to the best feasible rule across the design grid, but does not dominate every design. With the main first stage, cross-validation has lower average relative risk than the criterion at \(F=2.7\) and near the calibrated correlation. As identification strengthens, the criterion and cross-fitting perform more similarly. The criterion's largest relative risk is about 20\% above that of the best rule. The large BCCH entries are driven by a few replications at \(F=2.7\) in which BCCH selects no instruments and the pilot estimate is used in its place.

Setting the error moments to their population values, \(\sigma_\varepsilon^2=1\) and \(\sigma_{\varepsilon v}=\rho\), lowers the criterion's average relative risk for the main first stage from 1.042 to 0.946. The Known Moments benchmark still uses estimated first-stage quantities to rank candidates and need not attain the lowest risk in every design. We do not include it among the four rules used to normalize risk.

\subsection{Robustness to Non-Gaussian Errors}
\label{sec:application-laplace}

We repeat the simulations of \Cref{sec:application-simulation} under two departures from Gaussianity that differ in whether they preserve the constant conditional variance of the Gaussian design. The simulation details are otherwise left completely unchanged.

For the common-scale Laplace design, we draw mutually independent \(W_i\sim\operatorname{Exp}(1)\) and \(a_i,b_i\sim N(0,1)\), independently across observations, and set
\[
  v_i=\sqrt{W_i}\,a_i,\qquad
  \varepsilon_i=\sqrt{W_i}\left(\rho a_i+\sqrt{1-\rho^2}\,b_i\right).
\]
The scale is shared by the two errors within each observation, and both errors have variance-one Laplace marginals.

For the smooth-transformation design, we normalize the following transformation to unit variance. Its unnormalized form appears in \citet[Figure~1]{ChetverikovLiaoChernozhukov-2020-extended-main}:
\[
  \psi(t)=\frac{t+t^3}{\nu(1+2t^2)},\qquad
  \nu^2=\E\!\left[\left(\frac{G+G^3}{1+2G^2}\right)^2\right],
  \quad G\sim N(0,1).
\]
For independent standard Gaussian draws \(G_i,B_i\), independent across observations, we set
\[
  v_i=\psi(G_i),\qquad
  \varepsilon_i=\rho v_i+\sqrt{1-\rho^2}\,\psi(B_i).
\]
Both designs have mean-zero, unit-variance errors with correlation \(\rho\). The smooth-transformation design also satisfies \(\operatorname{Var}(\varepsilon_i\mid v_i)=1-\rho^2\), so it is covered by the non-Gaussian extension in the supplementary appendix. The common-scale Laplace design, in which this conditional variance depends on \(v_i\), goes beyond the assumptions of that extension. As before, the designs with \(\rho=0\) serve as exogenous benchmarks.

\Cref{tab:prsw-error-summary} reports the criterion's relative risk under each error distribution, with Gaussian errors as the reference row. We use the same 60 designs for each first-stage signal and 1,000 replications per design. Average relative risks remain similar across error distributions, although the maximum for the dense signal reaches 1.221 under smooth-transformation errors. Under these errors, the criterion's unnormalized risk differs from its risk under Gaussian errors by only between \(-4.1\%\) and \(3.1\%\), and its mean effective dimension changes by at most \(3.2\%\).

\begin{table}[!htp]
\centering
\begin{threeparttable}
\caption{Relative risk of the proposed criterion across error distributions.}
\label{tab:prsw-error-summary}
\setlength{\tabcolsep}{7pt}
\begin{tabular}{l*{4}{S[table-format=1.3]}}
\toprule
& \multicolumn{2}{c}{Main signal} & \multicolumn{2}{c}{Dense signal} \\
\cmidrule(lr){2-3}\cmidrule(l){4-5}
Error distribution & \multicolumn{1}{c}{Average} & \multicolumn{1}{c}{Maximum} & \multicolumn{1}{c}{Average} & \multicolumn{1}{c}{Maximum} \\
\midrule
Gaussian & 1.042 & 1.199 & 1.035 & 1.175 \\
Common-scale Laplace & 1.041 & 1.174 & 1.035 & 1.193 \\
Smooth transformation & 1.045 & 1.188 & 1.041 & 1.221 \\
\bottomrule
\end{tabular}
\begin{tabnotes}
Each signal comprises 60 designs (five values of \(F\) and twelve values of \(\rho\)) with 1,000 replications per design. Relative risk divides the criterion's risk by the lowest risk among the four feasible rules in the same design and under the same error distribution. Known Moments is not part of this comparison. The Gaussian design is described in \Cref{sec:application-simulation}, and the two non-Gaussian constructions are defined in \Cref{sec:application-laplace}.
\end{tabnotes}
\end{threeparttable}
\end{table}
For the main signal at \(F=7.7\) under Laplace errors, the mean effective dimension of the first stage selected by the criterion still falls from 74.3 to 9.7 as \(|\rho|\) increases. At the highest level of endogeneity, the criterion's risk is about \(28\%\) below that of cross-validation but \(5\%\) above that of cross-fitting. Cross-validation retains an advantage in some designs with intermediate levels of endogeneity, and cross-fitting has lower average relative risk at \(F=28.0\). The largest BCCH risks under Laplace errors are again driven by a few replications in which BCCH selects no instruments. While AMSE optimality under the smooth-transformation design is established in the supplementary appendix, the results under Laplace errors suggest that the AMSE in more general models may be well approximated by the AMSE in the Gaussian model.
{}
\startmainbibliography

\finishmainbibliography

\appendix
\counterwithin{lemma}{section}
\counterwithin{prop}{section}
\counterwithin{example}{section}
\counterwithin{corollary}{section}
\counterwithin{remark}{section}

\section{Notation and Preliminary Reductions}
\label{app:notation}

Here we collect the notation used throughout the appendix. We write
\(C\) for a positive constant that does not depend on \(n\) or on the
candidate and whose value may change from line to line. For positive
\(a_n\) and \(b_n\), we write \(a_n\asymp b_n\) when \(a_n/b_n\) is
bounded and bounded away from zero. Since \(J\)
is fixed, a rate or limit established for each candidate holds
uniformly over \(c\in[J]\). We therefore fix a candidate and suppress
its subscript, writing, for example, \(\widehat\Pi\), \(h\), \(d\), and
\(S\) for \(\widehat\Pi_c\), \(h_c\), \(d_c\), and \(S_c\). The
subscript reappears only where candidates are compared.

\paragraph{LASSO notation.}
To use the Stein identities, we need expressions for the derivatives of
the LASSO fit with respect to its response.
For \(w\in\SR^n\), let \(\widehat\pi(w)\) be a minimizer of
\eqref{eq:lasso} with \(x\) replaced by \(w\), and write
\[
    \mu(w):=Z\widehat\pi(w),
    \qquad
    u(w):=w-\mu(w),
    \qquad
    E(w):=\left\{j\in[p]:|\En[z_{ij}u_i(w)]|=\lambda\right\}.
\]
Thus \(\mu(x)=\widehat\Pi\), \(u(x)=u\), and \(E(x)=E\).

Lemma~\ref{lem:geometry}, below, recalls certain facts about the LASSO fit from \citet{TibshiraniTaylor-2012}. In particular, it notes that the LASSO fitted-value map
\(\mu:\SR^n\to\SR^n\) is one-Lipschitz, and hence
Lebesgue-almost-everywhere differentiable by Rademacher's theorem.  At a
differentiability point \(w\in\SR^n\), write \(D\mu(w)\) for the Jacobian of \(\mu\).
Since \(\Pi+v\) has a Lebesgue density under
Assumption~\ref{assm:gaussian}(i), the fit is almost surely
differentiable at the observed \(x=\Pi+v\), and we write
\(P:=D\mu(\Pi+v)\) for the Jacobian there. Then
\Cref{lem:geometry} and \eqref{eq:geometry-properties} give, almost
surely,
\begin{equation}
    P=P'=P^2,
    \qquad
    \operatorname{tr}(P)=d,
    \qquad
    P\widehat\Pi=\widehat\Pi.
    \label{eq:projection-properties}
\end{equation}
That is, \(P\) is the orthogonal projection onto the column space of
\(Z_{E(x)}\), the active columns indexed by the equicorrelation set
in \eqref{eq:equicorrelation}, so that both \(P\) and \(I-P\) are
contractions.

With \(d\) defined in \eqref{eq:realized-rank}, we will save notation later in the proofs and write
\begin{equation}
    \bar d:=\E[d],
    \qquad\text{so that}\qquad
    \E[d^2+d]=\bar d^2+\Var(d)+\bar d.
    \label{eq:rank-moments}
\end{equation}

\paragraph{Population projection residual.}
The minimum in \eqref{eq:amse-approx} is attained at \(a=h/H\). We
write
\begin{equation}
    M:=\widehat\Pi-\frac{h}{H}\Pi
    \label{eq:projection-residual}
\end{equation}
for the residual from the population projection of \(\widehat\Pi\)
on \(\Pi\), so that
\(\E\lVert M\rVert_2^2=n\mathcal A\). This is \(m(v)\) in
\Cref{lem:offset-fit-moments} with \(b=(h/H)\Pi\), and the identity
\eqref{eq:offset-stein} becomes
\begin{equation}
    \E[(v'M-\sigma_v^2d)^2]=\sigma_v^2n\mathcal A+\sigma_v^4\bar d.
    \label{eq:residual-stein}
\end{equation}

\paragraph{Gaussian error decomposition.}
Under Assumption~\ref{assm:gaussian}(i) set
\(e:=\varepsilon-\frac{\sigma_{\varepsilon v}}{\sigma_v^2}v\).  Then
\(e\) is Gaussian with mean zero and \(\Cov(e_i,v_i)=0\), hence
independent of \(v\):
\begin{equation}
    \varepsilon
    =\frac{\sigma_{\varepsilon v}}{\sigma_v^2}v+e,
    \qquad
    e\perp v,
    \qquad
    \Var(e_i)
    =\sigma_\varepsilon^2
      -\frac{\sigma_{\varepsilon v}^2}{\sigma_v^2}.
    \label{eq:comparison-error-decomposition}
\end{equation}

\paragraph{Positive denominator event.}
\Cref{lem:primitive-consequences} gives a fixed constant
\(\underline h>0\) with \(h\ge\underline h\) for all sufficiently large
\(n\). We take \(n\) this large throughout, so that \(H\), \(h\), and
their reciprocals are bounded, and enlarge \(C\) where needed so that
the bounds below hold for all such \(n\). Every function of \(v\) to which we apply the Stein identities
or the Poincar\'e inequality is locally Lipschitz, with value and
gradient of at most polynomial growth by \Cref{lem:geometry}, so these
results apply and all expectations below are finite. With
\(\widehat h\) the sample denominator in \eqref{eq:sample-denominator},
define
\begin{equation}
    \mathcal H_n
    :=\left\{\widehat h\ge\frac{\underline h}{2}\right\}.
    \label{eq:denominator-event}
\end{equation}
Since analyzing the properties of IV-LASSO is difficult when the denominator is close to zero, our proofs will often work on \(\calH_n\) to establish stochastic bounds. Because \(\P(\mathcal H_n)\to1\) by \Cref{lem:denominator}, these bounds also hold without restricting to \(\mathcal H_n\), while all expectations below are unconditional.
{}
\section{Proof of Theorem~\ref{thm:main}}
\label{sec:proof-thm-main}

\begin{proof}
Begin by defining the numerator, scaled by the population denominator,
and the denominator correction,
\begin{equation}
    \Gamma:=\frac{\sqrt n\,\En[\varepsilon_i\widehat\Pi_i]}{h},
    \qquad
    \Delta:=\frac{\widehat h-h}{h}.
    \label{eq:leading-term}
\end{equation}
On \(\mathcal H_n\), \Cref{lem:denominator} shows that
\(\widehat\beta\) is the ratio in \eqref{eq:iv-estimator}, and after
substituting \(y_i=\beta x_i+\varepsilon_i\) we can write
\begin{equation}
    \sqrt n(\widehat\beta-\beta)
    =\frac{\sqrt n\,\En[\varepsilon_i\widehat\Pi_i]}{\widehat h}
    =\frac{\Gamma}{1+\Delta}
    =\Gamma(1-\Delta)+\frac{\Gamma\Delta^2}{1+\Delta}.
    \label{eq:ratio-expansion}
\end{equation}
In the term \(\Gamma\Delta\), we replace \(\Gamma\) by the leading
term in \eqref{eq:first-order},
\begin{equation}
    L_n:=\frac{\sqrt n\,\En[\varepsilon_i\Pi_i]}{H},
    \label{eq:optimal-score}
\end{equation}
and, with \(q:=\En[\Pi_iv_i]/H\), define
\begin{equation}
\begin{gathered}
    \Lambda:=\Gamma-\Delta L_n,
    \qquad
    C_0:=(1-q)^2L_n^2,\\
    Q:=\Lambda^2-C_0,
    \qquad
    r:=n(\widehat\beta-\beta)^2-\Lambda^2.
\end{gathered}
    \label{eq:decomposition-construction}
\end{equation}
These definitions make \eqref{eq:squared-decomposition} an identity,
and \(C_0\) does not depend on the candidate. When the optimal
instrument \(\Pi\) is used in place of \(\widehat\Pi\), \(h\)
becomes \(H\) and \(\widehat h\) becomes \(\En[\Pi_ix_i]=H(1+q)\), so
that \(\Gamma=L_n\), \(\Delta=q\), \(\Lambda=(1-q)L_n\), and
\(Q=0\).

\emph{Step 1: \(\sqrt n(\widehat\beta-\beta)=\Lambda+o_p(S)\).}
Subtracting \(\Lambda\) from \eqref{eq:ratio-expansion}, we have
\[
    \sqrt n(\widehat\beta-\beta)-\Lambda
    =-\Delta(\Gamma-L_n)+\frac{\Gamma\Delta^2}{1+\Delta}
\]
on \(\mathcal H_n\), where \(1+\Delta=\widehat h/h\) is bounded away
from zero. \Cref{lem:denominator,lem:score-distance} give
\(\Delta=O_p(n^{-1/2})\) and \(\Gamma-L_n=O_p(\sqrt S)\), so that
\(\Gamma=O_p(1)\). Therefore
\[
    \sqrt n(\widehat\beta-\beta)-\Lambda
    =O_p\!\left(\sqrt{\frac{S}{n}}+\frac1n\right)
    =S\cdot O_p\!\left(\frac1{\sqrt{nS}}+\frac1{nS}\right),
\]
which is \(o_p(S)\) because \(nS\to\infty\) by
\eqref{eq:score-rates}. The same bounds show that
\begin{equation}
    \Lambda=O_p(1),
    \qquad
    \sqrt n(\widehat\beta-\beta)=O_p(1).
    \label{eq:estimator-order}
\end{equation}

\emph{Step 2: \(\E[\Gamma^2]=\sigma_\varepsilon^2/H+S\).}
The second-moment identity \eqref{eq:structural-second} in
\Cref{lem:gaussian-moments} gives
\[
    \E[\Gamma^2]
    =\frac{\E[(\varepsilon'\widehat\Pi)^2]}{nh^2}
    =\frac{\sigma_\varepsilon^2\E\lVert\widehat\Pi\rVert_2^2}{nh^2}
     +\frac{\sigma_{\varepsilon v}^2\E[d^2+d]}{nh^2}.
\]
By \eqref{eq:amse-approx},
\(\E\lVert\widehat\Pi\rVert_2^2/n=\mathcal A+h^2/H\), and
substituting this into the first term yields
\[
    \E[\Gamma^2]
    =\frac{\sigma_\varepsilon^2}{H}
     +\underbrace{\frac{\sigma_\varepsilon^2\mathcal A}{h^2}
       +\frac{\sigma_{\varepsilon v}^2\E[d^2+d]}{nh^2}}_{S}.
\]

\emph{Step 3: \(\E[Q]=S\{1+o(1)\}\).}
Expanding \(\Lambda^2=(\Gamma-\Delta L_n)^2\) and writing
\(\Gamma=L_n+(\Gamma-L_n)\) in the cross term,
\[
    \E[\Lambda^2]-\E[\Gamma^2]
    =-2\E[\Delta L_n^2]
     -2\E[\Delta L_n(\Gamma-L_n)]
     +\E[\Delta^2L_n^2].
\]
By \Cref{lem:denominator-correction},
\(|\E[\Delta L_n^2]|\le C\{\sqrt{S/n}+1/n\}\) and
\(\E[\Delta^2L_n^2]\le C/n\). For the middle term, Cauchy--Schwarz
and \Cref{lem:denominator-correction,lem:score-distance} give
\[
    |\E[\Delta L_n(\Gamma-L_n)]|
    \le\{\E[\Delta^2L_n^2]\,\E[(\Gamma-L_n)^2]\}^{1/2}
    \le C\sqrt{\frac{S}{n}}.
\]
We next compute \(\E[C_0]\). Under Assumption~\ref{assm:gaussian},
\((\Pi'\varepsilon,\Pi'v)\) is centered jointly Gaussian with
covariance matrix \(\lVert\Pi\rVert_2^2\Omega=nH\Omega\). Its odd moments
vanish, and the product-moment formula of \citet{Isserlis-1918} states
that \(\E[X^2Y^2]=\E[X^2]\E[Y^2]+2\{\E[XY]\}^2\) for a centered jointly
Gaussian pair \((X,Y)\). We therefore have
\[
    \E[L_n^2]=\frac{\sigma_\varepsilon^2}{H},
    \qquad
    \E[qL_n^2]=0,
    \qquad
    \E[q^2L_n^2]
    =\frac{\sigma_\varepsilon^2\sigma_v^2+2\sigma_{\varepsilon v}^2}{nH^2}.
\]
Expanding \(C_0=(1-2q+q^2)L_n^2\), we obtain
\(\E[C_0]=\sigma_\varepsilon^2/H+O(n^{-1})\), which is the bound on
\(\E[C_0]\) in part~(i). Combining these bounds with Step~2 and
\(\E[Q]=\E[\Lambda^2]-\E[C_0]\),
\[
    |\E[Q]-S|
    =\left|\E[\Lambda^2]-\E[\Gamma^2]
      +\frac{\sigma_\varepsilon^2}{H}-\E[C_0]\right|
    \le C\left\{\sqrt{\frac{S}{n}}+\frac1n\right\}.
\]
Dividing by \(S\) and using \(nS\to\infty\) once more, we have
\(\E[Q]/S\to1\), so that \(\E[Q]>0\) and part~(ii) follows.

\emph{Step 4: \(r/\E[Q]\to_p0\).}
Factoring the difference of squares in
\eqref{eq:decomposition-construction},
\[
    |r|
    =|\sqrt n(\widehat\beta-\beta)-\Lambda|\,
     |\sqrt n(\widehat\beta-\beta)+\Lambda|.
\]
The first factor is \(o_p(S)\) by Step~1, and the second is
\(O_p(1)\) by \eqref{eq:estimator-order}. Since \(S\le2\E[Q]\) for all
large \(n\) by Step~3, we have \(r/\E[Q]=o_p(1)\), which proves part~(i).

Finally, for the finitely many \(n\) at which \(h_c<\underline h\) or
\(\E[Q_c]\le0\) for some candidate, set \(C_0:=0\), and set \(Q_c:=1\) and
\(r_c:=n(\widehat\beta_c-\beta)^2-1\) for every \(c\in[J]\). This preserves
the decomposition and the positivity of \(\E[Q]\) without affecting the
asymptotic statements.
\end{proof}
{}
\Needspace{6\baselineskip}
\subsection{Supporting Lemmas}
\label{sec:thm-main-lemmas}

\begin{lemma}[The Sample Denominator]
\label{lem:denominator}
(i) Under Assumption~\ref{assm:gaussian},
\begin{equation}
    \E[\widehat h]=h+\frac{\sigma_v^2\bar d}{n},
    \qquad
    \Var(\widehat h)\le\frac{4\sigma_v^2(H+\sigma_v^2)}{n}.
    \label{eq:denominator-mean}
\end{equation}
(ii) If Assumptions~\ref{assm:primitive} and \ref{assm:strong-core} also
hold, with \(J<\infty\) fixed and \(\Delta\) as in
\eqref{eq:leading-term}, then \(\E[(\widehat h-h)^2]\le C/n\) and
\begin{equation}
    \widehat h-h=O_p(n^{-1/2}),
    \qquad
    \Delta=O_p(n^{-1/2}).
    \label{eq:denominator-rate}
\end{equation}
(iii) Under the hypotheses of (ii), \(\P(\mathcal H_n)\to1\), and on
\(\mathcal H_n\) we have \(\En[\widehat\Pi_i^2]>0\) and
\(\widehat\beta=\En[\widehat\Pi_iy_i]/\widehat h\).
\end{lemma}

\begin{proof}
Since \(n\widehat h=x'\widehat\Pi=\Pi'\widehat\Pi+v'\widehat\Pi\), the
first identity in \eqref{eq:lasso-stein-identities} gives
\(\E[\widehat h]=h+\sigma_v^2\bar d/n\). For the variance, the gradient
of \(x'\widehat\Pi=x'\mu(x)\) with respect to \(v\) is
\(\widehat\Pi+Px\), whose norm is at most \(2\lVert x\rVert_2\) by
\eqref{eq:projection-properties} and
\eqref{eq:fit-contraction-at-zero}. The Gaussian Poincar\'e inequality
\eqref{eq:gaussian-poincare} therefore gives
\[
    \Var(\widehat h)
    =\frac{\Var(x'\widehat\Pi)}{n^2}
    \le\frac{4\sigma_v^2\E\lVert x\rVert_2^2}{n^2}
    =\frac{4\sigma_v^2(H+\sigma_v^2)}{n},
\]
which proves (i). Under (ii), \(\bar d=o(\sqrt n)\) by
\eqref{eq:primitive-rank-conclusions}, so that
\(\E[(\widehat h-h)^2]=\Var(\widehat h)+\sigma_v^4\bar d^2/n^2\le C/n\),
and Markov's inequality and \(h\ge\underline h\) give
\eqref{eq:denominator-rate}. Since \(\mathcal H_n^c\) is contained in
\(\{|\widehat h-h|>\underline h/2\}\), it follows that
\(\P(\mathcal H_n)\to1\). Finally, \(\widehat h=0\) whenever
\(\widehat\Pi=0\), so on \(\mathcal H_n\) the fit \(\widehat\Pi\) is
nonzero and the estimator is the ratio in \eqref{eq:iv-estimator}.
\end{proof}

\begin{lemma}
\label{lem:score-distance}
Suppose Assumptions~\ref{assm:gaussian} and~\ref{assm:primitive} hold,
with \(J<\infty\) fixed, and let \(\Gamma\) and \(L_n\) be as in
\eqref{eq:leading-term} and \eqref{eq:optimal-score}. Then
\begin{equation}
    \E[(\Gamma-L_n)^2]\le2S.
    \label{eq:Gamma-L-bound}
\end{equation}
\end{lemma}

\begin{proof}
With \(M\) the population projection residual in
\eqref{eq:projection-residual}, we have
\(\Gamma-L_n=\varepsilon'\{\widehat\Pi-(h/H)\Pi\}/(\sqrt n\,h)
=\varepsilon'M/(\sqrt n\,h)\). Using the Gaussian error
decomposition \eqref{eq:comparison-error-decomposition}, we can write
\[
    \varepsilon'M-\sigma_{\varepsilon v}d
    =\frac{\sigma_{\varepsilon v}}{\sigma_v^2}(v'M-\sigma_v^2d)+e'M.
\]
Conditional on \(v\), the vector \(M\) and the rank \(d\) are fixed,
while \(e'M\) has mean zero and variance
\((\sigma_\varepsilon^2-\sigma_{\varepsilon v}^2/\sigma_v^2)
\lVert M\rVert_2^2\). By \eqref{eq:residual-stein} and
\(\E\lVert M\rVert_2^2=n\mathcal A\), we therefore have
\[
    \E[(\varepsilon'M-\sigma_{\varepsilon v}d)^2]
    =\frac{\sigma_{\varepsilon v}^2}{\sigma_v^4}
       (\sigma_v^2n\mathcal A+\sigma_v^4\bar d)
     +\left(\sigma_\varepsilon^2
            -\frac{\sigma_{\varepsilon v}^2}{\sigma_v^2}\right)n\mathcal A
    =\sigma_\varepsilon^2n\mathcal A+\sigma_{\varepsilon v}^2\bar d.
\]
Writing \(\varepsilon'M=(\varepsilon'M-\sigma_{\varepsilon v}d)
+\sigma_{\varepsilon v}d\) and using \((a+b)^2\le2a^2+2b^2\), we
conclude that
\[
    \E[(\Gamma-L_n)^2]
    =\frac{\E[(\varepsilon'M)^2]}{nh^2}
    \le\frac{2}{nh^2}
      \left\{\sigma_\varepsilon^2n\mathcal A
             +\sigma_{\varepsilon v}^2\E[d^2+d]\right\}
    =2S.
\]
\end{proof}

\begin{lemma}[Moments of the Denominator Correction]
\label{lem:denominator-correction}
Suppose Assumptions~\ref{assm:gaussian}, \ref{assm:primitive}, and
\ref{assm:strong-core} hold, with \(J<\infty\) fixed, and let
\(\Delta\) and \(L_n\) be as in \eqref{eq:leading-term} and
\eqref{eq:optimal-score}. Then
\begin{align}
    \E[\Delta^2L_n^2]&\le\frac Cn,
    \label{eq:denominator-score-square}\\
    |\E[\Delta L_n^2]|
    &\le C\left\{\sqrt{\frac{S}{n}}+\frac1n\right\}.
    \label{eq:denominator-score-signed-bound}
\end{align}
\end{lemma}

\begin{proof}
Both bounds rest on the Stein identity
\begin{equation}
    \E[(\Pi'v)g(v)]=\sigma_v^2\E[\Pi'\nabla g(v)],
    \label{eq:stein-identity}
\end{equation}
which follows from Stein's lemma \citep[Lemma~2]{Stein-1981}.

We also use three facts. First, by the proof of
\Cref{lem:denominator}, the gradient of \(\Delta\) with respect to
\(v\) is
\[
    \nabla\Delta=\frac{\widehat\Pi+Px}{nh},
\]
so that \(\E\lVert\nabla\Delta\rVert_2^2
\le4\E\lVert x\rVert_2^2/(nh)^2\le C/n\). Second,
\Cref{lem:denominator} gives \(\E[\Delta]=\sigma_v^2\bar d/(nh)\) and
\(\E[\Delta^2]\le C/n\). Third, by
\eqref{eq:comparison-error-decomposition},
\[
    L_n=\frac{\sigma_{\varepsilon v}}{\sigma_v^2\sqrt n\,H}\Pi'v
       +\frac{\Pi'e}{\sqrt n\,H},
\]
where the second term is independent of \(v\), with mean zero and
variance \((\sigma_\varepsilon^2-\sigma_{\varepsilon v}^2/\sigma_v^2)/H\),
so that
\begin{equation}
    \E[L_n^2\mid v]
    =\frac{\sigma_\varepsilon^2-\sigma_{\varepsilon v}^2/\sigma_v^2}{H}
     +\frac{\sigma_{\varepsilon v}^2}{\sigma_v^4nH^2}(\Pi'v)^2.
    \label{eq:L-conditional-second}
\end{equation}

\emph{Step 1: \(\E[\Delta^2L_n^2]\le C/n\).}
Since \(\Delta\) depends only on \(v\), multiplying
\eqref{eq:L-conditional-second} by \(\Delta^2\) and taking expectations
yields \(\E[\Delta^2L_n^2]\le C\E[\Delta^2]+(C/n)\E[(\Pi'v)^2\Delta^2]\),
so it suffices to show that \(\E[(\Pi'v)^2\Delta^2]\le C\). Apply
\eqref{eq:stein-identity} with \(g(v)=(\Pi'v)\Delta^2\). Since
\(\nabla(\Pi'v)=\Pi\) and \(\lVert\Pi\rVert_2^2=nH\), the product rule
gives
\[
\begin{aligned}
    \E[(\Pi'v)^2\Delta^2]
    &=\sigma_v^2nH\E[\Delta^2]
      +2\sigma_v^2\E[(\Pi'v)\Delta\cdot\Pi'\nabla\Delta]\\
    &\le\sigma_v^2nH\E[\Delta^2]
      +\tfrac12\E[(\Pi'v)^2\Delta^2]
      +2\sigma_v^4\E[(\Pi'\nabla\Delta)^2],
\end{aligned}
\]
where the inequality comes from Young's inequality.\footnote{Young's
inequality states that \(2ab\le\eta a^2+b^2/\eta\) for real \(a,b\) and
any \(\eta>0\), which follows from expanding
\((\sqrt\eta\,a-b/\sqrt\eta)^2\ge0\). Here we apply it with
\(a=(\Pi'v)\Delta\), \(b=\sigma_v^2\Pi'\nabla\Delta\), and
\(\eta=1/2\).} By Cauchy--Schwarz, the last term is at most
\(2\sigma_v^4nH\E\lVert\nabla\Delta\rVert_2^2\le C\), and the first is
at most \(C\) because \(\E[\Delta^2]\le C/n\). Moving
\(\tfrac12\E[(\Pi'v)^2\Delta^2]\) to the left-hand side, we obtain
\(\E[(\Pi'v)^2\Delta^2]\le C\).

\emph{Step 2: \(|\E[\Delta L_n^2]|\le C\{\sqrt{S/n}+1/n\}\).}
Multiply \eqref{eq:L-conditional-second} by
\(\Delta\) and take expectations. Applying \eqref{eq:stein-identity} with
\(g(v)=(\Pi'v)\Delta\), we have
\(\E[(\Pi'v)^2\Delta]=\sigma_v^2nH\E[\Delta]
+\sigma_v^2\E[(\Pi'v)\Pi'\nabla\Delta]\). After this substitution, the
two multiples of \(\E[\Delta]\) add up to
\((\sigma_\varepsilon^2-\sigma_{\varepsilon v}^2/\sigma_v^2)/H
+\sigma_{\varepsilon v}^2/(\sigma_v^2H)=\sigma_\varepsilon^2/H\), so
that
\begin{equation}
    \E[\Delta L_n^2]
    =\frac{\sigma_\varepsilon^2}{H}\E[\Delta]
     +\frac{\sigma_{\varepsilon v}^2}{\sigma_v^2nH^2}
      \E[(\Pi'v)\Pi'\nabla\Delta].
    \label{eq:denominator-score-signed-identity}
\end{equation}
The first term is at most \(C\bar d/n\). For the second, since
\(P\widehat\Pi=\widehat\Pi\), we have \((I-P)\Pi=-(H/h)(I-P)M\), so that
\[
    \E\lVert(I-P)\Pi\rVert_2^2
    \le\frac{H^2}{h^2}\E\lVert M\rVert_2^2
    =\frac{H^2}{h^2}n\mathcal A.
\]
Since \(P\) is symmetric, the formula for \(\nabla\Delta\) above shows
that \(nh\,\Pi'\nabla\Delta=\Pi'\widehat\Pi+x'P\Pi\), and writing
\(x'P\Pi=x'\Pi-x'(I-P)\Pi\) with \(x'\Pi=nH+\Pi'v\), we have
\[
    nh\,\E[(\Pi'v)\Pi'\nabla\Delta]
    =\underbrace{\E[(\Pi'v)\Pi'\widehat\Pi]}_{\mathrm{I}}
     +\underbrace{\E[(\Pi'v)(nH+\Pi'v)]}_{\mathrm{II}}
     -\underbrace{\E[(\Pi'v)x'(I-P)\Pi]}_{\mathrm{III}}.
\]
Applying \eqref{eq:stein-identity} with \(g(v)=\Pi'\widehat\Pi\), whose
gradient is \(P\Pi\), we have
\[
    0\le\mathrm I=\sigma_v^2\E[\Pi'P\Pi]\le\sigma_v^2nH,
\]
and
\(\mathrm{II}=\E[(\Pi'v)^2]=\sigma_v^2nH\) because \(\E[\Pi'v]=0\).
For \(\mathrm{III}\), Gaussian fourth moments give
\(\E[(\Pi'v)^4]+\E\lVert x\rVert_2^4\le Cn^2\), so that Cauchy--Schwarz
bounds \(\E[(\Pi'v)^2\lVert x\rVert_2^2]\) by \(Cn^2\), and a second
application yields
\[
    |\mathrm{III}|
    \le\left\{\E[(\Pi'v)^2\lVert x\rVert_2^2]\,
              \E\lVert(I-P)\Pi\rVert_2^2\right\}^{1/2}
    \le Cn^{3/2}\sqrt{\mathcal A}.
\]
Combining the bounds on \(\mathrm I\), \(\mathrm{II}\), and
\(\mathrm{III}\) and dividing by \(nh\), we have
\(|\E[(\Pi'v)\Pi'\nabla\Delta]|\le C(1+\sqrt{n\mathcal A})\), so that
by \eqref{eq:denominator-score-signed-identity},
\[
    |\E[\Delta L_n^2]|
    \le C\left\{\frac{\bar d}{n}+\sqrt{\frac{\mathcal A}{n}}+\frac1n\right\}.
\]
Finally, the first bound in \eqref{eq:score-ratio-bounds} gives
\(\bar d/n+\sqrt{\mathcal A/n}\le C\sqrt{S/n}\), which proves
\eqref{eq:denominator-score-signed-bound}.
\end{proof}

\begin{lemma}[The First-Order Expansion]
\label{lem:first-order}
Suppose Assumptions~\ref{assm:gaussian} and~\ref{assm:primitive} hold,
with \(J<\infty\) fixed. Then \eqref{eq:first-order} holds for every
candidate.
\end{lemma}

\begin{proof}
By \Cref{lem:primitive-consequences}, \(h\ge\underline h\) and
\(\mathcal A\to0\), while \Cref{lem:primitive-lasso} gives
\(d<\vartheta s+1\) on \(\mathcal G\). Since
\(\P(\mathcal G^c)\le\eta_n\to0\) by the proof of
\Cref{lem:primitive-consequences} and \(s/\sqrt n\to0\) by
\Cref{assm:primitive}(v), we have \(d/\sqrt n=o_p(1)\) and
\(\bar d/n\le(\vartheta s+1)/n+\P(\mathcal G^c)\to0\). The calculation in
the proof of \Cref{lem:score-distance} therefore gives
\[
    \frac{\E[(\varepsilon'M-\sigma_{\varepsilon v}d)^2]}{nh^2}
    =\frac{\sigma_\varepsilon^2\mathcal A+\sigma_{\varepsilon v}^2\bar d/n}{h^2}
    \longrightarrow0,
\]
so that \(\Gamma-L_n=(\varepsilon'M-\sigma_{\varepsilon v}d)/(\sqrt n\,h)
+\sigma_{\varepsilon v}d/(\sqrt n\,h)=o_p(1)\). By \Cref{lem:denominator}(i),
\(\E[\Delta]=\sigma_v^2\bar d/(nh)\to0\) and \(\Var(\Delta)\le C/n\), so
that \(\Delta=o_p(1)\) and \(\P(\mathcal H_n)\to1\). On \(\mathcal H_n\),
\(\sqrt n(\widehat\beta-\beta)=\Gamma/(1+\Delta)\) by
\eqref{eq:ratio-expansion}, and therefore
\[
    \sqrt n(\widehat\beta-\beta)-L_n
    =\frac{\Gamma-L_n-\Delta L_n}{1+\Delta}
    =o_p(1),
\]
since \(L_n=O_p(1)\). This is \eqref{eq:first-order}.
\end{proof}
{}
\Needspace{14\baselineskip}
\section{Proofs of Proposition~\ref{prop:ratio} and
Theorem~\ref{thm:feasible}}
\label{sec:proof-thm-feasible}

\subsection{Proof of Proposition~\ref{prop:ratio}}

\begin{proof}
The first conclusion \eqref{eq:rank-ratio-main} is part of
\eqref{eq:primitive-rank-conclusions} of
Lemma~\ref{lem:primitive-consequences}.  For
\eqref{eq:rank-factor-ratio-main}, write
\begin{equation}
    \frac{d^2+d}{\E[d^2+d]}
    =
    \frac{
       (d/\bar d)^2+(d/\bar d)/\bar d
    }{
       \E[d^2+d]/\bar d^2
    }.
    \label{eq:rank-factor-algebra}
\end{equation}
By \eqref{eq:primitive-rank-conclusions},
\(d/\bar d\to_p1\), \(\bar d\to\infty\), and
\(\E[d^2+d]/\bar d^2\to1\), so the right-hand side of
\eqref{eq:rank-factor-algebra} converges to one in probability, which
is \eqref{eq:rank-factor-ratio-main}.
\end{proof}

\Needspace{9\baselineskip}
\subsection{Proof of Theorem~\ref{thm:feasible}}

\begin{proof}
We work on \(\mathcal H_n\), where \(\widehat S\) is finite.

\emph{Step 1:
\(\widehat\sigma_\varepsilon^2\En[\widehat\Pi_i^2]/\widehat h^2
=\sigma_\varepsilon^2\mathcal A/h^2+B_n+o_p(S)\).}
Applying the decomposition in \eqref{eq:amse-approx} to sample moments,
and recalling that \(H=\En[\Pi_i^2]\), we have
\(\En[\widehat\Pi_i^2]=\widehat{\mathcal A}
+\{\En[\Pi_i\widehat\Pi_i]\}^2/H\), with \(\widehat{\mathcal A}\)
as defined in \Cref{sec:feasible}. Since \(x_i=\Pi_i+v_i\), we
also have
\(\En[\Pi_i\widehat\Pi_i]=\widehat h-\En[v_i\widehat\Pi_i]\),
so that, writing
\(\widehat q:=\En[v_i\widehat\Pi_i]/\widehat h\),
\begin{equation}
    \frac{\widehat\sigma_\varepsilon^2
          \En[\widehat\Pi_i^2]}{\widehat h^2}
      =\frac{\widehat\sigma_\varepsilon^2\widehat{\mathcal A}}
             {\widehat h^2}
       +\frac{\widehat\sigma_\varepsilon^2}{H}(1-\widehat q)^2.
    \label{eq:sample-fit-decomposition}
\end{equation}
For the first term, \Cref{lem:fitted-moment-expansion,lem:denominator,lem:pilot-rates}
give \(\widehat{\mathcal A}-\mathcal A=o_p(S)\) and
\(\widehat h-h,\widehat\sigma_\varepsilon^2-\sigma_\varepsilon^2
=O_p(n^{-1/2})\), so that, since \(\mathcal A\le CS\) by
\eqref{eq:score-rates},
\[
    \frac{\widehat\sigma_\varepsilon^2\widehat{\mathcal A}}
         {\widehat h^2}
      =\frac{\sigma_\varepsilon^2\mathcal A}{h^2}
        +o_p(S)+O_p(\mathcal An^{-1/2})
      =\frac{\sigma_\varepsilon^2\mathcal A}{h^2}+o_p(S).
\]
For the second term, \Cref{lem:fitted-moment-expansion} gives
\(\widehat q=q+o_p(S)\), where \(q=\En[\Pi_iv_i]/H=O_p(n^{-1/2})\)
because \(\E[q^2]=\sigma_v^2/(nH)\). Factoring the difference of squares,
\[
    (1-\widehat q)^2-(1-q)^2
    =(q-\widehat q)(2-q-\widehat q)
    =o_p(S),
\]
while \((1-q)^2=1-2q+q^2\) with \(q^2=O_p(n^{-1})=o_p(S)\), since
\(nS\to\infty\). As \(\widehat\sigma_\varepsilon^2/H=O_p(1)\),
the second term in \eqref{eq:sample-fit-decomposition} equals
\(\widehat\sigma_\varepsilon^2(1-2q)/H+o_p(S)=B_n+o_p(S)\), with
\(B_n\) the common centering in \eqref{eq:common-term}. We keep
\(\widehat\sigma_\varepsilon^2\) inside \(B_n\), since its estimation
error is of order \(n^{-1/2}\) and need not be \(o_p(S)\). Adding the
two terms proves Step~1.

\emph{Step 2:
\(\widehat\sigma_{\varepsilon v}^{\,2}(d^2+d)/(n\widehat h^2)
=\sigma_{\varepsilon v}^2\E[d^2+d]/(nh^2)+o_p(S)\).}
Write
\begin{equation}
    \frac{\widehat\sigma_{\varepsilon v}^{\,2}(d^2+d)}
         {n\widehat h^2}
    =\frac{\sigma_{\varepsilon v}^2\E[d^2+d]}{nh^2}
      \cdot
      \frac{\widehat\sigma_{\varepsilon v}^{\,2}}
           {\sigma_{\varepsilon v}^2}
      \cdot
      \frac{d^2+d}{\E[d^2+d]}
      \cdot
      \frac{h^2}{\widehat h^2}.
    \label{eq:rank-term-factorization}
\end{equation}
Of the last three factors, the first converges to one in probability
by \Cref{lem:pilot-rates} and \(\sigma_{\varepsilon v}\ne0\), the second
by \Cref{prop:ratio}, and the third by \Cref{lem:denominator}. Since the
leading factor is at most \(S\), the product equals
\(\sigma_{\varepsilon v}^2\E[d^2+d]/(nh^2)+o_p(S)\).

\emph{Step 3: Part~(i).}
Adding Steps~1 and~2 and using \eqref{eq:feasible-score} and
\eqref{eq:closed-form-criterion}, we obtain
\(\widehat S=B_n+S+o_p(S)\), which is
\eqref{eq:feasible-consistency}. By \Cref{thm:main}(ii),
\(\bar Q/S\to1\), so that
\[
    \widehat S-B_n-\bar Q
    =(\widehat S-B_n-S)+(S-\bar Q)
    =o_p(S)=o_p(\bar Q),
\]
which is \eqref{eq:feasible-expansion}.

\emph{Step 4: Part~(ii).}
This step compares candidates, so we restore the subscript. By
\Cref{thm:main}(ii), every \(\bar Q_c\) is positive. Let
\[
    \omega_n:=\max_{c\in[J]}
       \frac{|\widehat S_c-B_n-\bar Q_c|}{\bar Q_c}.
\]
By Step~3, \(\omega_n\to_p0\). On \(\{\omega_n<1\}\), since
\(\widehat c\) minimizes the feasible criterion, for every \(c\in[J]\),
\[
    B_n+(1-\omega_n)\bar Q_{\widehat c}
    \le\widehat S_{\widehat c}
    \le\widehat S_c
    \le B_n+(1+\omega_n)\bar Q_c.
\]
Cancelling \(B_n\), minimizing the right-hand side over \(c\), and
dividing by \((1-\omega_n)\min_c\bar Q_c\), we obtain
\begin{equation}
    0\le\frac{\bar Q_{\widehat c}}{\min_c\bar Q_c}-1
      \le\frac{2\omega_n}{1-\omega_n},
    \label{eq:oracle-ratio-bound}
\end{equation}
and the right-hand side converges to zero in probability. Since
\(\P(\omega_n<1)\to1\), this proves \eqref{eq:oracle-equivalence}.
\end{proof}

\Needspace{9\baselineskip}
\subsection{Supporting Lemmas}
\label{sec:thm-feasible-lemmas}

\begin{lemma}[Sample Counterparts of \(\mathcal A\) and \(q\)]
\label{lem:fitted-moment-expansion}
Suppose Assumptions~\ref{assm:gaussian}, \ref{assm:primitive}, and
\ref{assm:strong-core} hold, with
\(J<\infty\) fixed, and let \(q=\En[\Pi_iv_i]/H\) and
\(\widehat q=\En[v_i\widehat\Pi_i]/\widehat h\) be as in the
proofs of \Cref{thm:main,thm:feasible}. With
\(\widehat{\mathcal A}=\min_{a\in\mathbb R}\En[(\widehat\Pi_i-a\Pi_i)^2]\)
as defined in \Cref{sec:feasible}, we have
\begin{equation}
    \widehat{\mathcal A}-\mathcal A
    =O_p\!\left(\sqrt{\frac{\mathcal A}{n}}+\frac1n\right)
    =o_p(S),
    \label{eq:sample-approximation-rate}
\end{equation}
and, on \(\mathcal H_n\),
\begin{equation}
    \widehat q-q
    =O_p\!\left(\sqrt{\frac{\mathcal A}{n}}
                 +\frac{\bar d}{n}+\frac1n\right)
    =o_p(S).
    \label{eq:fitted-score-rate}
\end{equation}
\end{lemma}

\begin{proof}
\emph{Step 1: \(\widehat{\mathcal A}-\mathcal A
=O_p(\sqrt{\mathcal A/n}+1/n)\).}
Applying the decomposition in \eqref{eq:amse-approx} to sample moments,
as in the proof of \Cref{thm:feasible},
\[
    \widehat{\mathcal A}
    =\En[\widehat\Pi_i^2]-\frac{\{\En[\Pi_i\widehat\Pi_i]\}^2}{H}
    =\frac1n\left\lVert\left(I-\frac{\Pi\Pi'}{nH}\right)
                     \widehat\Pi\right\rVert_2^2.
\]
Since \(\Ebar[\Pi_i\widehat\Pi_i]=h\), we have
\(\E[\{\En[\Pi_i\widehat\Pi_i]\}^2]=h^2+\Var(\En[\Pi_i\widehat\Pi_i])\),
and taking expectations and comparing with \eqref{eq:amse-approx}, we
obtain
\(\mathcal A-\E[\widehat{\mathcal A}]=\Var(\En[\Pi_i\widehat\Pi_i])/H\).
The gradient of \(\En[\Pi_i\widehat\Pi_i]\) with respect to \(v\) is
\(P\Pi/n\), whose squared norm is at most \(\lVert\Pi\rVert_2^2/n^2=H/n\),
so that by the Gaussian Poincar\'e inequality \eqref{eq:gaussian-poincare},
\begin{equation}
    0\le\mathcal A-\E[\widehat{\mathcal A}]\le\frac{\sigma_v^2}{n}.
    \label{eq:sample-approximation-bias}
\end{equation}
Similarly, the gradient of \(\widehat{\mathcal A}\) is
\((2/n)P\{I-\Pi\Pi'/(nH)\}\widehat\Pi\), whose squared norm is at most
\(4\widehat{\mathcal A}/n\), so that
\[
    \Var(\widehat{\mathcal A})
    \le\frac{4\sigma_v^2}{n}\E[\widehat{\mathcal A}]
    \le\frac{4\sigma_v^2\mathcal A}{n}.
\]
The claim follows from Chebyshev's inequality and
\eqref{eq:sample-approximation-bias}.

\emph{Step 2: \(\widehat q-q
=O_p(\sqrt{\mathcal A/n}+\bar d/n+1/n)\) on \(\mathcal H_n\).}
With \(M\) the population projection residual in
\eqref{eq:projection-residual}, we have
\(\En[v_i\widehat\Pi_i]=v'M/n+hq\), and therefore
\begin{equation}
    \widehat q-q
    =\frac{v'M/n}{\widehat h}-q\frac{\widehat h-h}{\widehat h}.
    \label{eq:fitted-score-decomposition}
\end{equation}
By \eqref{eq:residual-stein} and Markov's inequality,
\(v'M=\sigma_v^2d+O_p(\sqrt{n\mathcal A}+\sqrt{\bar d})\), and
\(d=O_p(\bar d)\) with \(\bar d\to\infty\) by
\eqref{eq:primitive-rank-conclusions}, so that
\(v'M/n=O_p(\sqrt{\mathcal A/n}+\bar d/n)\). Substituting this bound,
\(q=O_p(n^{-1/2})\), \(\widehat h-h=O_p(n^{-1/2})\), and
\(\widehat h\ge\underline h/2\) on \(\mathcal H_n\) into
\eqref{eq:fitted-score-decomposition} yields the claim.

\emph{Step 3: Both rates are \(o_p(S)\).}
By the second bound in \eqref{eq:score-ratio-bounds},
\(\sqrt{\mathcal A/n}+\bar d/n\le CS/\bar d\), which is \(o(S)\)
because \(\bar d\to\infty\) by \eqref{eq:primitive-rank-conclusions}.
The remaining term \(1/n\) is also \(o(S)\), since \(nS\to\infty\) by
\eqref{eq:score-rates}.
\end{proof}

\Needspace{9\baselineskip}
\begin{lemma}[Estimating the Error Moments]
\label{lem:pilot-rates}
Suppose Assumptions~\ref{assm:gaussian}, \ref{assm:primitive}, and
\ref{assm:strong-core} hold and
\(J<\infty\) is fixed.  For any (possibly data-dependent)
\([J]\)-valued index \(\check c\),
\begin{equation}
    \widehat\beta_{\check c}-\beta=O_p(n^{-1/2}).
       \label{eq:baseline-rate}
\end{equation}
If the pilot satisfies \(\check\beta-\beta=O_p(n^{-1/2})\) and
\(\widehat\sigma_\varepsilon^2,\widehat\sigma_{\varepsilon v}\) are as in
\eqref{eq:nuisance-estimators}, then
\begin{equation}
    \widehat\sigma_\varepsilon^2-\sigma_\varepsilon^2
       =O_p(n^{-1/2}),
    \qquad
    \widehat\sigma_{\varepsilon v}-\sigma_{\varepsilon v}=O_p(n^{-1/2}).
       \label{eq:nuisance-rates}
\end{equation}
\end{lemma}

\Needspace{6\baselineskip}
\begin{proof}
Since \(\check c\in[J]\) and \(J\) is fixed, \eqref{eq:estimator-order}
in the proof of \Cref{thm:main} gives \eqref{eq:baseline-rate}.

By independence, Gaussian fourth moments, and Chebyshev's inequality,
\[
    \En[\varepsilon_i^2]
       =\sigma_\varepsilon^2+O_p(n^{-1/2}),\qquad
    \En[\varepsilon_i x_i]
       =\sigma_{\varepsilon v}+O_p(n^{-1/2}),
\]
while \(\En[x_i^2]=O_p(1)\) by Markov's inequality. Since
\(y_i-\check\beta x_i=\varepsilon_i-\delta_\beta x_i\) for
\(\delta_\beta:=\check\beta-\beta\), we have
\begin{align}
    \widehat\sigma_\varepsilon^2
       =\En[\varepsilon_i^2]
          -2\delta_\beta\En[\varepsilon_i x_i]
          +\delta_\beta^2\En[x_i^2], \quad
    \widehat\sigma_{\varepsilon v}
       =\En[\varepsilon_i x_i]
          -\delta_\beta\En[x_i^2].
          \label{eq:residual-moment-expansion}
\end{align}
Substituting these moment bounds and \(\delta_\beta=O_p(n^{-1/2})\)
into the identities \eqref{eq:residual-moment-expansion} proves
\eqref{eq:nuisance-rates}.
\end{proof}
{}
\section{LASSO Fitted-Value Geometry}
\label{sec:proofs-geometry}

We treat \(Z\) and
\(\lambda>0\) as fixed and place no distributional restrictions on the
response. As in \Cref{app:notation}, let \(u(w)\) and \(E(w)\) denote the
LASSO residual and equicorrelation set for response \(w\). Both are
determined by the unique fitted vector \(\mu(w)\). Define the closed
convex set containing zero
\[
    \mathcal K
    :=\big\{
       a\in\mathbb R^n:
       \lVert Z'a/n\rVert_\infty\le\lambda
      \big\}.
\]
In \Cref{lem:geometry}, \(\widehat\pi(w)\) is any minimizer of
\eqref{eq:lasso} with response \(w\), while \Cref{lem:independent-support} fixes a Borel
choice.

\Needspace{6\baselineskip}
\begin{lemma}[Geometry of the LASSO Fit]
\label{lem:geometry}
For every \(w\in\mathbb R^n\) and every minimizer \(\widehat\pi(w)\) of
the problem in \eqref{eq:lasso} with response \(w\),
\begin{equation}
    |\En[z_{ij}u_i(w)]|\le\lambda,
    \qquad
    \widehat\pi_j(w)\En[z_{ij}u_i(w)]
       =\lambda|\widehat\pi_j(w)|,
    \quad j\in[p],
    \label{eq:kkt-conditions}
\end{equation}
and hence
\(\supp\{\widehat\pi(w)\}\subseteq E(w)\) and
\begin{equation}
    \En[\mu_i(w)w_i]
    =\En[\mu_i(w)^2]
      +\lambda\lVert\widehat\pi(w)\rVert_1.
    \label{eq:kkt-denominator}
\end{equation}
Moreover,
\begin{equation}
    u(w)=P_{\mathcal K}(w),
    \qquad
    \mu(w)=w-P_{\mathcal K}(w),
    \label{eq:dual-projection}
\end{equation}
where \(P_{\mathcal K}\) is Euclidean projection onto \(\mathcal K\).
The map \(\mu\) is one-Lipschitz, \(\mu(0)=0\), and
\begin{equation}
    \lVert\mu(w)\rVert_2\le\lVert w\rVert_2.
    \label{eq:fit-contraction-at-zero}
\end{equation}
For Lebesgue-almost every \(w\in\mathbb R^n\),
\begin{equation}
    D\mu(w)=P_{\col(Z_{E(w)})},
    \qquad
    D\mu(w)'=D\mu(w)=D\mu(w)^2,
    \qquad
    D\mu(w)\mu(w)=\mu(w).
    \label{eq:geometry-properties}
\end{equation}
\end{lemma}

\begin{proof}
The KKT conditions imply \eqref{eq:kkt-conditions}, and summing
the second equality over
\(j\in[p]\) gives \eqref{eq:kkt-denominator}.
Multiplying the objective in \eqref{eq:lasso} by \(n\) puts it in the
standard form with penalty \(n\lambda\).
Lemmas~1 and 3 of \citet{TibshiraniTaylor-2012} then give
\eqref{eq:dual-projection} and the one-Lipschitz property. Since
\(0\in\mathcal K\), \(\mu(0)=0\), which proves
\eqref{eq:fit-contraction-at-zero}. As in the proof of their Theorem~1, their Lemma~5
yields \(D\mu(w)=Z_{E(w)}Z_{E(w)}^+=P_{\col(Z_{E(w)})}\) for
Lebesgue-almost every \(w\), with \(+\) denoting the Moore--Penrose
pseudoinverse. The remaining identities in \eqref{eq:geometry-properties}
follow because this is an orthogonal projection and
\(\mu(w)\in\col(Z_{E(w)})\) by \eqref{eq:kkt-conditions}.
\end{proof}

\begin{lemma}[A Measurable LASSO Solution]
\label{lem:independent-support}
There exists a Borel-measurable LASSO solution
\(w\mapsto\widehat\pi(w)\) such that, for every \(w\in\mathbb R^n\),
\begin{equation}
    \rank\!\left(
       Z_{\supp\{\widehat\pi(w)\}}
    \right)
    =
    \left|\supp\{\widehat\pi(w)\}\right|
    \le\rank(Z)
    \le n,
    \label{eq:independent-support}
\end{equation}
and \(Z\widehat\pi(w)\) equals the unique fitted vector \(\mu(w)\).
\end{lemma}

\begin{proof}
Lemma~14 of \citet{Tibshirani-2013}, with penalty \(n\lambda\), gives
for every \(w\) a minimizer whose support indexes linearly independent
columns. Order the finite set
\(\mathcal B=\{B\subseteq[p]:\rank(Z_B)=|B|\}\), including
\(B=\varnothing\). For nonempty \(B\in\mathcal B\), set
\[
    b_B^{(B)}(w)=(Z_B'Z_B)^{-1}Z_B'\mu(w),\qquad
    b_{B^c}^{(B)}(w)=0,
\]
and set \(b^{(\varnothing)}(w)=0\). Let \(\mathcal R_B\) be the set of
responses for which \(b^{(B)}(w)\) satisfies the KKT conditions
\eqref{eq:kkt-conditions} with residual \(w-Zb^{(B)}(w)\).
Since \(\mu\) and \(b^{(B)}\) are continuous by \Cref{lem:geometry},
each \(\mathcal R_B\) is the preimage of a closed set under a continuous
map, hence closed. For any
\(w\), the minimizer above has support \(B\) and,
by uniqueness of \(\mu(w)\) and full column rank of \(Z_B\), equals
\(b^{(B)}(w)\), including when \(B=\varnothing\). The sets
\(\mathcal R_B\) therefore cover \(\mathbb R^n\). Choose \(b^{(B)}(w)\) for the
first \(B\) with \(w\in\mathcal R_B\). The regions on which each \(B\) is chosen
are Borel, and the KKT conditions are sufficient for a minimum, so this
is a Borel-measurable minimizer.
Its support lies in the independent set \(B\), which gives
\eqref{eq:independent-support}, and its fitted vector is \(\mu(w)\).
\end{proof}
{}
\section{Gaussian Identities and Rank Variance}

By Rademacher's theorem, any locally Lipschitz
\(f:\mathbb R^n\to\mathbb R\) is differentiable almost everywhere.  Define
\(\nabla f(w)\) to be its gradient where it is differentiable and zero
otherwise.

\begin{lemma}[Poincar\'e Inequality and Stein Identities]
\label{lem:gaussian-tools}
Let \(W\sim N(0,\sigma^2I_n)\) with \(\sigma>0\). If
\(g:\mathbb R^n\to\mathbb R\) is locally Lipschitz with
\(\E[g(W)^2]<\infty\) and
\(\E[\lVert\nabla g(W)\rVert_2^2]<\infty\), then
\begin{equation}
    \Var\{g(W)\}
    \le\sigma^2\E[\lVert\nabla g(W)\rVert_2^2].
    \label{eq:gaussian-poincare}
\end{equation}
Separately, let \(m:\mathbb R^n\to\mathbb R^n\) be Lipschitz, let \(Dm(w)\)
be its Jacobian where \(m\) is differentiable and zero otherwise, and define
\(\diver m(w):=\operatorname{tr}\{Dm(w)\}\). Then
\begin{align}
    \E[W'm(W)]&=\sigma^2\E[\diver m(W)],
       \label{eq:stein-one-general}\\
    \E[\{W'm(W)-\sigma^2\diver m(W)\}^2]
       &=\sigma^2\E[\lVert m(W)\rVert_2^2]
         +\sigma^4\E[\operatorname{tr}\{Dm(W)^2\}].
       \label{eq:stein-two-general}
\end{align}
For this \(m\), if \(g\) is a locally Lipschitz scalar function with
\(|g(w)|+\lVert\nabla g(w)\rVert_2\le C(1+\lVert w\rVert_2^r)\)
almost everywhere for some \(C<\infty\) and \(r\ge0\), then
\begin{equation}
    \E[\{W'm(W)-\sigma^2\diver m(W)\}g(W)]
       =\sigma^2\E[m(W)'\nabla g(W)].
    \label{eq:stein-covariance}
\end{equation}
\end{lemma}

\begin{proof}
Inequality \eqref{eq:gaussian-poincare} is the Gaussian Poincar\'e
inequality, which is Theorem~2.2 of \citet{BellecZhang-2021} with the
vector field set to zero.
By Lipschitz continuity, \(m\) has at most linear growth and an
almost-everywhere bounded weak Jacobian, so its Gaussian moments are
finite. The
first-order Stein formula of \citet{Stein-1981}, applied componentwise,
gives \eqref{eq:stein-one-general}. For \eqref{eq:stein-covariance}, we
apply it to \(m_i g\) instead, which the growth condition on \(g\) makes
integrable, and sum over \(i\) with the product rule. Theorem~2.1(ii) of
\citet{BellecZhang-2021} gives \eqref{eq:stein-two-general}.
The transpose convention in their Jacobian is immaterial here, since it
leaves the divergence and \(\operatorname{tr}\{(Dm)^2\}\) unchanged.
\end{proof}

\begin{lemma}[Moments of a Recentered LASSO Fit]
\label{lem:offset-fit-moments}
Suppose Assumption~\ref{assm:gaussian}(i) holds, and let
\(b\in\mathbb R^n\) be nonrandom.  Define
\(m(a):=\mu(\Pi+a)-b\) for \(a\in\mathbb R^n\).  Let \(d\) be as in
\eqref{eq:realized-rank} and \(\bar d:=\E[d]\) as in \eqref{eq:rank-moments}.
Then \(m\) is one-Lipschitz, and
\begin{equation}
    \E[v'm(v)]=\sigma_v^2\bar d,
    \qquad
    \E[\{v'm(v)-\sigma_v^2d\}^2]
       =\sigma_v^2\E[\lVert m(v)\rVert_2^2]+\sigma_v^4\bar d.
       \label{eq:offset-stein}
\end{equation}
\end{lemma}

\begin{proof}
By Lemma~\ref{lem:geometry}, \(m\) is one-Lipschitz and its weak
Jacobian \(Dm(a)\) is almost everywhere an orthogonal projection. Since
\(\diver m(v)=\operatorname{tr}\{Dm(v)^2\}=d\) almost surely,
\eqref{eq:stein-one-general} and \eqref{eq:stein-two-general} give
\eqref{eq:offset-stein}.
\end{proof}

\begin{lemma}[Second Moment of the Numerator]
\label{lem:gaussian-moments}
Suppose Assumption~\ref{assm:gaussian}(i) holds.  Then
\begin{equation}
    \E[(\varepsilon'\widehat\Pi)^2]
    =\sigma_\varepsilon^2
       \E[\lVert\widehat\Pi\rVert_2^2]
      +\sigma_{\varepsilon v}^2\E[d^2+d].
    \label{eq:structural-second}
\end{equation}
Define
\begin{equation*}
    \zeta:=v'\widehat\Pi-\sigma_v^2d.
\end{equation*}
Then
\begin{equation}
    \E[v'\widehat\Pi]=\sigma_v^2\E[d],\qquad
    \E[\zeta^2]=\sigma_v^2\E[\lVert\widehat\Pi\rVert_2^2]
      +\sigma_v^4\E[d],\qquad
    \E[d\zeta]=0.
    \label{eq:lasso-stein-identities}
\end{equation}
\end{lemma}

\begin{proof}
Lemma~\ref{lem:offset-fit-moments} at \(b=0\) gives
the first two identities in \eqref{eq:lasso-stein-identities}.  Write
\(m(a):=\mu(\Pi+a)\) and \(P:=Dm(v)\).  By
Lemma~\ref{lem:geometry}, \(P'=P=P^2\) and
\(P\widehat\Pi=\widehat\Pi\) almost surely.  For \(g(a):=a'm(a)\),
\(\nabla g(v)=\widehat\Pi+Pv\).
By \eqref{eq:stein-covariance} together with the first two identities
in \eqref{eq:lasso-stein-identities},
\[
\begin{aligned}
    \sigma_v^2\E[d\zeta]
    &=\E[(v'\widehat\Pi)\zeta]-\E[\zeta^2]\\
    &=\sigma_v^2\E[\widehat\Pi'(\widehat\Pi+Pv)]-\E[\zeta^2]\\
    &=\sigma_v^2\E[\lVert\widehat\Pi\rVert_2^2]
       +\sigma_v^4\E[d]-\E[\zeta^2]=0.
\end{aligned}
\]
This proves the last identity in \eqref{eq:lasso-stein-identities}.
Therefore,
\[
    \E[(v'\widehat\Pi)^2]
    =\E[(\zeta+\sigma_v^2d)^2]
    =\sigma_v^2\E[\lVert\widehat\Pi\rVert_2^2]
      +\sigma_v^4\E[d^2+d].
\]
The Gaussian error decomposition \eqref{eq:comparison-error-decomposition}
and the fact that \(\widehat\Pi\) is \(\sigma(v)\)-measurable yield
\[
    \E[(\varepsilon'\widehat\Pi)^2]
    =\frac{\sigma_{\varepsilon v}^2}{\sigma_v^4}
       \E[(v'\widehat\Pi)^2]
     +\left(\sigma_\varepsilon^2
             -\frac{\sigma_{\varepsilon v}^2}{\sigma_v^2}\right)
       \E[\lVert\widehat\Pi\rVert_2^2].
\]
Substitution proves \eqref{eq:structural-second}.
\end{proof}

\citet[Theorem~3.11]{BellecZhang-2021} give a bound of this form for
the size of the LASSO support when the dictionary is in general
position, so that the LASSO solution is unique. We do not need this
condition, although our proof follows theirs.

\begin{lemma}[Variance of the Effective Dimension]
\label{lem:rank-variance}
Fix \(\Pi\in\mathbb R^n\), \(Z\in\mathbb R^{n\times p}\), and
\(\lambda>0\), and let
\(\mu\) denote the LASSO fitted-value map associated with \eqref{eq:lasso}.  For
\(v\sim N(0,\sigma_v^2I_n)\) with \(\sigma_v^2>0\), define
\begin{equation*}
    d:=\rank(Z_{E(\Pi+v)}),
    \qquad
    \bar d:=\E[d].
\end{equation*}
Then \(d=\diver\mu(\Pi+v)\) \(\P\)-almost surely.
There is a universal constant \(C<\infty\) such that
\begin{equation}
    \Var(d)
    \le C\left[
       1+\bar d\log\!\left\{
          \frac{ep}{\bar d\vee1}
       \right\}
    \right].
    \label{eq:rank-variance-bound}
\end{equation}
\end{lemma}

\begin{proof}
Put
\[
    G:=v/\sigma_v,
    \qquad
    F(w):=\mu(\Pi+\sigma_vw)/\sigma_v.
\]
Then \(G\sim N(0,I_n)\).
By \Cref{lem:geometry},
\(DF(G)=P=P_{\col(Z_{E(\Pi+v)})}\) is an orthogonal projection with
\(\operatorname{tr}(P)=\operatorname{tr}(P^2)=d\) almost surely.
The map \(F\) is one-Lipschitz, while \(w'F(w)\) and its gradient are
square integrable under \(N(0,I_n)\), so Theorem~2.2
of \citet{BellecZhang-2021} gives
\begin{equation}
    \Var(d)\le\bar d+\E[\lVert PG\rVert_2^2].
    \label{eq:rank-var-projection}
\end{equation}
For \(1\le r\le\rank(Z)\), let
\[
    \mathscr P_r
    :=\bigl\{
       P_{\col(Z_B)}:
       B\subseteq[p],\ |B|=\rank(Z_B)=r
      \bigr\}.
\]
Since \(Z\) is fixed, this family is deterministic. On \(\{d=r\}\),
an independent subset of \(E(\Pi+v)\) spans \(\col(Z_{E(\Pi+v)})\), so
\begin{equation}
    |\mathscr P_r|\le\binom{p}{r}\le(ep/r)^r,
    \qquad
    P\in\mathscr P_r\quad\text{on }\{d=r\}.
    \label{eq:rank-projection-family}
\end{equation}
For each fixed \(P_0\in\mathscr P_r\),
\(\lVert P_0G\rVert_2^2\sim\chi_r^2\). We apply
\citet[Lemma~1]{LaurentMassart-2000} with \(u=2r\log(ep/r)+r+t\) and
take a union bound. The threshold
\(r+2\sqrt{ru}+2u\) is then at most
\(11\{r\log(ep/r)+t\}\), and the resulting bound for each \(r\)
is \(|\mathscr P_r|e^{-u}\le e^{-r-t}\).
Since \(P=0\) on \(\{d=0\}\), summing over \(r\ge1\) gives
\begin{equation}
    \P\!\left\{\lVert PG\rVert_2^2
       >11[d\log\{ep/(d\vee1)\}+t]\right\}
    \le\frac{e^{-t}}{e-1},\qquad t\ge0.
    \label{eq:rank-projection-tail}
\end{equation}
Tail integration yields
\begin{equation}
    \E[\lVert PG\rVert_2^2]
    \le11\E[d\log\{ep/(d\vee1)\}]+\frac{11}{e-1}.
    \label{eq:rank-projection-moment}
\end{equation}
The function \(x\mapsto x\log\{ep/(x\vee1)\}\) is concave and dominates
\(x\) on \([0,p]\). Jensen's inequality and
\eqref{eq:rank-var-projection}--\eqref{eq:rank-projection-moment} give
\(\Var(d)\le12[1+\bar d\log\{ep/(\bar d\vee1)\}]\), which is
\eqref{eq:rank-variance-bound}.
\end{proof}
{}
\section{Sparse First-Stage Bounds}

Following \citet{BCCH-2012}, define the score event
\begin{equation}
    \mathcal G
    :=\left\{
       \left\lVert
          \En[z_iv_i]
       \right\rVert_\infty
       \le\frac{\lambda}{c_\lambda}
     \right\}.
    \label{eq:score-event}
\end{equation}
\begin{lemma}[LASSO Rates and Recovery of the Strong Coordinates]
\label{lem:primitive-lasso}
Under \Cref{assm:primitive}, for all sufficiently large \(n\), the
following hold on the event \(\mathcal G\) defined by
\eqref{eq:score-event}.  First, the prediction norm satisfies
\begin{equation}
    \En[\{z_i'(\widehat\pi-\pi^0)\}^2]^{1/2}
       \le
       2\En[\xi_i^2]^{1/2}
       +\frac{2(1+c_\lambda^{-1})\lambda\sqrt{s}}
             {\kappa_0}.
       \label{eq:prediction-bound}
\end{equation}
Next, for the measurable solution \(\widehat\pi\) of
\Cref{lem:independent-support},
\begin{equation}
    \lVert\widehat\pi-\pi^0\rVert_2
    \le
    \frac{2}{\sqrt{\underline\phi}}
    \left\{
       \En[\xi_i^2]^{1/2}
       +\frac{(1+c_\lambda^{-1})\lambda\sqrt{s}}
              {\kappa_0}
    \right\}.
    \label{eq:l2-bound}
\end{equation}
The effective dimension \(d\) defined in \eqref{eq:realized-rank}
satisfies
\begin{equation}
    d<\vartheta s+1.
    \label{eq:rank-upper-bound}
\end{equation}
\Needspace{4\baselineskip}
If \Cref{assm:strong-core} holds as well, then, again on
\(\mathcal G\),
\begin{equation}
    k\le d,
    \label{eq:rank-event-bound}
\end{equation}
and, for the measurable solution \(\widehat\pi\) above,
\begin{equation}
    A\subseteq\supp(\widehat\pi),
    \qquad
    \sign(\widehat\pi_j)
       =\sign(\pi_j^0),
    \quad j\in A.
    \label{eq:core-recovery}
\end{equation}
\end{lemma}

\begin{proof}
Define
\(\delta:=\widehat\pi-\pi^0\), \(a:=\En[\xi_i^2]^{1/2}\),
\(t:=\En[(z_i'\delta)^2]^{1/2}\), and \(\theta:=c_\lambda^{-1}\).
By \Cref{assm:primitive}(ii)--(iii),
\begin{equation}
    \frac{a}{\lambda\sqrt s}
    \le\frac{C_a}{c_\lambda\sqrt2\sigma_v
       \{\log(2p)+\alpha\log\log(p\vee n)\}^{1/2}}
    \longrightarrow0,
    \label{eq:approximation-penalty-ratio}
\end{equation}
where the denominator diverges because \(p\vee n\ge n\).

\emph{Prediction bound.}
The basic-inequality argument of \citet[Lemma~6]{BCCH-2012}, with
their penalty \(2n\lambda\) and unit loadings under our normalization,
compares the objective at \(\widehat\pi\) and \(\pi^0\). On
\(\mathcal G\), H\"older's and Cauchy--Schwarz inequalities and
\(\pi^0_{T^c}=0\) give
\begin{equation}
    \frac12t^2+(1-\theta)\lambda\lVert\delta_{T^c}\rVert_1
    \le(1+\theta)\lambda\lVert\delta_T\rVert_1+at.
    \label{eq:cone-inequality}
\end{equation}
If \(t\le2a\), \eqref{eq:prediction-bound} is immediate. Otherwise
\eqref{eq:cone-inequality} gives the restricted-eigenvalue cone
\(\lVert\delta_{T^c}\rVert_1\le
L_\lambda\lVert\delta_T\rVert_1\), so
\(\lVert\delta_T\rVert_1\le\sqrt s\,t/\kappa_0\) by
\Cref{assm:primitive}(iv). Substitution into the basic inequality gives
\[
    t\le2a+\frac{2(1+\theta)\lambda\sqrt s}{\kappa_0},
\]
which proves \eqref{eq:prediction-bound}.

\emph{Rank bound.}
Put \(E:=E(x)\), so that \(d=\rank(Z_E)\), and suppose that
\(d\ge\lceil\vartheta s\rceil\). Choose \(B\subseteq E\) with
\(|B|=\lceil\vartheta s\rceil\) and linearly independent columns. Since
\(\lceil\vartheta s\rceil\le d\le p\), \Cref{assm:primitive}(iv) applies at
\(\lceil\vartheta s\rceil\). The residual is \(u(x)=\xi+v-Z\delta\), so the
KKT equalities on \(B\), the score event, and that upper sparse-eigenvalue
bound give
\[
\begin{aligned}
    \lambda\sqrt{\lceil\vartheta s\rceil}
    &=\lVert(\En[z_iu_i(x)])_B\rVert_2
    \le\lVert(\En[z_iv_i])_B\rVert_2
      +\sup_{\lVert\gamma\rVert_2=1,\ \supp\gamma\subseteq B}
        \En[(z_i'\gamma)(\xi_i-z_i'\delta)]\\
    &\le\theta\lambda\sqrt{\lceil\vartheta s\rceil}
      +\sqrt{\overline\phi}\,(a+t).
\end{aligned}
\]
By \eqref{eq:prediction-bound}, \(a+t\le3a+2(1+\theta)\lambda\sqrt s/\kappa_0\),
and therefore
\[
    \frac{\lceil\vartheta s\rceil}{s}
    \le\frac{\overline\phi}{(1-\theta)^2}
       \left\{\frac{3a}{\lambda\sqrt s}+\frac{2(1+\theta)}{\kappa_0}\right\}^{\!2}
    <\vartheta
\]
for large \(n\), since \(a/(\lambda\sqrt s)\to0\) by
\eqref{eq:approximation-penalty-ratio} and
\(\vartheta>4L_\lambda^2\overline\phi/\kappa_0^2\). This contradicts
\(\lceil\vartheta s\rceil\ge\vartheta s\), so
\(d<\lceil\vartheta s\rceil<\vartheta s+1\), which proves
\eqref{eq:rank-upper-bound}.

\emph{Coefficient bound and strong core.}
Now use the measurable solution of \Cref{lem:independent-support} and
put \(\widehat T:=\supp(\widehat\pi)\).
By \Cref{lem:geometry,lem:independent-support},
\[
    |\widehat T|\le d,
    \qquad
    |\supp(\delta)|\le(s+d)\wedge p
       \le\lceil(\vartheta+1)s\rceil\wedge p.
\]
The lower sparse-eigenvalue bound in \Cref{assm:primitive}(iv) and
\eqref{eq:prediction-bound} therefore yield
\[
    \sqrt{\underline\phi}\,\lVert\delta\rVert_2
    \le t
    \le2a+\frac{2(1+\theta)\lambda\sqrt s}{\kappa_0},
\]
which also holds when \(\delta=0\) and proves \eqref{eq:l2-bound}.
If \Cref{assm:strong-core} holds, \Cref{assm:primitive}(ii) and the
beta-min threshold of \Cref{assm:strong-core} give
\[
    |\delta_j|\le\lVert\delta\rVert_2
    \le\underline b_n/2\le|\pi_j^0|/2,
    \qquad j\in A.
\]
Every \(j\in A\) is therefore selected with the correct sign, and
\(k\le|\widehat T|\le d\).  This proves
\eqref{eq:core-recovery} and \eqref{eq:rank-event-bound}.
\end{proof}

\begin{lemma}[Consequences of the First-Stage Conditions]
\label{lem:primitive-consequences}
Suppose \Cref{assm:gaussian,assm:primitive} hold, with \(J<\infty\)
fixed. Let \(\Phi\) denote the standard normal distribution function
and define
\begin{equation*}
    \eta_n
    :=2p\Phi\!\left(
       -\frac{\sqrt n\lambda}{c_\lambda\sigma_v}
    \right)
\end{equation*}
and
\begin{equation*}
    a_n
    :=\En[\xi_i^2]+s\lambda^2+\sqrt{\eta_n}.
\end{equation*}
Then
\begin{equation*}
    \eta_n\to0,
    \qquad
    a_n\to0,
\end{equation*}
and
\begin{gather}
    \Ebar[(\widehat\Pi_i-\Pi_i)^2]
       =O(a_n),
       \label{eq:prediction-moment}\\
    h-H=O(\sqrt{a_n}),
       \qquad
    \Ebar[\widehat\Pi_i^2]-H=O(\sqrt{a_n}),
       \quad
    \mathcal A=O(a_n).
       \label{eq:fitted-moment-rates}
\end{gather}
There are constants \(\underline h>0\) and \(C<\infty\), independent of
\(n\) and of the candidate, such that, for all sufficiently large \(n\),
\begin{equation}
    0<\underline h\le h,
    \qquad
    h+\Ebar[\widehat\Pi_i^2]\le C.
    \label{eq:population-bounds}
\end{equation}
\Needspace{5\baselineskip}
If \Cref{assm:strong-core} also holds, then, for all sufficiently large
\(n\),
\begin{equation}
    k(1-\eta_n)
    \le\bar d
    \le2(\vartheta s+1),
    \label{eq:mean-rank-bounds}
\end{equation}
and
\begin{equation}
    \begin{gathered}
        \bar d\asymp k,
        \qquad
        \bar d\to\infty,
        \qquad
        \frac{\Var(d)}{\bar d^2}\to0,\\
        \frac{d}{\bar d}\to_p1,
        \qquad
        \frac{\E[d^2+d]}{\bar d^2}\to1,
        \qquad
        \frac{\bar d}{\sqrt n}\to0.
    \end{gathered}
    \label{eq:primitive-rank-conclusions}
\end{equation}
Under the same additional Assumption~\ref{assm:strong-core},
\begin{equation}
    S
    \asymp \mathcal A+\frac{\E[d^2+d]}{n}
    \asymp \mathcal A+\frac{\bar d^2}{n},
    \qquad
    S\to0,
    \qquad
    nS\to\infty.
    \label{eq:score-rates}
\end{equation}
Moreover,
\begin{equation}
    \begin{aligned}
        \sqrt{\frac{\mathcal A}{n}}+\frac{\bar d}{n}
           \le C\sqrt{\frac{S}{n}},\quad
        \sqrt{\frac{\mathcal A}{n}}+\frac{\bar d}{n}
           \le\frac{CS}{\bar d}.
    \end{aligned}
    \label{eq:score-ratio-bounds}
\end{equation}
\end{lemma}

\begin{proof}
Column normalization, the Gaussian bound
\(\Phi(-u)\le\exp(-u^2/2)\) for \(u\ge0\), and
\Cref{assm:primitive}(iii) give
\[
    \P(\mathcal G^c)
    \le2p\Phi\!\left(
       -\frac{\sqrt n\lambda}{c_\lambda\sigma_v}
    \right)
    =\eta_n
    \le\{\log(p\vee n)\}^{-\alpha}\to0.
\]
Parts (ii) and (v) of \Cref{assm:primitive} then yield
\(a_n\to0\).

\emph{Prediction moments.}
Write
\(U(v):=\lVert\mu(\Pi+v)-\Pi\rVert_2/\sqrt n\), so
\(\E[U^2]=\Ebar[(\widehat\Pi_i-\Pi_i)^2]\).
By \Cref{lem:geometry}, \(U\) is \(n^{-1/2}\)-Lipschitz, so the
Gaussian Poincar\'e inequality \eqref{eq:gaussian-poincare} gives
\(\Var(U)\le\sigma_v^2/n\).  On \(\mathcal G\),
\Cref{lem:primitive-lasso} and the triangle inequality imply
\begin{equation}
    U\le\overline U
    :=3\En[\xi_i^2]^{1/2}
      +\frac{2(1+c_\lambda^{-1})\lambda\sqrt{s}}{\kappa_0}.
    \label{eq:fit-bound-event}
\end{equation}
Since \(\P(\mathcal G)\ge1/2\) eventually and, when
\(\E U>\overline U\), the event \(\mathcal G\) entails
\(|U-\E U|\ge\E U-\overline U\),
\[
    \P(\mathcal G)
    (\E U-\overline U)_+^2
    \le\Var(U).
\]
Consequently \(\E U\le\overline U+\sqrt{2\Var(U)}\), so that
\[
    \E[U^2]
    =\Var(U)+(\E U)^2
    \le2\overline U^2+5\sigma_v^2/n
    =O(a_n).
\]
Here \(n^{-1}=O(s\lambda^2)\) follows from
\Cref{assm:primitive}(iii) and \(s\ge1\).  This proves
\eqref{eq:prediction-moment}.

\emph{Fitted moments.}
By \eqref{eq:h-def}, Cauchy--Schwarz, and
\eqref{eq:prediction-moment},
\[
    |h-H|^2
    \le H\E[U^2]=O(a_n),
    \qquad
    0\le\mathcal A\le\E[U^2]=O(a_n),
\]
where the bound on \(\mathcal A\) follows by evaluating the minimum
in \eqref{eq:amse-approx} at coefficient one.  Moreover,
\[
    \Ebar[\widehat\Pi_i^2]-H
    =2(h-H)+\E[U^2]
    =O(\sqrt{a_n}).
\]
Since \(H\) is bounded and bounded away from zero by
\eqref{eq:strength-endogeneity}, these bounds prove
\eqref{eq:fitted-moment-rates} and \eqref{eq:population-bounds}.

\emph{Rank moments and ratios.}
Suppose \Cref{assm:strong-core} holds.  By
\eqref{eq:rank-event-bound},
\(\bar d\ge k\P(\mathcal G)
\ge k(1-\eta_n)\ge k/2\) eventually.
Using this lower bound in \Cref{lem:rank-variance} and
\eqref{eq:strong-core-growth} gives
\[
    V:=\frac{\Var(d)}{\bar d^2}
    \le C\left\{
       \frac1{k^2}
       +\frac{\log(2ep/k)}{k}
    \right\}\to0.
\]
Here \(d\le p\) implies \(\bar d\le p\), while
\(\bar d\ge k/2\) bounds
\(\log(ep/\bar d)\le\log(2ep/k)\).
Since \(d<\vartheta s+1\) on \(\mathcal G\),
\[
    (1-\eta_n)
    \left(1-\frac{\vartheta s+1}{\bar d}\right)_+^2
    \le V.
\]
If \(\bar d>2(\vartheta s+1)\) along a subsequence, the left side is
bounded away from zero, a contradiction. Therefore
\(\bar d\le2(\vartheta s+1)\) eventually, which proves
\eqref{eq:mean-rank-bounds}, and with \(s\le C_sk\) it also gives
\(\bar d\asymp k\to\infty\).
Chebyshev's inequality yields \(d/\bar d\to_p1\).  Finally,
\[
    \frac{\E[d^2+d]}{\bar d^2}
    =1+V+\frac1{\bar d}\to1,
    \qquad
    \frac{\bar d}{\sqrt n}
    \le\frac{2(\vartheta s+1)}{\sqrt n}\to0
\]
by \eqref{eq:rank-moments} and \Cref{assm:primitive}(v), which completes
the proof of \eqref{eq:primitive-rank-conclusions}.

\emph{Criterion rates.}
The closed form \eqref{eq:closed-form-criterion}, the bounds
\eqref{eq:population-bounds}, and fixed positive
\(\sigma_\varepsilon^2\) and \(|\sigma_{\varepsilon v}|\) give
\[
    S\asymp\mathcal A+\E[d^2+d]/n
       \asymp\mathcal A+\bar d^2/n.
\]
Here \(\E[d^2+d]/\bar d^2\to1\) by
\eqref{eq:primitive-rank-conclusions}, while
\(\sigma_\varepsilon^2>0\) by \Cref{assm:gaussian}(i) and
\(\sigma_{\varepsilon v}\ne0\) by
\eqref{eq:strength-endogeneity}. Since \(\mathcal A\to0\),
\(\bar d/\sqrt n\to0\), and \(\bar d\to\infty\) by
the preceding parts, these comparisons prove \eqref{eq:score-rates}.

For \eqref{eq:score-ratio-bounds}, the same comparison gives
\(\mathcal A+\bar d^2/n\le CS\), and
\[
    2\bar d\sqrt{\mathcal A/n}
    \le\mathcal A+\bar d^2/n.
\]
Taking square roots gives the first bound in
\eqref{eq:score-ratio-bounds}, and dividing by \(\bar d\) gives the second.
\end{proof}
{}
\section{Included Exogenous Controls}
\label{sec:controls}

We now allow a fixed number of exogenous controls in the structural
equation and every first stage. Partialling out the controls and
expressing the residuals in an orthonormal basis gives the model in
\Cref{sec:framework}, with sample size equal to the residual degrees of
freedom. Under conditional Gaussian sampling, the transformed errors
remain independent and identically distributed. If the residualized
design satisfies the identification and first-stage conditions, the
AMSE criterion applies with complexity measured only in the remaining
excluded-instrument directions.

\subsection{Model and First-Stage Estimation}

Let the original sample size be \(N\), and let
\(W\in\mathbb R^{N\times q}\) have rank \(q\), where \(q\) is fixed as
\(N\to\infty\). The columns of \(W\), which may include an intercept,
are included in both the structural equation and every first stage.
Write \(\mathsf Z_c\in\mathbb R^{N\times p_c}\) for candidate \(c\)'s
excluded-instrument dictionary and
\(\mathscr Z:=(W,\mathsf Z_1,\ldots,\mathsf Z_J)\). As in the main
text, we condition on the realized design \(\mathscr Z\), and take the
candidate list and penalties to be fixed under this conditioning.
Consider
\begin{equation}
    y^\circ=\beta x^\circ+W\gamma+\varepsilon^\circ,
    \qquad
    x^\circ=m+v^\circ,
    \qquad
    m:=\E[x^\circ\mid\mathscr Z].
    \label{eq:controls-raw-model}
\end{equation}
Conditional on \(\mathscr Z\), suppose the pairs
\(\{(\varepsilon_i^\circ,v_i^\circ)\}_{i\in[N]}\) are independent and
identically distributed \(N(0,\Omega)\), with \(\Omega\) as in
\Cref{assm:gaussian}(i). No linearity restriction is imposed on \(m\).

Define
\[
    P_W:=W(W'W)^{-1}W',\qquad
    M_W:=I_N-P_W,\qquad n:=N-q.
\]
Choose a \(\mathscr Z\)-measurable matrix
\(U\in\mathbb R^{N\times n}\) whose columns form an orthonormal basis
for \(\operatorname{col}(M_W)\). Thus
\begin{equation}
    U'U=I_n,\qquad UU'=M_W,\qquad U'W=0,
    \qquad (U'a)'(U'b)=a'M_Wb.
    \label{eq:controls-rotation}
\end{equation}
Set
\begin{equation}
    \begin{split}
        y&:=U'y^\circ,\qquad x:=U'x^\circ,\qquad \Pi:=U'm,\\
        \varepsilon&:=U'\varepsilon^\circ,\qquad
        v:=U'v^\circ,\qquad Z_c:=U'\mathsf Z_c.
    \end{split}
    \label{eq:controls-coordinates}
\end{equation}
These variables satisfy
\begin{equation}
    y=\beta x+\varepsilon,\qquad x=\Pi+v.
    \label{eq:controls-reduced-model}
\end{equation}
The residual vectors \(M_W\varepsilon^\circ\) and \(M_Wv^\circ\)
generally have dependent coordinates. In the orthonormal coordinates
\eqref{eq:controls-coordinates}, the error pairs are again independent
and identically distributed, as shown in the proof below.

Drop any dictionary column annihilated by \(M_W\), since its
contribution can be absorbed by the controls. Scale each remaining
column so that
\(\mathsf Z_{c,\cdot j}'M_W\mathsf Z_{c,\cdot j}/n=1\), and assign
\(\lambda_c>0\) after this scaling. The transformed dictionary then
satisfies \Cref{assm:primitive}(i). Estimate the first stage with the
controls left unpenalized:
\begin{equation}
    (\widehat\alpha_c,\widehat\pi_c^W)
    \in\argmin_{\alpha\in\mathbb R^q,\,\pi\in\mathbb R^{p_c}}
    \left\{
       \frac{\lVert x^\circ-W\alpha-\mathsf Z_c\pi\rVert_2^2}{2n}
       +\lambda_c\lVert\pi\rVert_1
    \right\}.
    \label{eq:controls-partial-lasso}
\end{equation}
Profiling out \(\alpha\) gives the LASSO problem in
\eqref{eq:lasso} with response \(x\) and dictionary \(Z_c\). In
particular, the residualized fitted instrument is uniquely defined by
\begin{equation}
    \widehat g_c:=M_W\mathsf Z_c\widehat\pi_c^W
       =U\widehat\Pi_c,
    \qquad \widehat\Pi_c:=\mu_c(x).
    \label{eq:controls-fitted-instrument}
\end{equation}
If the loss is instead normalized by \(1/(2N)\) with penalty
\(\lambda_c^\circ\) and the same scaled dictionary columns, the
equivalent penalty in \eqref{eq:controls-partial-lasso} is
\(\lambda_c=(N/n)\lambda_c^\circ\).

\subsection{Estimator and Feasible Criterion}

Let \(\widehat\beta_c^W\) be the coefficient on \(x^\circ\) in the IV
regression of \(y^\circ\) on \((x^\circ,W)\), using
\((\widehat g_c,W)\) as instruments. The full first-stage fit is
\(P_Wx^\circ+\widehat g_c\), and using it in place of
\(\widehat g_c\) gives the same estimator because \(W\) is included.
Partialling out \(W\) gives
\begin{equation}
    \widehat\beta_c^W
       =\frac{\widehat g_c'M_Wy^\circ}
              {\widehat g_c'M_Wx^\circ}
       =\frac{\widehat\Pi_c'y}{\widehat\Pi_c'x}
       =\widehat\beta_c
    \label{eq:controls-iv-equivalence}
\end{equation}
when the fit is nonzero. Set \(\widehat\beta_c^W=0\) when
\(\widehat g_c=0\), as in \eqref{eq:iv-estimator}.

Define the equicorrelation set and effective dimension in the original
coordinates by
\begin{align}
    E_c^W(x^\circ)
    &:=\left\{j\in[p_c]:
       \left|\frac{\mathsf Z_{c,\cdot j}'M_W
          (x^\circ-\mathsf Z_c\widehat\pi_c^W)}n\right|
           =\lambda_c\right\},\notag\\
    d_c^W&:=\rank(M_W\mathsf Z_{c,E_c^W(x^\circ)})=d_c.
    \label{eq:controls-effective-dimension}
\end{align}
The effective dimension counts only excluded-instrument directions
remaining after partialling out \(W\). The derivative of the full
first-stage fitted-value map has trace \(q+d_c^W\) almost surely,
while its degrees of freedom are \(q+\E[d_c^W]\). If an implementation
reports the realized divergence of the full fit, subtract \(q\) to
obtain \(d_c^W\). As in \Cref{sec:effective-dimension}, a count of
nonzero coefficients need not equal the rank in
\eqref{eq:controls-effective-dimension} when coefficient solutions
are nonunique.

Write \(\widehat h_c^W:=\widehat g_c'M_Wx^\circ/n\). Choose the same
pilot candidate \(\check c\in[J]\) as in the residual-coordinate
procedure and set \(\check\beta^W:=\widehat\beta_{\check c}^W\).
The index may be data dependent; for example, maximizing
\(\widehat h_c^W\) with ties broken by the smallest index gives the
same pilot in either set of coordinates. Define
\begin{align}
    \widehat\sigma_{\varepsilon,W}^2
    &:=\frac{(y^\circ-\check\beta^Wx^\circ)'M_W
                    (y^\circ-\check\beta^Wx^\circ)}n,\notag\\
    \widehat\sigma_{\varepsilon v,W}
    &:=\frac{(y^\circ-\check\beta^Wx^\circ)'M_Wx^\circ}n,
    \label{eq:controls-feasible-moments}
\end{align}
and
\begin{equation}
    \widehat S_c^W
    :=\frac{
       \widehat\sigma_{\varepsilon,W}^2\lVert\widehat g_c\rVert_2^2/n
       +\widehat\sigma_{\varepsilon v,W}^{\,2}
          \{(d_c^W)^2+d_c^W\}/n
    }{(\widehat h_c^W)^2}.
    \label{eq:controls-feasible-score}
\end{equation}
Set \(\widehat S_c^W:=+\infty\) if \(\widehat h_c^W=0\), and select
\begin{equation}
    \widehat c^W:=\min\left(\argmin_{c\in[J]}\widehat S_c^W\right).
    \label{eq:controls-selected-candidate}
\end{equation}
Thus \(\widehat c^W=1\) if every fit is zero. The quantities in
\eqref{eq:controls-feasible-moments}--\eqref{eq:controls-selected-candidate}
can be computed by residualizing the original data, without
constructing \(U\). The same control set must be used for every
candidate, so that the residual space and optimal instrument are
common to the comparison.

\subsection{AMSE and Candidate Selection}

To state the result in the original coordinates, define
\begin{equation}
    \begin{split}
        m_W&:=M_Wm,\qquad H_W:=\lVert m_W\rVert_2^2/n,\\
        h_c^W&:=\E[\widehat g_c'm_W]/n,\qquad
        \mathcal A_c^W:=\min_{a\in\mathbb R}
                    \E[\lVert\widehat g_c-am_W\rVert_2^2]/n.
    \end{split}
    \label{eq:controls-population-objects}
\end{equation}
These equal \(H\), \(h_c\), and \(\mathcal A_c\) in the transformed
model, so its infeasible criterion is
\begin{equation}
    S_c^W:=\frac{\sigma_\varepsilon^2\mathcal A_c^W}{(h_c^W)^2}
       +\frac{\sigma_{\varepsilon v}^2\E[(d_c^W)^2+d_c^W]}
                    {n(h_c^W)^2}.
    \label{eq:controls-leading-characterization}
\end{equation}

\Needspace{6\baselineskip}
\begin{corollary}[Fixed-Control Equivalence]
\label{cor:controls}
Suppose \(q\) and \(J\) are fixed, \(W\) has rank \(q\), and
\eqref{eq:controls-raw-model} satisfies the conditional Gaussian
sampling condition stated above. Suppose \(H_W\) and
\(\sigma_{\varepsilon v}\) satisfy \Cref{assm:gaussian}(ii), and the
transformed systems \((\Pi,Z_c,\lambda_c)\) satisfy
\Cref{assm:primitive,assm:strong-core} with \(n=N-q\). Then
\Cref{thm:main,prop:ratio,thm:feasible} hold for the controlled
procedure. In particular, there are random variables \(C_0^W\),
\(Q_c^W\), and \(r_c^W\) such that
\begin{equation}
    n(\widehat\beta_c^W-\beta)^2=C_0^W+Q_c^W+r_c^W,
    \qquad \bar Q_c^W:=\E[Q_c^W]=S_c^W\{1+o(1)\},
    \label{eq:controls-amse-definition}
\end{equation}
where \(C_0^W\) is common to all candidates and
\[
    \E[C_0^W]=\sigma_\varepsilon^2/H_W+O(n^{-1}),
    \qquad \max_c|r_c^W|/\bar Q_c^W\to_p0.
\]
The selected candidate satisfies
\begin{equation}
    \frac{\bar Q_{\widehat c^W}^W}{\min_{c\in[J]}\bar Q_c^W}
       \longrightarrow_p1.
    \label{eq:controls-oracle}
\end{equation}
The decomposition in \eqref{eq:controls-amse-definition} also holds
with \(N(\widehat\beta_c^W-\beta)^2\) on the left and a remainder
satisfying the same bound. The first-order expansion holds with either
\(\sqrt n\) or \(\sqrt N\) scaling.
\end{corollary}

The replacement of \(n\) by \(N\) uses the growing effective
dimension in \Cref{assm:strong-core}. All population quantities and criteria above are
normalized by \(n\).

\begin{proof}
\emph{Step 1: The transformed model.}
Conditional on \(\mathscr Z\), the matrix \(U\) is fixed. Hence
\((\varepsilon',v')'\) is Gaussian with mean zero and covariance
\[
    \begin{pmatrix}
        \sigma_\varepsilon^2 U'U&\sigma_{\varepsilon v}U'U\\
        \sigma_{\varepsilon v}U'U&\sigma_v^2U'U
    \end{pmatrix}
    =\begin{pmatrix}
        \sigma_\varepsilon^2 I_n&\sigma_{\varepsilon v}I_n\\
        \sigma_{\varepsilon v}I_n&\sigma_v^2I_n
    \end{pmatrix}.
\]
The coordinate pairs are therefore independent and identically
distributed \(N(0,\Omega)\). Since \(U'W=0\), premultiplying
\eqref{eq:controls-raw-model} by \(U'\) gives
\eqref{eq:controls-reduced-model}, and \(H=H_W\) by
\eqref{eq:controls-rotation}. Thus \Cref{assm:gaussian} holds for the
transformed model. The first-stage conditions do not depend on the
choice of \(U\): the dictionary Gram matrices and prediction norms
are determined by \(M_W\), while the coefficient vectors and penalties
are unchanged by an orthogonal change of basis.

\Needspace{5\baselineskip}
\emph{Step 2: First-stage fits and effective dimension.}
Fix a candidate and suppress its subscript. For each \(\pi\), least
squares over \(\alpha\) gives
\[
    \min_\alpha\lVert x^\circ-W\alpha-\mathsf Z\pi\rVert_2^2
       =\lVert M_W(x^\circ-\mathsf Z\pi)\rVert_2^2
       =\lVert x-Z\pi\rVert_2^2.
\]
The last equality follows from \(M_W=UU'\) and \(U'U=I_n\).
The profiled coefficient minimizers thus coincide with those in
\eqref{eq:lasso}. Uniqueness of the LASSO fitted vector gives
\(Z\widehat\pi^W=\widehat\Pi\), and
\(M_W\mathsf Z=UZ\) gives \(\widehat g=U\widehat\Pi\).
For each dictionary column,
\[
    \frac{\mathsf Z_{\cdot j}'M_W
       (x^\circ-\mathsf Z\widehat\pi^W)}n
       =\frac{Z_{\cdot j}'(x-Z\widehat\pi^W)}n.
\]
The equicorrelation sets therefore agree. Since \(U\) is injective,
\[
    \rank(M_W\mathsf Z_E)=\rank(UZ_E)=\rank(Z_E)=d,
\]
which proves \eqref{eq:controls-effective-dimension}.

The full first-stage fitted-value map is
\(w^\circ\mapsto P_Ww^\circ+U\mu(U'w^\circ)\). At the observed
response its derivative is \(P_W+UPU'\) almost surely, where
\(P=D\mu(x)\) as in \eqref{eq:projection-properties}. Its trace is
\[
    \operatorname{tr}(P_W+UPU')
       =q+\operatorname{tr}(PU'U)=q+d.
\]

\emph{Step 3: Equality of the IV estimators and criteria.}
The instrument moments for the controlled IV regression imply
\(\widehat g'M_W(y^\circ-\widehat\beta^Wx^\circ)=0\).
Using \(\widehat g=U\widehat\Pi\) and \(U'M_W=U'\) gives
\eqref{eq:controls-iv-equivalence} whenever the denominator is
nonzero. The LASSO KKT conditions give
\[
    \widehat h^W=\frac{\widehat\Pi'x}{n}
       =\frac{\lVert\widehat\Pi\rVert_2^2}{n}
          +\lambda\lVert\widehat\pi^W\rVert_1.
\]
Thus the denominator is positive when the fit is nonzero, and both
estimators are zero by convention otherwise. In particular,
\(\widehat\beta_c^W=\widehat\beta_c\) for every candidate, so the
pilot estimates also agree.

The identities \(\widehat g_c=U\widehat\Pi_c\), \(m_W=U\Pi\), and
\(U'U=I_n\) yield
\[
    H_W=H,\qquad h_c^W=h_c,\qquad
    \mathcal A_c^W=\mathcal A_c,\qquad
    \widehat h_c^W=\widehat h_c,\qquad
    \lVert\widehat g_c\rVert_2^2/n=\En[\widehat\Pi_{c,i}^2].
\]
Also, \(U'(y^\circ-\check\beta^Wx^\circ)=y-\check\beta x\), so
\eqref{eq:controls-feasible-moments} gives
\[
    \widehat\sigma_{\varepsilon,W}^2=\widehat\sigma_\varepsilon^2,
    \qquad
    \widehat\sigma_{\varepsilon v,W}=\widehat\sigma_{\varepsilon v}.
\]
Combining these equalities with \(d_c^W=d_c\), we obtain
\begin{equation}
    S_c^W=S_c,\qquad \widehat S_c^W=\widehat S_c,
    \qquad \widehat c^W=\widehat c.
    \label{eq:controls-score-equalities}
\end{equation}
These equalities include the zero-fit conventions. Since \(J\) is
fixed, \Cref{lem:first-order} gives \(\check\beta-\beta=O_p(n^{-1/2})\)
even when the pilot index is data dependent. Applying
\Cref{thm:main,prop:ratio,thm:feasible} to the transformed model
and taking \(C_0^W=C_0\), \(Q_c^W=Q_c\), and \(r_c^W=r_c\) gives
\eqref{eq:controls-amse-definition}, and
\eqref{eq:controls-score-equalities} gives
\eqref{eq:controls-oracle}.

\emph{Step 4: Scaling by the original sample size.}
Since \(q\) is fixed, \(N/n-1=q/n=O(n^{-1})\) and
\(\sqrt{N/n}-1=O(n^{-1})\). By \Cref{lem:first-order},
\(\max_c|\sqrt n(\widehat\beta_c^W-\beta)|=O_p(1)\). Consequently,
\[
    N(\widehat\beta_c^W-\beta)^2
       -n(\widehat\beta_c^W-\beta)^2
       =\frac qn\,n(\widehat\beta_c^W-\beta)^2=O_p(n^{-1})
\]
uniformly over candidates. By \eqref{eq:score-rates} and
\Cref{thm:main}(ii), \(\min_c n\bar Q_c^W\to\infty\). The added term
is therefore \(o_p(\bar Q_c^W)\), proving the stated squared-error
expansion. The difference between the \(\sqrt N\)- and
\(\sqrt n\)-scaled estimators is likewise \(O_p(n^{-1})\), so the
first-order expansion is unchanged.
\end{proof}
{}
\section{Simulation Implementation and Additional Error Design}
\label{sec:sim-diagnostics}
\markboth{Simulation Implementation}{}

\subsection{Implementation Details}
\label{sec:sim-implementation}

We use the credit and export data for the intensive-margin sample of
\citet{PRSW-2015}, with 14,208 firm--product--destination export flows
from 2,091 firms. We partial out product--destination fixed effects by
demeaning credit growth and every instrument within the 5,996
product--destination groups. The Original dictionary contains the ten
exposure measures used for LASSO by \citet{PRSW-2015}: the share of a
firm's 2006 credit from banks with more than 10\% foreign funding, an
indicator that this share is at least one half, and credit-weighted
averages of the first eight powers of its lenders' foreign-funding
shares. The Expanded dictionary adds pairwise products, alternative
exposure measures, threshold and bin indicators, cubic B-splines, and
interactions. An instrument-only screen removes two collinear or nearly
collinear columns; the resulting dictionary has 114 columns and
numerical rank 109. We center each column and scale it to unit empirical
root mean square.

The main first-stage signal is a least-squares refit on 17 variables
selected by LASSO from the Original instruments and their pairwise
products (55 variables in all). Its penalty minimizes ten-fold
cross-validation error averaged over five fold assignments, each
keeping observations from the same firm together. The dense signal is
the least-squares fit on all 55 variables, rescaled to match the main
signal's \(H\). Neither signal uses the export outcome. We divide both
signals by the 55-variable fit's residual standard deviation, 1.975,
to put them on the simulated unit-error scale.
We use \(H\in\{0.00605,0.0121,0.0242,0.0484,0.0968\}\), where
0.0484 is the main signal's value in the data. Both signals lie in the
span of the Original instruments and their pairwise products, which has
numerical rank 51. The iid-equivalent
expected \(F\) statistic \(1+nH/51\) takes the reported values 2.7,
4.4, 7.7, 14.5, and 28.0; it differs from the smaller firm-clustered
first-stage \(F\) statistic in the data. We set
\(\rho=c\widehat\rho\), where \(\widehat\rho=-0.218\) is the
correlation of the residuals from the just-identified IV regression of
log export volume on credit growth, using exposed-bank share as the
instrument, and the residuals from the main-signal fit. Section 4
reports \(c=0,0.25,\ldots,2.75\); the additional \(c=3\) cases remain
in the design-level outputs. Each design uses 1,000 replications with
paired underlying error draws across designs.

For each dictionary with \(p_c\) columns, the first-stage BCCH plug-in
penalty and its 37-multiplier grid are
\[
\begin{aligned}
  \lambda_c&=\frac{1.1\,\varkappa_c}{\sqrt n}
    \Phi^{-1}\!\left(1-\frac{\gamma_c}{2p_c}\right),
  &\gamma_c&=\frac{0.1}{\log(\max\{p_c,n\})},\\
  \varkappa_c&=2^a,\quad a=-15,-14.5,\ldots,3.
\end{aligned}
\]
Here \(\Phi\) is the standard normal distribution function. The
first-stage error standard deviation is known to be one in these
simulations, but the structural-error moments in the feasible criterion
are estimated using the pilot below. A rule's simulated risk is the
replication average of \(n(\widehat\beta-\beta)^2\). Its relative risk
divides that risk by the lowest risk among the four feasible rules in
the same design: the criterion, cross-validation, cross-fitting, and
BCCH. Known Moments is excluded from this denominator and can have
relative risk below one.

The cross-validation rule uses one prespecified, firm-grouped ten-fold
assignment. Each training fit uses the full-sample penalty level and
instruments rescaled to unit root mean square within the training
sample; average prediction error over held-out observations selects a
candidate, which is refit on the full sample. The cross-fitting rule
uses a separate prespecified, firm-grouped ten-fold assignment. For
each outer fold, it selects the dictionary and penalty by ten-fold
cross-validation using only the other nine folds, again splitting by
firm, computes plug-in
penalties using the size of that outer training sample, and fits the
selected candidate there. The held-out predictions serve as
instruments, and the IV numerator and denominator are pooled across
folds. Thus the selected candidate may differ across outer folds.
Cross-fitting's effective dimension averages fold-specific dimensions
with weights proportional to held-out fold size.

The pilot uses the cross-fitting rule's folds, with a BCCH-penalized
first stage on the Original dictionary in each fold and a pooled IV
numerator and denominator. If all ten BCCH fits are empty, it instead
uses in every fold the candidate minimizing pooled ten-fold held-out
prediction error. In this exceptional branch, the pooled choice can
depend on each fold's held-out observations, unlike the training-only
tuning of the cross-fitting rule. When the full-sample BCCH rule selects
no instruments, its reported IV estimate is replaced by the pilot
estimate. These fallbacks occur only rarely in the weakest
identification designs. The design-level outputs retain the
full-sample BCCH fallback frequencies.

Effective dimension is the rank of instruments with nonzero LASSO
coefficients; in our checks it equals the rank of the equicorrelation
set in \eqref{eq:realized-rank}. A candidate fit with a
Karush--Kuhn--Tucker violation exceeding \(10^{-3}\lambda_c\) is
unavailable to every rule. When full-sample BCCH selects no
instruments, its effective dimension is recorded as zero even though
its IV estimate is replaced by the pilot.

The common-scale Laplace errors are defined in
\Cref{sec:application-laplace}. The Gaussian designs use the same
underlying draws \(a_i\) and \(b_i\), and all other implementation
choices are shared. Known
Moments evaluates the same feasible criterion with the
structural-error moments set to their population values,
\(\sigma_\varepsilon^2=1\) and \(\sigma_{\varepsilon v}=\rho\), leaving
its other estimated inputs unchanged. Design-level simulation outputs
and the code used to generate the reported summaries are provided in
the replication materials.

\Needspace{3\baselineskip}
\subsection{Smooth-Transformation Errors}
\label{sec:sim-transformation}

For the smooth-transformation errors defined in
\Cref{sec:application-laplace}, we keep the first-stage signals,
candidate dictionaries, penalty grid, and implementation choices of
\Cref{sec:application-simulation,sec:sim-implementation} unchanged. We
apply \(Q\) to the same Gaussian draws used in the Gaussian designs
and use 1,000 replications for each of the 120 designs. As in the
Gaussian designs, we retain the \(\rho=0\) exogenous benchmarks.

\Cref{sec:application-laplace} and \Cref{tab:prsw-error-summary} summarize the
criterion's risk and complexity results. Cross-validation retains an
advantage near the calibrated correlation at \(F=7.7\) and \(F=14.5\),
while cross-fitting remains competitive under stronger identification.
The unnormalized risks of these two rules change by at most \(6\%\)
relative to their Gaussian counterparts. BCCH exhibits larger
differences at \(F=2.7\), where a few empty-selection replications
invoke the pilot fallback.
{}

\supplement
\supplementtitlepage[Supplemental Appendix: Non-Gaussian Errors]
\providecommand{\ngnorm}[1]{\lVert#1\rVert_2}
\providecommand{\ngop}{\mathrm{op}}
\providecommand{\ngtrace}{\operatorname{tr}}
\providecommand{\ngcol}{\operatorname{col}}
\providecommand{\ngDel}{\mathfrak D}
\section{Non-Gaussian Extension of the Main Results}
\label{ng:sec:results}

From the main text, we retain the model in \eqref{eq:model}, the
candidate LASSO fits and IV estimators in
\eqref{eq:lasso}--\eqref{eq:iv-estimator}, and the effective dimension in
\eqref{eq:realized-rank}. For the proofs below, we also consider fitting
the LASSO to an arbitrary response \(w\in\SR^n\). We write \(\mu_c(w)\)
for the fitted values from \eqref{eq:lasso} with \(x\) replaced by \(w\).
Let \(E_c(w)\) denote the corresponding equicorrelation set in
\eqref{eq:equicorrelation} and define
\(d_c(w):=\rank(Z_{c,E_c(w)})\). Thus \(\mu_c(x)=\widehat\Pi_c\),
\(E_c(x)=E_c\), and \(d_c(x)=d_c\). As in the main text, there is one
endogenous regressor, and all statements condition on the realized
instruments. This appendix does not provide a separate treatment of
exogenous controls. The expectations and probabilities below are taken
under the non-Gaussian error distribution specified here. In particular, the
maintained conditional-mean restrictions give
\(\E[\varepsilon_i]=\E[v_i]=0\) and \(\Pi_i=\E[x_i]\).

\subsection{Non-Gaussian Sampling Conditions}

The following conditions replace the main-text \Cref{assm:gaussian}.
Using the noise representation of
\citet{ChetverikovLiaoChernozhukov-2020-extended}, we allow the
first-stage error to be a smooth transformation of a standard Gaussian
variable. Many of the ideas behind the extended Stein identities and
Poincar\'e inequality in \Cref{ng:sec:identities} are also taken from
that paper.

\Needspace{5\baselineskip}
\begin{assumption}[Non-Gaussian Sampling]
\label{ng:ass:sampling}
Conditional on the instruments, the pairs \((\varepsilon_i,v_i)\) are
independent and identically distributed with a fixed distribution, and:
\begin{enumerate}[(i)]
\item For fixed \(a\ne0\), \(\sigma_u^2>0\), and \(C_u<\infty\),
\begin{equation}
 \varepsilon_i=av_i+u_i,\qquad
 \E[u_i\mid v_i]=0,\qquad
 \E[u_i^2\mid v_i]=\sigma_u^2,\qquad
 \E[|u_i|^4]\le C_u.
                                                        \label{ng:eq:conditional}
\end{equation}
\item There is a twice continuously differentiable function
\(\psi:\mathbb R\to\mathbb R\) such that
\begin{equation}
 v_i=\psi(e_i),\qquad e_i\sim N(0,1),
                                                        \label{ng:eq:noise-transform}
\end{equation}
with \(\E[\psi(e_i)]=0\). For fixed \(0<c_\psi\le C_\psi<\infty\),
\begin{equation}
 c_\psi\le\psi'(t)\le C_\psi\quad(t\in\mathbb R).
                                                        \label{ng:eq:transform-regularity}
\end{equation}
\item The strong-identification bound on
\(H=\En[\Pi_i^2]\) in the main-text \eqref{eq:strength-endogeneity} holds, and for a fixed
\(K_\Pi<\infty\),
\begin{equation}
 \max_{i\in[n]}|\Pi_i|\le K_\Pi.
                                                        \label{ng:eq:identification}
\end{equation}
\end{enumerate}
\end{assumption}

Part (i) specifies a linear conditional mean and a constant conditional
variance for the structural error given \(v_i\). Independence of \(u_i\)
and \(v_i\) is not required. The fourth-moment bound is used
to estimate the structural-error moments. Write
\(\sigma_v^2=\E[v_i^2]\),
\(\sigma_{\varepsilon v}=a\sigma_v^2\), and
\(\sigma_\varepsilon^2=a^2\sigma_v^2+\sigma_u^2\). Thus the error
covariance is nonsingular and \(\sigma_{\varepsilon v}\ne0\), as in the
main text. Part (ii) permits asymmetric, non-Gaussian first-stage errors.
The Gaussian specification is obtained by taking \(\psi(t)=\sigma_v t\). \Cref{ng:lem:kernel} derives the
tail, Stein-kernel, and Poincar\'e bounds used below from the bounds on
\(\psi'\), and \(\psi\in C^2\) makes the density and Stein kernel continuously
differentiable, as \Cref{ng:lem:stein} requires. The additional bound in
part (iii) controls projections weighted by the Stein kernel and the
third-moment term in \(\E[C_0]\).

The smooth-transformation construction in \Cref{sec:application-laplace} provides an
example satisfying \Cref{ng:ass:sampling}(i)--(ii) when
\(0<|\rho|<1\). In that construction, \(a=\rho\),
\(u_i=\sqrt{1-\rho^2}\,\psi(B_i)\), and
\(\sigma_u^2=1-\rho^2\). The transformation satisfies
\(7/(16\nu)\le\psi'(t)\le1/\nu\), with bounded second and third
derivatives.

\subsection{Changes to the First-Stage Conditions}

As in the main text, we use the approximate-sparsity representation
\eqref{eq:sparse-approximation}, the restricted eigenvalue in
\eqref{eq:restricted-eigenvalue}, the upper and lower sparse
eigenvalues defined immediately after it, and the constants of
\Cref{assm:primitive}. The next assumption
states which parts of that assumption we retain and which we replace.

\Needspace{5\baselineskip}
\begin{assumption}[Non-Gaussian First-Stage Conditions]
\label{ng:ass:design}
For every \(c\in[J]\), with constants common to all candidates:
\begin{enumerate}[(i)]
\item The column normalization in the main-text
\Cref{assm:primitive}(i) holds, and,
for a fixed \(K_Z<\infty\),
\(\max_{i,j}|z_{c,ij}|\le K_Z\).
\item The approximate-sparsity bound in the main-text
\Cref{assm:primitive}(ii) holds.
\item In place of the main-text \Cref{assm:primitive}(iii), the penalty satisfies
\begin{equation}
 \lambda_c\ge c_\lambda C_\psi
   \sqrt{\frac{2\{\log(2p_c)+20\log n\}}{n}}.
                                                        \label{ng:eq:penalty-lower}
\end{equation}
\item The restricted and sparse eigenvalue bounds in the main-text
\Cref{assm:primitive}(iv) hold.
\item In place of the main-text \Cref{assm:primitive}(v), the sparsity indices,
dimensions, and penalties satisfy
\begin{equation}
 \sqrt{s_c}\lambda_c\to0,\qquad
 \frac{\log(p_cn)}{s_c}\to0,\qquad
 \frac{s_c^2\log(p_cn)}{n}\to0.
                                                        \label{ng:eq:growth}
\end{equation}
\end{enumerate}
\end{assumption}

The bound on the dictionary entries in part (i) controls the weighted Gram
matrices in the Stein identities. We strengthen the penalty lower bound in part (iii) so that the empirical
score is dominated with probability at least \(1-n^{-20}\), which lets
us integrate the first-stage bounds over the exceptional event. Part (v) controls terms from the nonconstant Stein
kernel and implies \(s_c/\sqrt n\to0\), as required in the main text. We also impose
\Cref{assm:strong-core} as stated in the main text, using the
same \(C_a,\kappa_0,\underline\phi\) and threshold, strong-coordinate
set \(A_c\), and size \(k_c\). Exact support recovery is not needed.

\subsection{Extension of the Main Results}

Define \(h_c\), \(\mathcal A_c\), and \(S_c\) by
\eqref{eq:h-def}, \eqref{eq:amse-approx}, and
\eqref{eq:closed-form-criterion}, respectively, with moments under
\Cref{ng:ass:sampling}, and write \(\bar Q_c=\E[Q_c]\) for the random
variables \(Q_c\) below. We use the residual moment estimators, feasible
criterion, selected candidate, and common centering \(B_n\) in
\eqref{eq:nuisance-estimators}, \eqref{eq:feasible-score}, and
\eqref{eq:common-term}. Both criteria are therefore defined as in the main text,
but with population moments computed under
non-Gaussian errors. Since our results concern the AMSE, they do not
require the estimator's exact mean squared error to exist.

\Needspace{5\baselineskip}
\begin{theorem}[Infeasible AMSE Criterion]
\label{ng:thm:infeasible}
Suppose \Cref{ng:ass:sampling,ng:ass:design} and the main-text
\Cref{assm:strong-core} hold,
with \(J<\infty\) fixed. Then there exist random variables \(C_0\),
\(\{Q_c\}_{c\in[J]}\), and \(\{r_c\}_{c\in[J]}\) such that
\begin{enumerate}[(i)]
\item For every candidate,
\begin{equation}
 n(\widehat\beta_c-\beta)^2=C_0+Q_c+r_c,
                                                        \label{ng:eq:amse}
\end{equation}
where \(C_0\) does not depend on the candidate,
\(\E[C_0]=\sigma_\varepsilon^2/H+O(n^{-1})\), and the remainder
satisfies \(\max_{c\in[J]}|r_c|/\bar Q_c\to_p0\).
\item For every candidate, \(\bar Q_c>0\) and
\(\bar Q_c=S_c\{1+o(1)\}\).
\end{enumerate}
\end{theorem}

\begin{theorem}[Feasible AMSE Selection]
\label{ng:thm:feasible}
Suppose the hypotheses of \Cref{ng:thm:infeasible} hold and that the
pilot satisfies \(\check\beta-\beta=O_p(n^{-1/2})\).
\begin{enumerate}[(i)]
\item Up to the common centering \(B_n\), the feasible criterion
consistently estimates the infeasible criterion,
\begin{equation}
 \max_{c\in[J]}\frac{|\widehat S_c-B_n-S_c|}{S_c}\to_p0,
                                                        \label{ng:eq:selection-expansion}
\end{equation}
and thus the expectation of the candidate-specific term,
\begin{equation}
 \max_{c\in[J]}\frac{|\widehat S_c-B_n-\bar Q_c|}{\bar Q_c}\to_p0.
                                                        \label{ng:eq:selection-amse-expansion}
\end{equation}
\item The selected candidate \(\widehat c\) satisfies
\begin{equation}
 \frac{\bar Q_{\widehat c}}{\min_{c\in[J]}\bar Q_c}\to_p1.
                                                        \label{ng:eq:selection}
\end{equation}
\end{enumerate}
\end{theorem}

The pilot may be the IV-LASSO estimate from a possibly data-dependent
candidate in \([J]\) (\Cref{ng:lem:pilot}). Proofs of both theorems are in
\Cref{ng:sec:main-proofs}.
{}
\section{Notation and Preliminary Reductions}
\label{ng:sec:proof-notation}

We write \(C\) for a constant that depends only on the fixed constants
in the assumptions and whose value may change between displays. Since
\(J\) is fixed, a rate or limit established for each candidate holds
uniformly over \(c\in[J]\). As in the main appendix, we therefore fix a
candidate and suppress its subscript. We write \(T^0\) for the support
\(T_c\) in \eqref{eq:sparse-approximation}, reserving \(T\) for the
Stein-kernel matrix.

\paragraph{Squared-error decomposition.}
We use \(\Gamma\), \(\Delta\), \(L_n\), \(q\), \(\Lambda\), \(C_0\),
\(Q\), and \(r\) as defined in
\eqref{eq:leading-term}--\eqref{eq:decomposition-construction} in the
proof of the main-text \Cref{thm:main}, including its convention for
small \(n\) (with \(\underline H/2\) in place of \(\underline h\)), with
expectations now taken under \Cref{ng:ass:sampling}.

\paragraph{Positive denominator event.}
By \eqref{ng:eq:fit-moments}, \(h\to H\), and
\(\underline H\le H\le\overline H\) by the main-text bound
\eqref{eq:strength-endogeneity}, so that
\(\underline H/2\le h\le2\overline H\) for all sufficiently large \(n\).
We take \(n\) this large throughout, so that \(H\), \(h\), and their
reciprocals are bounded, and enlarge \(C\) where needed so that the
bounds below hold for all such \(n\).
Every function of \(v\) to which we apply the weighted Stein identities
or the weighted Poincar\'e inequality is locally Lipschitz, with value
and gradient of at most polynomial growth by \Cref{lem:geometry}, so
these results apply and, by \eqref{ng:eq:tails}, all expectations below
are finite. In place of the event used in the main appendix, we work on
\[
 \mathcal H_n:=\{\widehat h\ge\underline H/4\},
\]
on which the sample denominator is bounded away from zero. Since
\(\widehat h=0\) whenever \(\widehat\Pi=0\), \(\widehat\beta\) is the
ratio in the main-text \eqref{eq:iv-estimator} on \(\mathcal H_n\).
\Cref{ng:lem:denominator} shows that \(\P(\mathcal H_n)\to1\).

\paragraph{Population projection residual.}
With \(P\) as in \eqref{ng:eq:projection}, write
\begin{equation}
 m=\mu(\Pi+v)-\frac hH\Pi
                                                        \label{ng:eq:centered-fit-definition}
\end{equation}
for the residual from the population projection in the main-text
identity \eqref{eq:amse-approx}. This is the residual \(M\) of the main
appendix, written in the notation of \Cref{ng:lem:stein}. The definition
of \(h\), that identity, and \eqref{ng:eq:projection} give
\begin{equation}
 \E[\ngnorm m^2]=n\mathcal A,
 \qquad Dm=P\quad\text{almost everywhere}.
                                                        \label{ng:eq:centered-fit-properties}
\end{equation}
Let \(T\) and \(D_\tau=\ngtrace(TP)\) be the diagonal Stein-kernel
matrix and the kernel-weighted rank in \eqref{ng:eq:kernel}. Taking
\(b=(h/H)\Pi\) in \eqref{ng:eq:offset-stein} gives the equality in
\begin{equation}
 \E[(v'm-D_\tau)^2]=\E[m'Tm+\ngtrace(TPTP)]\le C(n\mathcal A+s),
                                                        \label{ng:eq:residual-stein}
\end{equation}
and \eqref{ng:eq:kernel-bounds} and \(P^2=P\), together with
\eqref{ng:eq:fit-moments} and \eqref{ng:eq:centered-fit-properties},
bound its two terms.
{}
\section{Proofs of the Main Results}
\label{ng:sec:main-proofs}

\subsection{Proof of Theorem~\ref{ng:thm:infeasible}}
\begin{proof}
The proof follows the four steps of the proof of the main-text
\Cref{thm:main}. \Cref{ng:lem:numerator} replaces
\eqref{eq:structural-second} and \Cref{lem:score-distance}. Similarly,
\Cref{ng:lem:denominator,ng:lem:directional} replace
\Cref{lem:denominator,lem:denominator-correction}, while
\Cref{ng:lem:oracle} replaces the Gaussian moments of \(L_n\) and
\(C_0\).

\emph{Step 1: \(\sqrt n(\widehat\beta-\beta)=\Lambda+o_p(S)\).}
Step~1 of the Gaussian proof uses \(\Delta=O_p(n^{-1/2})\),
\(\E[(\Gamma-L_n)^2]\le CS\), \(\E[L_n^2]=\sigma_\varepsilon^2/H\),
\(S\to0\), \(nS\to\infty\), and the fact that
\(1+\Delta=\widehat h/h\) is bounded away from zero on
\(\mathcal H_n\). The first five hold by
\Cref{ng:lem:denominator,ng:lem:numerator,ng:lem:oracle}, and the last
because \(\widehat h\ge\underline H/4\) on \(\mathcal H_n\) and
\(h\le2\overline H\). The argument of that step therefore proves the
claim and also shows that \(\Lambda=O_p(1)\) and
\(\sqrt n(\widehat\beta-\beta)=O_p(1)\).

\emph{Step 2: \(\E[\Gamma^2]=\sigma_\varepsilon^2/H+S+o(S)\).}
We divide \eqref{ng:eq:numerator} by \(nh^2\) and substitute
\(\E\ngnorm{\widehat\Pi}^2/n=\mathcal A+h^2/H\) from the
main-text identity \eqref{eq:amse-approx} to obtain
\[
 \E[\Gamma^2]
 =\frac{\sigma_\varepsilon^2}{H}
  +\frac{\sigma_\varepsilon^2\mathcal A}{h^2}
  +\frac{\sigma_{\varepsilon v}^2\E[d^2+d]}{nh^2}
  +\frac{o(n\mathcal A+s^2)}{nh^2}
 =\frac{\sigma_\varepsilon^2}{H}+S+o(S),
\]
where the last equality uses \(S\asymp\mathcal A+s^2/n\) from
\Cref{ng:lem:numerator}.

\emph{Step 3: \(\bar Q=S\{1+o(1)\}\).}
Expanding \(\Lambda^2\) as in Step~3 of the proof of the main-text
\Cref{thm:main},
\[
 \E[\Lambda^2]-\E[\Gamma^2]
 =-2\E[\Delta L_n^2]
  -2\E[\Delta L_n(\Gamma-L_n)]
  +\E[\Delta^2L_n^2].
\]
By \Cref{ng:lem:directional},
\(|\E[\Delta L_n^2]|\le C\{\sqrt{S/n}+1/n\}\) and
\(\E[\Delta^2L_n^2]\le C/n\). Cauchy--Schwarz and
\Cref{ng:lem:numerator} then bound the middle term by
\(C\sqrt{S/n}\).
Meanwhile, \(\E[C_0]=\sigma_\varepsilon^2/H+O(n^{-1})\) by
\Cref{ng:lem:oracle}. Combining these bounds with Step~2,
\[
\begin{aligned}
 |\bar Q-S|
 &\le\left|\E[\Lambda^2]-\E[\Gamma^2]\right|
  +\left|\E[\Gamma^2]-\frac{\sigma_\varepsilon^2}{H}-S\right|
  +\left|\frac{\sigma_\varepsilon^2}{H}-\E[C_0]\right|\\
 &\le C\left\{\sqrt{\frac{S}{n}}+\frac1n\right\}+o(S),
\end{aligned}
\]
which is \(o(S)\) because \(nS\to\infty\). Dividing by \(S\), we have
\(\bar Q/S\to1\), so that \(\bar Q>0\), which proves part~(ii).

\emph{Step 4: \(r/\bar Q\to_p0\).}
This follows from Steps~1 and~3 by the argument in Step~4 of the
proof of the main-text \Cref{thm:main}, and it proves part~(i).
\end{proof}
{}

\Needspace{9\baselineskip}
\subsection{Proof of Theorem~\ref{ng:thm:feasible}}
\begin{proof}
The proof follows the four steps of the proof of the main-text
\Cref{thm:feasible}, with Lemmas~\ref{ng:lem:fitted-moment}, \ref{ng:lem:pilot},
\ref{ng:lem:rank}, and~\ref{ng:lem:denominator} in place of the
main-text Lemmas~\ref{lem:fitted-moment-expansion} and~\ref{lem:pilot-rates},
\Cref{prop:ratio}, and \Cref{lem:denominator}, respectively. We work on
\(\mathcal H_n\), whose probability tends to one by
\Cref{ng:lem:denominator}.

\emph{Step 1:
\(\widehat\sigma_\varepsilon^2\En[\widehat\Pi_i^2]/\widehat h^2
=\sigma_\varepsilon^2\mathcal A/h^2+B_n+o_p(S)\).}
The decomposition \eqref{eq:sample-fit-decomposition} is an algebraic
identity, and the rest of Step~1 of the Gaussian proof uses only the
following facts. \Cref{ng:lem:fitted-moment} shows that
\(\widehat{\mathcal A}-\mathcal A\) and \(\widehat q-q\) are
\(o_p(S)\), while \(\widehat h-h\) and
\(\widehat\sigma_\varepsilon^2-\sigma_\varepsilon^2\) are
\(O_p(n^{-1/2})\) by \Cref{ng:lem:denominator,ng:lem:pilot}. By
independence, \(\E[q^2]=\sigma_v^2/(nH)\). Finally,
\(\mathcal A\le CS\) and \(nS\to\infty\) by
\Cref{ng:lem:numerator}. With these facts, the Gaussian argument
proves Step~1.

\emph{Step 2:
\(\widehat\sigma_{\varepsilon v}^{\,2}(d^2+d)/(n\widehat h^2)
=\sigma_{\varepsilon v}^2\E[d^2+d]/(nh^2)+o_p(S)\).}
Of the last three factors in \eqref{eq:rank-term-factorization}, the
first converges to one in probability by \Cref{ng:lem:pilot} and
\(\sigma_{\varepsilon v}=a\sigma_v^2\ne0\), the second by
\Cref{ng:lem:rank}, and the third by \Cref{ng:lem:denominator}. Since
the leading factor is at most \(S\), the claim follows.

\emph{Step 3: Part~(i).}
Adding Steps~1 and~2, we obtain \(\widehat S=B_n+S+o_p(S)\), which is
\eqref{ng:eq:selection-expansion}. Since \(\bar Q/S\to1\) by
\Cref{ng:thm:infeasible}, adding and subtracting \(S\) in
\(\widehat S-B_n-\bar Q\) yields \eqref{ng:eq:selection-amse-expansion}.

\emph{Step 4: Part~(ii).}
By \Cref{ng:thm:infeasible}, every \(\bar Q_c\) is positive. Define
\(\omega_n\) as in the proof of the main-text \Cref{thm:feasible}, with
\(\bar Q_c\) now taken under \Cref{ng:ass:sampling}. Then
\(\omega_n\to_p0\) by Step~3, and on \(\{\omega_n<1\}\) the comparison
leading to \eqref{eq:oracle-ratio-bound} gives
\[
 0\le\frac{\bar Q_{\widehat c}}{\min_c\bar Q_c}-1
   \le\frac{2\omega_n}{1-\omega_n}.
\]
Since \(\P(\omega_n<1)\to1\), the bound proves \eqref{ng:eq:selection}.
\end{proof}
{}
{}
\begingroup
\renewcommand{\sectionmark}[1]{\markboth{Non-Gaussian Auxiliary Results}{}}
\section{Conditional Moments, Geometry, and Stein Identities}
\label{ng:sec:identities}

\Needspace{7\baselineskip}
\begin{lemma}[Conditional Moments of the Structural Error]
\label{ng:lem:structural}
Under \Cref{ng:ass:sampling}, for any measurable vectors \(b(v),g(v)\)
for which the expectations exist,
\begin{align}
 \E[\varepsilon'b\mid v]=a v'b, \quad
 \E[(\varepsilon'b)(\varepsilon'g)\mid v]
   =a^2(v'b)(v'g)+\sigma_u^2 b'g.                       \label{ng:eq:structural}
\end{align}
\end{lemma}

\Needspace{3\baselineskip}
\begin{proof}
Independence of the pairs and \eqref{ng:eq:conditional} imply
\(\E[u\mid v]=0\) and \(\E[uu'\mid v]=\sigma_u^2 I_n\), so that
\(\E[\varepsilon\mid v]=av\) and
\(\E[\varepsilon\varepsilon'\mid v]=a^2vv'+\sigma_u^2 I_n\).
Multiplying by the \(v\)-measurable vectors \(b\) and \(g\) gives
\eqref{ng:eq:structural}.
\end{proof}

\Needspace{3\baselineskip}
\paragraph{LASSO geometry.}
We use the deterministic properties of the fit in \Cref{lem:geometry} and
the Borel choice of coefficients in \Cref{lem:independent-support}.
Define \(P(w):=P_{\ngcol(Z_{E(w)})}\) for every
\(w\). This map is Borel because the continuous residual determines
\(E(w)\) among finitely many subsets. Write \(P:=P(\Pi+v)\), \(d:=d(\Pi+v)\), and
\(\mu:=\mu(\Pi+v)\) when no argument is shown. Under
\Cref{ng:ass:sampling}, \(\Pi+v\) has a Lebesgue density, so
\eqref{eq:geometry-properties} gives, almost surely,
\begin{equation}
    D\mu(\Pi+v)=P,\qquad
    P=P'=P^2,\qquad P\mu=\mu,\qquad
    \ngtrace P=d.                                    \label{ng:eq:projection}
\end{equation}
Whenever coefficients are needed, \(\widehat\pi(w)\) denotes the
selection in \Cref{lem:independent-support}. Its support is contained
in \(E(w)\) by \eqref{eq:kkt-conditions}, and its selected columns are
linearly independent by that lemma, so its support size is at most
\(\rank(Z_{E(w)})=d(w)\).

Let \(p_v\) denote the density of \(v_i\), which exists and is positive
because \(\psi'>0\). Define the canonical Stein kernel and its diagonal
matrix by
\begin{equation}
 \begin{aligned}
 \tau(t)&=\frac{\int_t^\infty z p_v(z)\,dz}{p_v(t)},\quad
 T=\operatorname{diag}\{\tau(v_i)\},\quad
 W=T-\sigma_v^2I,\quad D_\tau=\ngtrace(TP).
 \end{aligned}
                                                        \label{ng:eq:kernel}
\end{equation}

\Needspace{7\baselineskip}
\begin{lemma}[Properties of the First-Stage Errors]
\label{ng:lem:kernel}
Under \Cref{ng:ass:sampling}, \(p_v\) is positive and
continuously differentiable, and \(\tau\) is continuously
differentiable. The kernel satisfies
\begin{equation}
 \begin{aligned}
 c_\psi^2\le\tau(t)\le C_\psi^2,\quad
 \E[\tau(v_i)]=\sigma_v^2,\quad
 (\tau p_v)'(t)=-t\,p_v(t).
 \end{aligned}                                        \label{ng:eq:kernel-bounds}
\end{equation}
For every \(b\in\mathbb R^n\) and \(r>0\),
\begin{equation}
 \begin{aligned}
 \E e^{b'v}\le e^{C_\psi^2\ngnorm b^2/2},\quad
 \P(|b'v|>r)\le2e^{-r^2/(2C_\psi^2\ngnorm b^2)}\quad(b\ne0).
 \end{aligned}                                        \label{ng:eq:tails}
\end{equation}
For locally Lipschitz \(f\) with square-integrable weak gradient,
\begin{equation}
 \Var(f(v))\le\E[\nabla f(v)'T\nabla f(v)]
                  \le C_\psi^2\E[\ngnorm{\nabla f(v)}^2].
                                                        \label{ng:eq:poincare}
\end{equation}
Finally, if \(f:\mathbb R\to\mathbb R\) is locally absolutely
continuous and
\(\E|v_i f(v_i)|+\E|\tau(v_i)f'(v_i)|<\infty\), then
\begin{equation}
 \E[v_i f(v_i)]=\E[\tau(v_i)f'(v_i)].
                                                        \label{ng:eq:scalar-stein}
\end{equation}
\end{lemma}

\begin{proof}
\emph{Step 1: Density and kernel bounds.}

Let \(\varphi\) denote the standard normal density. Since
\(c_\psi\le\psi'\le C_\psi\),
\[
 c_\psi|z|\le|\psi(z)-\psi(0)|\le C_\psi|z|,\qquad
 \psi(\mathbb R)=\mathbb R.
\]
The change-of-variable formula gives
\begin{equation}
 p_v(\psi(z))=\frac{\varphi(z)}{\psi'(z)}.
                                                        \label{ng:eq:transport-density}
\end{equation}
Since \(\psi\in C^2\) and \(\psi'>0\), \(p_v\) is positive and continuously
differentiable, and
\begin{equation}
 \frac{e^{-(t-\psi(0))^2/(2c_\psi^2)}}{C_\psi\sqrt{2\pi}}
 \le p_v(t)\le
 \frac{e^{-(t-\psi(0))^2/(2C_\psi^2)}}{c_\psi\sqrt{2\pi}}.
                                                        \label{ng:eq:density-envelope}
\end{equation}
In particular, \(\int(1+|t|^j)e^{\gamma t}p_v(t)\,dt<\infty\) for all
fixed \(j\ge0\) and \(\gamma\in\mathbb R\).

For independent standard Gaussian \(e,\widetilde e\), and
nondecreasing \(F,G\) with the following expectations finite,
\[
 \Cov(F(e),G(e))
 =\tfrac12\E[(F(e)-F(\widetilde e))(G(e)-G(\widetilde e))]
 \ge0.
\]
Apply this identity with \(G(e)=\ind_{\{e>z\}}\) and, in turn,
\(F(e)=\psi(e)-c_\psi e\) and \(F(e)=C_\psi e-\psi(e)\). Both choices of
\(F\) are centered and nondecreasing. Since
\(\E[e\ind_{\{e>z\}}]=\varphi(z)\), for every \(z\in\mathbb R\),
\begin{align*}
 c_\psi\varphi(z)
 \le\int_z^\infty\psi(t)\varphi(t)\,dt\le C_\psi\varphi(z),\quad
 \tau(\psi(z))
 =\frac{\psi'(z)}{\varphi(z)}\int_z^\infty\psi(t)\varphi(t)\,dt.
\end{align*}
These bounds give \(c_\psi^2\le\tau(\psi(z))\le C_\psi^2\) for all \(z\), as required.
The definition of \(\tau\) gives \((\tau p_v)'(t)=-t\,p_v(t)\). Since
\(p_v\) is positive and continuously differentiable, \(\tau\)
is continuously differentiable.

\emph{Step 2: First-order identity and tails.}

For compactly supported, locally absolutely continuous \(f\),
\begin{align*}
 \E[v_i f(v_i)]
 =-\int_{\mathbb R} f(t)(\tau p_v)'(t)\,dt=\int_{\mathbb R} f'(t)\tau(t)p_v(t)\,dt.
\end{align*}
For general \(f\) satisfying the stated integrability conditions,
choose smooth \(\chi_R\) equal to one on \([-R,R]\), zero outside
\([-2R,2R]\), and with \(|\chi_R'|\le C/R\). Applying the preceding
identity to \(f\chi_R\) gives
\begin{align*}
 \E[v_i f(v_i)\chi_R(v_i)]
 =\E[\tau(v_i)f'(v_i)\chi_R(v_i)]+\E[\tau(v_i)f(v_i)\chi_R'(v_i)].
\end{align*}
The last term satisfies
\[
 |\E[\tau(v_i)f(v_i)\chi_R'(v_i)]|
 \le\frac C{R^2}\E[|v_i f(v_i)|\ind_{\{|v_i|\ge R\}}]
 \longrightarrow0.
\]
Dominated convergence proves \eqref{ng:eq:scalar-stein}.
At \(f(t)=t\), it gives \(\E\tau(v_i)=\sigma_v^2\).
By the density envelope \eqref{ng:eq:density-envelope}, all
exponential moments are finite. Since \(\E[v_i]=0\) and
\(\lVert\tau\rVert_\infty\le C_\psi^2\), the Laplace-transform bound in
the proof of Theorem~14 of \citet{Saumard-2018-Stein}, applied to the
one-Lipschitz functions \(t\mapsto t\) and \(t\mapsto-t\), gives
\[
 \E e^{t v_i}\le e^{C_\psi^2t^2/2},\qquad t\in\mathbb R.
\]
Independence yields
\(\E e^{b'v}=\prod_i\E e^{b_iv_i}
 \le e^{C_\psi^2\ngnorm b^2/2}\).
The two-sided Chernoff bound gives \eqref{ng:eq:tails} for \(b\ne0\).
For \(b=0\), the tail probability is zero.

\emph{Step 3: Weighted Poincar\'e inequality.}

Since the density is positive on \(\mathbb R\) and \(\E|v_i|<\infty\),
Theorem~6 of \citet{Saumard-2018-Stein} gives
\[
 \Var(f(v_i))\le\E[\tau(v_i)f'(v_i)^2],
 \qquad f\in L^2(p_v)\text{ absolutely continuous}.
\]
For a smooth compactly supported function of \(n\) variables, let
\(v_{-1}=(v_2,\ldots,v_n)\). Conditional variance, induction, and
Jensen's inequality give
\begin{align*}
 \Var(f(v))
 &=\E[\Var(f(v)\mid v_{-1})]
       +\Var\{\E[f(v)\mid v_{-1}]\}\\*
 &\le\E[\tau(v_1)(\partial_1f)^2]
       +\sum_{j=2}^n\E[\tau(v_j)
                         \{\E[\partial_jf\mid v_{-1}]\}^2]\\*
 &\le\sum_{j=1}^n\E[\tau(v_j)(\partial_jf)^2]
 \le C_\psi^2\E[\ngnorm{\nabla f}^2].
\end{align*}
For locally Lipschitz \(f\) with square-integrable weak gradient, put
\(f_N=(-N)\vee(f\wedge N)\). Spatial cutoffs and local convolution, as in Step~3 of the proof of
\Cref{ng:lem:stein} below, extend the inequality to each bounded \(f_N\), with
\(\ngnorm{\nabla f_N}\le\ngnorm{\nabla f}\) almost everywhere, so that
\(\sup_N\Var(f_N)\le C_\psi^2\E\ngnorm{\nabla f}^2\).
On a fixed ball \(\mathcal U\) with positive probability, \(|f_N|\le\sup_{\mathcal U}|f|\), so
\[
 \P(v\in\mathcal U)(|\E f_N|-\sup_{\mathcal U}|f|)_+^2\le\Var(f_N).
\]
It follows that \(\sup_N\E[f_N^2]<\infty\), and Fatou's lemma gives
\(f\in L^2\). Passing to the limit in the truncations proves
\eqref{ng:eq:poincare}.
\end{proof}

For a locally Lipschitz vector field \(m\), set
\begin{equation}
 \ngDel(m)=v'm-\sum_i\tau(v_i)\partial_i m_i.
                                                        \label{ng:eq:divergence}
\end{equation}
Derivatives refer to weak derivatives, defined almost everywhere. The
identities below extend Stein's lemma to the first-stage errors of
\Cref{ng:ass:sampling}, following
\citet{ChetverikovLiaoChernozhukov-2020-extended}.

\Needspace{7\baselineskip}
\begin{lemma}[Weighted Stein Identities]
\label{ng:lem:stein}
Under \Cref{ng:ass:sampling}, let \(m,g\) and their first derivatives
have polynomial growth, with
\(m:\mathbb R^n\to\mathbb R^n\) and \(g:\mathbb R^n\to\mathbb R\)
locally Lipschitz. Then
\begin{align}
 \E[\ngDel(m)]&=0,\notag\\
 \E[\ngDel(m)g]&=\E[m'T\nabla g],                     \label{ng:eq:covariance}\\
 \E[\ngDel(m)^2]
   &=\E\{m'Tm+\ngtrace(TDm\,TDm)\}.                      \label{ng:eq:stein-second}
\end{align}
In particular, for \(m(v)=\mu(\Pi+v)-b\) and nonrandom \(b\),
\begin{equation}
 \begin{aligned}
 \E[v'm]=\E[D_\tau], \quad
 \E[(v'm-D_\tau)^2]
    =\E[m'Tm+\ngtrace(TPTP)].
 \end{aligned}                                        \label{ng:eq:offset-stein}
\end{equation}
\end{lemma}

\Needspace{8\baselineskip}
\begin{proof}
\emph{Step 1: Coordinate and covariance identities.}

Equation \eqref{ng:eq:scalar-stein}, conditional on the other
coordinates, gives
\[
 \E[v_i f(v)]=\E[\tau(v_i)\partial_i f(v)].
\]
By Fubini's theorem, the same identity holds for locally Lipschitz \(f\)
whenever both sides are absolutely integrable. For smooth compactly
supported \(m\) and continuously differentiable compactly supported
\(g\), the product rule gives
\begin{align*}
 \E[\ngDel(m)g]
 &=\sum_i\E[v_i m_i g-\tau(v_i)(\partial_i m_i)g]\\
 &=\sum_i\E[\tau(v_i)\partial_i(m_i g)
                    -\tau(v_i)(\partial_i m_i)g]\\
 &=\sum_i\E[\tau(v_i)m_i\partial_i g]
 =\E[m'T\nabla g].
\end{align*}
Applying the coordinate identity to each \(m_i\) also gives
\(\E\ngDel(m)=0\).

\emph{Step 2: Second moments for smooth fields.}

Take \(m\) smooth and compactly supported. Since \(\tau\) is
continuously differentiable, all terms involving \(\tau'\) are
integrable on this support. Moreover, \(\ngDel(m)\) is continuously
differentiable and compactly supported, so Step 1 applies with
\(g=\ngDel(m)\).
Write \(m_{j,i}=\partial_i m_j\), \(\tau_i=\tau(v_i)\), and
\(\tau_i'=\tau'(v_i)\).
By Step 1,
\begin{align*}
 \E[\ngDel(m)^2]
 &=\E\sum_i\tau_i m_i\partial_i
          \left[\sum_j\{v_jm_j-\tau_jm_{j,j}\}\right]\\
 &=\E\sum_i\tau_i m_i^2
      +\E\sum_{i,j}\tau_i m_i v_jm_{j,i}\\
 &\quad-\E\sum_i\tau_i\tau_i'm_i m_{i,i}
      -\E\sum_{i,j}\tau_i\tau_jm_i m_{j,ji}.
\end{align*}
Integration by parts in coordinate \(j\) gives
\begin{align*}
 \E\sum_{i,j}\tau_i m_i v_jm_{j,i}
 &=\E\sum_{i,j}\tau_j\partial_j(\tau_i m_i m_{j,i})\\
 &=\E\sum_{i,j}\tau_i\tau_jm_{i,j}m_{j,i}
      +\E\sum_i\tau_i\tau_i'm_i m_{i,i}\\
 &\quad+\E\sum_{i,j}\tau_i\tau_jm_i m_{j,ij}.
\end{align*}
Since \(m_{j,ij}=m_{j,ji}\), substitution yields
\[
 \E[\ngDel(m)^2]
 =\E\sum_i\tau_i m_i^2
       +\E\sum_{i,j}\tau_i\tau_jm_{i,j}m_{j,i}
 =\E\{m'Tm+\ngtrace(TDm\,TDm)\}.
\]
The terms involving \(\tau'\) cancel on each compact support, and the
approximation below requires only boundedness of \(\tau\).

\emph{Step 3: Extension to locally Lipschitz fields.}

Fix \(n\) and put
\[
 w_0(v)=1+\ngnorm v^2.
\]
For every fixed \(j\ge0\), \eqref{ng:eq:tails} gives
\begin{equation}
 \E[w_0(v)(1+\ngnorm v^j)]<\infty.
                                                        \label{ng:eq:approximation-weight}
\end{equation}
Choose smooth cutoffs \(\chi_R\) equal to one on
\(\{\ngnorm v\le R\}\), zero on \(\{\ngnorm v\ge2R\}\), and with
\(\ngnorm{\nabla\chi_R}\le C/R\). For a smooth compactly supported
approximate identity \(\omega_\delta\), set
\[
 m_R=\chi_Rm,\qquad m_{R,\delta}=\omega_\delta*m_R,
\]
and define \(g_R,g_{R,\delta}\) in the same way. On compact sets,
\(w_0\) is bounded and the density \(\prod_ip_v(v_i)\) of \(v\) is bounded
above and below by positive constants, so that, for fixed \(R\),
\begin{align*}
 &\E\!\left[w_0\ngnorm{m_{R,\delta}-m_R}^2
       +\lVert Dm_{R,\delta}-Dm_R\rVert_F^2\right]
 \longrightarrow0 \text{ as } \delta\downarrow0.
\end{align*}
Polynomial growth and \eqref{ng:eq:approximation-weight} give,
by dominated convergence, including the term
\(m(\nabla\chi_R)'\),
\begin{align*}
 &\E\!\left[w_0\ngnorm{m_R-m}^2
          +\lVert Dm_R-Dm\rVert_F^2\right]
 \longrightarrow0 \text{ as } R\to\infty.
\end{align*}
The same limits hold for \(g\) and its gradient. Choose a diagonal
sequence \(m_\ell,g_\ell\) of these smooth approximations. Write
\(\lVert\cdot\rVert_{L^2(w_0)}\) for the norm \(\E[w_0\lVert\cdot\rVert^2]^{1/2}\),
with Frobenius norms for matrices, and \(\lVert\cdot\rVert_{L^2}\) when
\(w_0\) is omitted. Then
\begin{align*}
 \lVert\ngDel(m_\ell)-\ngDel(m)\rVert_{L^2}
 &\le\lVert m_\ell-m\rVert_{L^2(w_0)}
      +C_\psi^2\sqrt n\lVert Dm_\ell-Dm\rVert_{L^2}
 \longrightarrow0.
\end{align*}
Boundedness of \(T\) and Cauchy--Schwarz also give
\begin{align*}
 &\E\left|\ngtrace(TDm_\ell\,TDm_\ell)-\ngtrace(TDm\,TDm)\right|\\
 &\qquad\le C\left(\lVert Dm_\ell\rVert_{L^2}
                 +\lVert Dm\rVert_{L^2}\right)
             \lVert Dm_\ell-Dm\rVert_{L^2}
 \longrightarrow0,\\
 \E[m_\ell'Tm_\ell]&\to\E[m'Tm], \quad
 \E[\ngDel(m_\ell)g_\ell]\to\E[\ngDel(m)g],\quad
 \E[m_\ell'T\nabla g_\ell]\to\E[m'T\nabla g].
\end{align*}
These limits prove the weighted identities.

For the offset fit, \Cref{lem:geometry} and
\eqref{ng:eq:projection} give
\[
 \ngnorm m\le\ngnorm{\Pi+v}+\ngnorm b,\qquad
 Dm=P\quad\text{almost everywhere},\qquad
 \lVert Dm\rVert_{\ngop}\le1.
\]
The product density assigns zero probability to the exceptional
Lebesgue-null set in \Cref{lem:geometry}. Substitution in
\eqref{ng:eq:covariance} with \(g\equiv1\) and in
\eqref{ng:eq:stein-second} proves \eqref{ng:eq:offset-stein}.
\end{proof}
{}
\endgroup
\begingroup
\renewcommand{\sectionmark}[1]{\markboth{Sparse Fits and Effective Dimension}{}}
\section{Sparse Fits, Uniform Projections, and Effective Dimension}
\label{ng:sec:sparse}

Let
\begin{equation}
 \mathcal G=\{\lVert Z'v/n\rVert_\infty\le\lambda/c_\lambda\}.
                                                        \label{ng:eq:score-event}
\end{equation}

\Needspace{7\baselineskip}
\begin{lemma}[LASSO Rates under Non-Gaussian Errors]
\label{ng:lem:sparse}
Under \Cref{ng:ass:sampling,ng:ass:design},
\(\P(\mathcal G^c)\le n^{-20}\). On \(\mathcal G\), eventually,
\begin{align}
 \ngnorm{Z(\widehat\pi-\pi^0)}/\sqrt n
 &\le2C_a\sqrt{s/n}
      +2(1+c_\lambda^{-1})\lambda\sqrt{s}/\kappa_0,\notag\\
 d&<\lceil\vartheta s\rceil,\qquad
 \ngnorm{\widehat\pi-\pi^0}\le\underline b_n/2.
                                                        \label{ng:eq:sparse-bounds}
\end{align}
The coefficient bound uses the Borel solution in
\Cref{lem:independent-support}. Under \Cref{assm:strong-core},
\(A\subseteq\supp(\widehat\pi)\) with correct signs, and
\(k\le d\) on this event. Furthermore, with
\(a_n=s\lambda^2+s/n\),
\begin{equation}
 \begin{aligned}
 \E[\ngnorm{\widehat\Pi-\Pi}^2]/n&\le Ca_n,\quad
 |h-H|\le C\sqrt{a_n},\\
 0\le\mathcal A&\le Ca_n=o(1),\quad
 \E[d^r]\le Cs^r\quad(1\le r\le8).
 \end{aligned}                                        \label{ng:eq:fit-moments}
\end{equation}
In particular, \(h\) is bounded away from zero eventually and
\(\E[d]\asymp s\) under \Cref{assm:strong-core}.
\end{lemma}

\begin{proof}
By normalization, \eqref{ng:eq:tails},
and \eqref{ng:eq:penalty-lower},
\[
 \P(\mathcal G^c)
 \le2p\exp\{-n\lambda^2/(2c_\lambda^2C_\psi^2)\}
 \le n^{-20}.
\]
\Cref{ng:ass:design}(ii) and \eqref{ng:eq:penalty-lower} give
\begin{equation}
 \frac{\ngnorm\xi}{\sqrt n\,\lambda\sqrt s}
 \le\frac{C_a}{c_\lambda C_\psi
       \sqrt{2\{\log(2p)+20\log n\}}}
 \to0.
                                      \label{ng:eq:approximation-penalty-ratio}
\end{equation}
On \(\mathcal G\), \eqref{ng:eq:sparse-bounds} follows from the
deterministic arguments in the proof of \Cref{lem:primitive-lasso}, since the
approximation, restricted eigenvalue, and sparse eigenvalue bounds used
there hold here by \Cref{ng:ass:design}(ii),(iv) and
\eqref{ng:eq:approximation-penalty-ratio}. Under
\Cref{assm:strong-core}, the same proof gives the asserted inclusion,
signs, and \(k\le d\).

For the moments, let
\(U(v)=\ngnorm{\mu(\Pi+v)-\Pi}/\sqrt n\). By
\Cref{lem:geometry}, \(U\) is \(n^{-1/2}\)-Lipschitz, so
\eqref{ng:eq:poincare} gives
\(\Var(U)\le C_\psi^2/n\). On \(\mathcal G\), \eqref{ng:eq:sparse-bounds}
and \(\ngnorm\xi/\sqrt n\le C_a\sqrt{s/n}\) give
\(U\le\overline U:=3C_a\sqrt{s/n}+2(1+c_\lambda^{-1})\lambda\sqrt s/\kappa_0\), so that
\[
 \P(\mathcal G)(\E U-\overline U)_+^2\le\Var(U),
 \qquad
 \E[U^2]\le2\overline U^2+5C_\psi^2/n
 \le Ca_n,
\]
where \(\P(\mathcal G)\ge1/2\) eventually and \(s\ge1\).
By the definitions of \(h,\mathcal A\) and Cauchy--Schwarz,
\[
 |h-H|^2\le H\E[U^2]\le Ca_n,
 \qquad 0\le\mathcal A\le\E[U^2]\le Ca_n=o(1),
\]
where the limit follows from \eqref{ng:eq:growth}.

We will also use the following bound in the next section. For \(1\le r\le4\),
\eqref{ng:eq:tails} gives \(\E[(n+\ngnorm v^2)^{2r}]\le C_rn^{2r}\).
Therefore, if \(\P(\mathcal B)\le Cn^{-20}\), Cauchy--Schwarz gives
\begin{equation}
 \E[(n+\ngnorm v^2)^r\ind_{\mathcal B}]
 \le C_r n^{r-10}.                                    \label{ng:eq:exceptional}
\end{equation}
Finally, \(d=\rank(Z_E)\le n\), so for \(1\le r\le8\),
\[
 \E[d^r]\le(\vartheta s+1)^r+n^r\P(\mathcal G^c)\le Cs^r.
\]
Under \Cref{assm:strong-core},
\[
 k(1-n^{-20})\le\E d\le Cs,\qquad
 k\le s\le C_sk,
\]
so \(\E d\asymp s\). Also \(h\ge H-C\sqrt{a_n}\) is bounded
away from zero eventually.
\end{proof}

\Needspace{5\baselineskip}
For a deterministic subset \(B\subseteq[p]\), let \(P_B\) be
projection onto \(\ngcol(Z_B)\), whether or not its columns are
independent.

\begin{lemma}[Projections onto Sparse Sets of Columns]
\label{ng:lem:uniform}
Under \Cref{ng:ass:sampling,ng:ass:design}, set
\(m^*=\lceil(\vartheta+1)s\rceil\wedge p\). For
\(\zeta=v\) and for \(\zeta=W\Pi\), and every fixed \(r\in[1,8]\),
\begin{equation}
 \E\left[\max_{|B|\le m^*}\ngnorm{P_B\zeta}^{\,r}\right]
       \le C_r\{s\log(pn)\}^{r/2}.                           \label{ng:eq:uniform-projections}
\end{equation}
There are a fixed \(C_\eta<\infty\) and an event \(\mathcal W\), with
\(\P(\mathcal W^c)\le n^{-20}\), on which, for every
\(|B|\le m^*\),
\begin{equation}
 \lVert P_BWP_B\rVert_{\ngop}\le\eta,
 \qquad \eta:=C_\eta s\sqrt{\log(pn)/n}=o(1).                 \label{ng:eq:restricted-gram}
\end{equation}
Consequently, for the random projection \(P\),
\begin{align}
 \E[\ngnorm{Pv}^{\,r}]+\E[\ngnorm{PW\Pi}^{\,r}]
       &\le C_r\{s\log(pn)\}^{r/2},\qquad 1\le r\le8,
                                                        \label{ng:eq:adaptive-projections}\\
 \E[(D_\tau-\sigma_v^2d)^2]&\le C\eta^2s^2+Cn^{-18}.
                                                        \label{ng:eq:weighted-rank}
\end{align}
\end{lemma}

\begin{proof}
\emph{Step 1: Uniform projection moments.}
For \(\zeta=v\), \eqref{ng:eq:tails} gives a centered sub-Gaussian
bound on the moment generating function. The coordinates of \(\zeta=W\Pi\)
are independent, centered by \eqref{ng:eq:kernel-bounds}, and
bounded by \eqref{ng:eq:kernel-bounds} and \Cref{ng:ass:sampling}(iii), so Hoeffding's lemma
gives the same bound with a possibly larger fixed constant. In both
cases, for every \(b\in\mathbb R^n\),
\[
 \E e^{b'\zeta}\le\exp\{K_{\rm sg}^2\ngnorm b^2/2\},
 \qquad K_{\rm sg}<\infty.
\]
For deterministic \(B\) with \(r_B:=\rank(P_B)>0\), apply
\citet[Theorem~2.1]{HsuKakadeZhang-2012} with \(P_B\) in the role of
their matrix \(A\), so that
\(\Sigma=P_B\), and with zero mean:
\[
 \P\!\left\{\ngnorm{P_B\zeta}^2>
 K_{\rm sg}^2\{r_B+2\sqrt{r_Bt}+2t\}\right\}\le e^{-t},
 \qquad t>0.
\]
The rank-zero projection is identically zero. Since \(r_B\le|B|\)
and \(\sum_{j=0}^\ell\binom pj\le(p+1)^\ell\), a union bound with
\(t=\ell\log(p+1)+u\) gives, for \(1\le\ell\le p\),
\[
 \P\!\left\{\max_{|B|\le\ell}\ngnorm{P_B\zeta}
      >C\sqrt{\ell\log(2p)+u}\right\}\le e^{-u},
 \qquad u\ge0.
\]
Integrating this tail for fixed \(r\in[1,8]\) and using
\(\ell=m^*\le Cs\) and \(\log(2p)\le C\log(pn)\), we obtain
\eqref{ng:eq:uniform-projections}.

\emph{Step 2: Restricted weighted Gram matrix.}
For fixed \(j\) and \(k\), the variables \(z_{ij}z_{ik}W_{ii}\), \(i\in[n]\), are
independent, centered, and uniformly bounded by \Cref{ng:ass:design}(i) and
\eqref{ng:eq:kernel-bounds}. Hoeffding's inequality and a union bound
over the pairs \((j,k)\) give an event \(\mathcal W\), with \(\P(\mathcal W^c)\le n^{-20}\), on
which
\[
 \max_{j,k}|(Z'WZ/n)_{jk}|\le C\sqrt{\log(pn)/n}.
\]
For \(g=Z_Bb\) with \(|B|\le m^*\), the lower sparse-eigenvalue
bound in \Cref{ng:ass:design}(iv) therefore yields
\[
 |g'Wg|\le n\max_{j,k}|(Z'WZ/n)_{jk}|\lVert b\rVert_1^2
 \le C|B|\sqrt{\log(pn)/n}\,\ngnorm g^2.
\]
Symmetry of \(W\) implies \(\lVert P_BWP_B\rVert_{\ngop}\le\eta\),
including \(B=\varnothing\), and \(\eta\to0\) by
\eqref{ng:eq:growth}. This proves \eqref{ng:eq:restricted-gram}.

\emph{Step 3: Selected projections.}
On \(\mathcal G\), an independent basis \(B\subseteq E(x)\) has
\(|B|=d<\lceil\vartheta s\rceil\) and \(|B|\le m^*\), so \(P=P_B\).
On \(\mathcal G^c\), use \(\ngnorm{Pv}\le\ngnorm v\) and
\(\ngnorm{PW\Pi}\le C\sqrt n\). Cauchy--Schwarz and
\eqref{ng:eq:tails} then give, for either \(\zeta=v,W\Pi\),
\[
 \E[\ngnorm{P\zeta}^{\,r}]
 \le\E\max_{|B|\le m^*}\ngnorm{P_B\zeta}^{\,r}
       +O(n^{r/2-10})\le C_r\{s\log(pn)\}^{r/2},\quad1\le r\le8.
\]
On \(\mathcal G\cap\mathcal W\),
\(|D_\tau-\sigma_v^2d|=|\ngtrace(PWP)|\le\eta d\le Cs\eta\).
Since \(D_\tau+d\le Cn\) everywhere,
\[
 \E[(D_\tau-\sigma_v^2d)^2]
 \le C\eta^2s^2+Cn^2\{\P(\mathcal G^c)+\P(\mathcal W^c)\}
 \le C\eta^2s^2+Cn^{-18}.
\]
These are \eqref{ng:eq:adaptive-projections}--\eqref{ng:eq:weighted-rank}.
\end{proof}

\Needspace{7\baselineskip}
\begin{lemma}[Concentration of the Effective Dimension]
\label{ng:lem:rank}
Under the hypotheses of \Cref{ng:thm:infeasible},
\begin{equation}
 \begin{aligned}
 \Var(d)&\le C\{s\log(pn)+\eta^2s^2+n^{-18}\}=o(s^2),\\
 \E[d^2+d]&=\E[d]^2\{1+o(1)\}\asymp s^2.
 \end{aligned}
                                                        \label{ng:eq:rank-variance}
\end{equation}
Consequently,
\begin{equation}
 \begin{aligned}
 \frac{d}{\E[d]}\to_p1,\quad
 \frac{d^2+d}{\E[d^2+d]}\to_p1.
 \end{aligned}                                        \label{ng:eq:rank-ratios}
\end{equation}
\end{lemma}

\Needspace{4\baselineskip}
\begin{proof}
Let \(m_0(v)=\mu(\Pi+v)-\E[\mu(\Pi+v)]\). Applying
\eqref{ng:eq:poincare} to each coordinate, \eqref{ng:eq:projection}
and \Cref{ng:lem:sparse} give
\[
 \E\ngnorm{m_0}^2\le\E[\ngtrace(PTP)]
 =\E[D_\tau]\le C\E d\le Cs.
\]
For \(F_0=v'm_0\), \(\nabla F_0=m_0+Pv\) almost everywhere, so
\eqref{ng:eq:poincare}, \eqref{ng:eq:adaptive-projections},
and \eqref{ng:eq:offset-stein} give
\[
 \Var(F_0)\le Cs\log(pn),\qquad
 \E[(F_0-D_\tau)^2]
 =\E[m_0'Tm_0+\ngtrace(TPTP)]\le Cs.
\]
These bounds give \(\Var(D_\tau)\le Cs\log(pn)\), and
\eqref{ng:eq:weighted-rank} yields
\[
 \sigma_v^4\Var(d)
 \le2\Var(D_\tau)+2\E[(D_\tau-\sigma_v^2d)^2]
 \le C\{s\log(pn)+\eta^2s^2+n^{-18}\}.
\]
Since \(\E d\asymp s\ge k\to\infty\) by \Cref{ng:lem:sparse} and
\Cref{assm:strong-core}, \(\log(pn)/s\to0\), and \(\eta\to0\), this is
\eqref{ng:eq:rank-variance}. Chebyshev's inequality gives
\(d/\E d\to_p1\). Moreover,
\[
 \frac{\E[d^2+d]}{(\E d)^2}
 =1+\frac{\Var(d)}{(\E d)^2}+\frac1{\E d}\to1,
\]
and \((d^2+d)/(\E d)^2\to_p1\), which proves \eqref{ng:eq:rank-ratios}.
\end{proof}
{}
\endgroup
\section{The Second Moment of the Numerator}
\label{ng:sec:numerator}

\Needspace{7\baselineskip}
\begin{lemma}[Second Moment of the Numerator]
\label{ng:lem:numerator}
Under the hypotheses of \Cref{ng:thm:infeasible},
\begin{align}
 \E[(\varepsilon'\widehat\Pi)^2]
 &=\sigma_\varepsilon^2\E[\ngnorm{\widehat\Pi}^2]
      +\sigma_{\varepsilon v}^2\E[d^2+d]
      +o(n\mathcal A+s^2),                         \label{ng:eq:numerator}\\
 \E[(\Gamma-L_n)^2]&\le CS.
                                                        \label{ng:eq:gamma}
\end{align}
Also, \(S\asymp\mathcal A+s^2/n=o(1)\) and
\(nS\to\infty\).
\end{lemma}

\begin{proof}
By \Cref{ng:lem:sparse,ng:lem:rank}, the main-text
\eqref{eq:closed-form-criterion}, and \(a\ne0\),
\[
 \mathcal R:=n\mathcal A+s^2\asymp nS,\qquad
 S=o(1),\qquad s^2\le CnS.
\]
Since \(\log(pn)\ge\log n\), \eqref{ng:eq:growth} forces
\(s\to\infty\), so that \(nS\ge s^2/C\to\infty\). With \(m\) as in
\eqref{ng:eq:centered-fit-definition}, replacing \(\Pi\) in the
centering term \((h/H)\Pi\) by its sparse approximation \(Z\pi^0\)
defines the auxiliary fit
\(\tilde m=m+(h/H)\xi=\mu(\Pi+v)-(h/H)Z\pi^0\). By \Cref{ng:ass:design}(ii)
and \eqref{ng:eq:centered-fit-properties},
\begin{equation}
 \E\ngnorm{\tilde m}^2
 \le2\E\ngnorm m^2+2\frac{h^2}{H^2}\ngnorm\xi^2
 \le C(n\mathcal A+s).
                                                        \label{ng:eq:centered-auxiliary-bound}
\end{equation}

\emph{Step 1: the decomposition \eqref{ng:eq:five-corrections}.}
Put \(F=v'\mu\) and \(\delta=F-D_\tau\). The weighted
covariance and second-moment identities
\eqref{ng:eq:covariance} and \eqref{ng:eq:stein-second}, together with
\(\nabla F=\mu+Pv\), give
\[
 \E[\delta F]=\E[\mu'T\mu+\mu'TPv],\qquad
 \E[\delta^2]=\E[\mu'T\mu+\ngtrace(TPTP)].
\]
Using \(T=W+\sigma_v^2I\), \(P\mu=\mu\), and
\(\E[\mu'v]=\E[D_\tau]\) by \eqref{ng:eq:offset-stein}, we obtain
\begin{align*}
 \E[2\delta F]
 &=2\sigma_v^2\E[\ngnorm\mu^2]+2\E[\mu'W\mu]
       +2\sigma_v^2\E[D_\tau]+2\E[\mu'WPv],\\
 \E[\delta^2]
 &=\sigma_v^2\E[\ngnorm\mu^2]+\E[\mu'W\mu]
       +\E[\ngtrace(TPTP)].
\end{align*}
Since \(D_\tau=F-\delta\), squaring and rearranging gives
\(F^2=D_\tau^2+2\delta F-\delta^2\). Substituting the preceding
expressions, subtracting
\(\sigma_v^2\E[\ngnorm\mu^2]+\sigma_v^4\E[d^2+d]\), and collecting
terms yields
\begin{align}
 \E[F^2]-\sigma_v^2\E[\ngnorm\mu^2]-\sigma_v^4\E[d^2+d]
 &=\E[\mu'W\mu]+\E[D_\tau^2-\sigma_v^4d^2]\notag\\
 &\quad-\E[\ngtrace(TPTP)-\sigma_v^4d]
       +2\sigma_v^2\E[D_\tau-\sigma_v^2d]\notag\\
 &\quad+2\E[\mu'WPv].                                 \label{ng:eq:five-corrections}
\end{align}

\emph{Step 2: the bounds \eqref{ng:eq:quadratic-kernel-bounds}.}
Because \(\E W=0\) and \(\Pi\), \(\xi\), and \(h/H\) are nonrandom,
\[
 \E[\mu'W\mu]=2\frac hH\E[\Pi'W\tilde m]+\E[m'Wm].
\]
For the Borel solution in
\Cref{lem:independent-support}, whose support lies in \(E\) by
\eqref{eq:kkt-conditions}, set
\(B=\supp(\widehat\pi)\cup T^0\). On \(\mathcal G\),
\eqref{ng:eq:sparse-bounds} gives \(|B|\le m^*\) and \(\tilde m=P_B\tilde m\).
We can bound
\(|\Pi'W\tilde m|\le\ngnorm{P_BW\Pi}\ngnorm{\tilde m}\), so that
\eqref{ng:eq:uniform-projections},
\eqref{ng:eq:centered-auxiliary-bound}, and Cauchy--Schwarz control the
first term on that event. For the second, we write \(m=\tilde m-(h/H)\xi\) and
use \(\lVert P_BWP_B\rVert_{\ngop}\le\eta\) on
\(\mathcal G\cap\mathcal W\) together with \(\ngnorm\xi^2\le Cs\). The complements are controlled by
\eqref{ng:eq:exceptional}, since
\(\ngnorm\mu\le\ngnorm{\Pi+v}\), \(\ngnorm\Pi=O(\sqrt n)\),
and \(\ngnorm\xi=O(\sqrt s)\) give
\(\ngnorm m+\ngnorm{\tilde m}\le C(\sqrt n+\ngnorm v)\), so that each quadratic
product is bounded by \(C(n+\ngnorm v^2)\). These arguments show that
\begin{align}
 |\E[\Pi'W\tilde m]|
 &\le C\sqrt{s\log(pn)(n\mathcal A+s)}+Cn^{-8},\notag\\
 |\E[m'Wm]|
 &\le C\{\eta(n\mathcal A+s)
          +\sqrt{s(n\mathcal A+s)}+s+n^{-8}\}.
                                                        \label{ng:eq:quadratic-kernel-bounds}
\end{align}

\emph{Step 3: the bounds \eqref{ng:eq:rank-kernel-bounds} and
\(|\E[\mu'WPv]|\le C\{s\log(pn)+\sqrt{n\mathcal A\,s\log(pn)}\}\).}
As in Step~3 of the proof of \Cref{ng:lem:uniform},
\(|D_\tau-\sigma_v^2d|\le\eta d\) on \(\mathcal G\cap\mathcal W\). Moreover,
\(0\le D_\tau\le Cd\) and \(0\le\ngtrace(TPTP)\le Cd\)
by \eqref{ng:eq:kernel-bounds} and \(P^2=P\).
Since \(\E[d^j]\le Cs^j\) for \(j=1,2\), these bounds control the
expectations over \(\mathcal G\cap\mathcal W\). On its complement,
\(D_\tau+d\le Cn\) and \(\P((\mathcal G\cap\mathcal W)^c)\le2n^{-20}\), so
that the quadratic and linear terms contribute \(O(n^{-18})\) and
\(O(n^{-19})\), respectively. We therefore have
\begin{align}
 |\E[D_\tau^2-\sigma_v^4d^2]|&\le C\eta s^2+Cn^{-18},\notag\\
 |\E[D_\tau-\sigma_v^2d]|&\le C\eta s+Cn^{-19},\notag\\
 |\E[\ngtrace(TPTP)-\sigma_v^4d]|&\le Cs.             \label{ng:eq:rank-kernel-bounds}
\end{align}
Finally, write \(\mu=(h/H)\Pi+m\). Since \(P\) is an orthogonal
projection and \(\lVert W\rVert_{\ngop}\le C\), the adaptive
projection moments in \eqref{ng:eq:adaptive-projections} give
\[
 |\E[\mu'WPv]|
 \le C\E[(\ngnorm{PW\Pi}+\ngnorm m)\ngnorm{Pv}]
 \le C\{s\log(pn)+\sqrt{n\mathcal A\,s\log(pn)}\}.
\]

\emph{Step 4: \eqref{ng:eq:numerator} and \eqref{ng:eq:gamma}.}
The five terms in \eqref{ng:eq:five-corrections}, bounded as above,
have total absolute value at most
\[
 C\bigl\{\eta\mathcal R+\sqrt{s\log(pn)(n\mathcal A+s)}
       +\sqrt{s(n\mathcal A+s)}+s\log(pn)+s+n^{-8}\bigr\}.
\]
Since \(s\ge1\), we have \(n\mathcal A+s\le\mathcal R\), and
\(\mathcal R\ge s^2\) by definition. Relative to \(\mathcal R\), this total is at most
\[
 C\left\{\eta+\sqrt{\log(pn)/s}+s^{-1/2}+\log(pn)/s+s^{-1}
           +n^{-8}/\mathcal R\right\}\longrightarrow0
\]
by \eqref{ng:eq:growth} and \(s\to\infty\). The conditional structural moments
\eqref{ng:eq:structural} now imply
\begin{align*}
 \E[(\varepsilon'\widehat\Pi)^2]
 =a^2\E[F^2]+\sigma_u^2\E[\ngnorm\mu^2]=\sigma_\varepsilon^2\E[\ngnorm\mu^2]
       +\sigma_{\varepsilon v}^2\E[d^2+d]+o(\mathcal R),
\end{align*}
which is \eqref{ng:eq:numerator}.

Lastly, \eqref{ng:eq:residual-stein} and \(D_\tau\le Cd\) yield
\[
 \E[(v'm)^2]
 \le2\E[(v'm-D_\tau)^2]+2\E[D_\tau^2]
 \le C\{n\mathcal A+s+\E d^2\}\le C\mathcal R.
\]
By \eqref{ng:eq:structural} and
\(\Gamma-L_n=\varepsilon'm/(\sqrt n\,h)\), we then have
\[
 \E[(\Gamma-L_n)^2]
 =\frac{a^2\E[(v'm)^2]+\sigma_u^2\E\ngnorm m^2}{nh^2}
 \le C\mathcal R/n\le CS.
\]
\end{proof}
{}
\Needspace{6\baselineskip}
\section{The Denominator and the Optimal Instrument}
\label{ng:sec:denominator}

\begin{lemma}[The Sample Denominator]
\label{ng:lem:denominator}
Under the hypotheses of \Cref{ng:thm:infeasible},
\begin{align}
 \E[\widehat h]&=h+\E[D_\tau]/n,
 &\E[(\widehat h-h)^2]&\le C/n,                   \label{ng:eq:denominator-moments}\\
 \widehat h-h&=O_p(n^{-1/2}),
 &\P(\mathcal H_n)&\to1.                               \label{ng:eq:denominator-rate}
\end{align}
Consequently, \(\Delta=O_p(n^{-1/2})\).
\end{lemma}

\begin{proof}
By the weighted first-moment identity
\eqref{ng:eq:offset-stein} with \(b=0\), \(\E[v'\mu]=\E[D_\tau]\), so
that \(\E[\widehat h]-h=\E[D_\tau]/n\), and \(\E[D_\tau]\le C\E[d]\le
Cs\) by \eqref{ng:eq:kernel-bounds} and \eqref{ng:eq:fit-moments}. For
the variance, the gradient of \(x'\mu\) with respect to \(v\) is
\(\mu+Px\) almost everywhere. Since \(\ngnorm\mu\le\ngnorm x\) by
\Cref{lem:geometry} and \(P\) is a projection, the weighted Poincar\'e
inequality \eqref{ng:eq:poincare} yields
\[
 \Var(\widehat h)
 \le\frac{C}{n^2}\E\ngnorm{\mu+Px}^2
 \le\frac{C}{n^2}\E\ngnorm x^2\le\frac Cn.
\]
By \eqref{ng:eq:growth}, \(s=o(\sqrt n)\), so that
\(\E[(\widehat h-h)^2]=\Var(\widehat h)+\E[D_\tau]^2/n^2
\le C/n+Cs^2/n^2\le C/n\), which proves
\eqref{ng:eq:denominator-moments}. Since \(h\ge\underline H/2\), the
first bound in \eqref{ng:eq:denominator-rate}, the bound on \(\Delta\),
and \(\P(\mathcal H_n^c)\le\P(|\widehat h-h|>\underline H/4)\le C/n\to0\)
then follow from Markov's inequality.
\end{proof}

\Needspace{6\baselineskip}
\begin{lemma}[Moments of the Denominator Correction]
\label{ng:lem:directional}
Under the hypotheses of \Cref{ng:thm:infeasible}, with \(\Delta\)
as in \eqref{eq:leading-term},
\begin{align}
 \E[\Delta^2L_n^2]&\le C/n,
                                                        \label{ng:eq:directional-moments}\\
 |\E[\Delta L_n^2]|&\le C\{\sqrt{S/n}+1/n\}.
                                                        \label{ng:eq:directional-signed}
\end{align}
\end{lemma}

\begin{proof}
We follow the proof of
\Cref{lem:denominator-correction} in the main appendix, replacing
\eqref{eq:stein-identity} by the weighted Stein identity
\begin{equation}
 \E[(\Pi'v)f(v)]=\E[(T\Pi)'\nabla f(v)].
                                                        \label{ng:eq:directional-identity}
\end{equation}
This identity is
\eqref{ng:eq:covariance} applied with the vector \(\Pi\) in place of the
field there, for which the weighted divergence in
\eqref{ng:eq:divergence} is \(\Pi'v\).

As in the Gaussian case, we use three facts. First,
\(\nabla\Delta=(\mu+Px)/(nh)\), so that
\(\E\ngnorm{\nabla\Delta}^2\le4\E\ngnorm x^2/(nh)^2\le C/n\). Second,
\Cref{ng:lem:denominator} gives \(\E[\Delta]=\E[D_\tau]/(nh)=O(s/n)\)
and \(\E[\Delta^2]\le C/n\). Third, by the conditional structural
moments \eqref{ng:eq:structural},
\begin{equation}
 \E[L_n^2\mid v]
 =\frac{\sigma_u^2}{H}+\frac{a^2}{nH^2}(\Pi'v)^2.
                                                        \label{ng:eq:oracle-conditional}
\end{equation}
Since \(\Delta\) depends only on \(v\), we may multiply
\eqref{ng:eq:oracle-conditional} by functions of \(\Delta\) and take
expectations.

\emph{Step 1: \(\E[\Delta^2L_n^2]\le C/n\).}
By \eqref{ng:eq:oracle-conditional}, it suffices to show that
\(\E[(\Pi'v)^2\Delta^2]\le C\). Apply
\eqref{ng:eq:directional-identity} with \(f=(\Pi'v)\Delta^2\). Because
\(\nabla(\Pi'v)=\Pi\), the product rule gives
\[
\begin{aligned}
 \E[(\Pi'v)^2\Delta^2]
 &=\E[(\Pi'T\Pi)\Delta^2]
       +2\E[(\Pi'v)\Delta\cdot(T\Pi)'\nabla\Delta]\\
 &\le Cn\E[\Delta^2]+\tfrac12\E[(\Pi'v)^2\Delta^2]
       +2\E[((T\Pi)'\nabla\Delta)^2],
\end{aligned}
\]
where the inequality uses \(\lVert T\rVert_{\ngop}\le C\) and Young's
inequality, as in Step~1 of the proof of
\Cref{lem:denominator-correction}. The last term is at most
\(Cn\E\ngnorm{\nabla\Delta}^2\le C\) because \(\ngnorm{T\Pi}^2\le Cn\).
Moving
\(\tfrac12\E[(\Pi'v)^2\Delta^2]\) to the left-hand side, we obtain
\(\E[(\Pi'v)^2\Delta^2]\le C\).

\emph{Step 2: \(|\E[\Delta L_n^2]|\le C\{\sqrt{S/n}+1/n\}\).}
Multiply \eqref{ng:eq:oracle-conditional} by \(\Delta\) and take
expectations. Applying \eqref{ng:eq:directional-identity} with
\(f=(\Pi'v)\Delta\), we have
\(\E[(\Pi'v)^2\Delta]=\E[(\Pi'T\Pi)\Delta]+\E[(\Pi'v)(T\Pi)'\nabla\Delta]\),
where \(\Pi'T\Pi=\sigma_v^2nH+\Pi'W\Pi\) because \(T=\sigma_v^2I+W\).
Since \(\sigma_\varepsilon^2=\sigma_u^2+a^2\sigma_v^2\), the two
multiples of \(\E[\Delta]\) again add up to \(\sigma_\varepsilon^2/H\),
so that
\begin{equation}
 \E[\Delta L_n^2]
 =\frac{\sigma_\varepsilon^2}{H}\E[\Delta]
   +\frac{a^2}{nH^2}
       \left\{\E[(\Pi'W\Pi)\Delta]+\E[(\Pi'v)(T\Pi)'\nabla\Delta]\right\}.
                                                        \label{ng:eq:denominator-signed-identity}
\end{equation}
The first term is \(O(s/n)\). Since \(W\) is diagonal with
independent, mean-zero entries and \(|\Pi_i|\le K_\Pi\) by
\eqref{ng:eq:identification},
\[
 \E[(\Pi'W\Pi)^2]
 =\Var(\tau(v_i))\sum_{i=1}^n\Pi_i^4\le Cn,
\]
so that \(|\E[(\Pi'W\Pi)\Delta]|\le C\) by Cauchy--Schwarz.

For the last term in \eqref{ng:eq:denominator-signed-identity}, split
\((T\Pi)'\nabla\Delta=\sigma_v^2\Pi'\nabla\Delta+(W\Pi)'\nabla\Delta\).
The first part satisfies
\(|\E[(\Pi'v)\Pi'\nabla\Delta]|\le C(1+\sqrt{n\mathcal A})\) by the
argument in Step~2 of the proof of \Cref{lem:denominator-correction},
with two changes. First, \eqref{ng:eq:directional-identity} gives
\(\E[(\Pi'v)\Pi'\mu]=\E[(T\Pi)'P\Pi]\), which is at most \(Cn\) in
absolute value. Second, the bound
\(\E[(\Pi'v)^4]+\E\ngnorm x^4\le Cn^2\) now follows from
\eqref{ng:eq:tails}. For the second part,
\(nh\,(W\Pi)'\nabla\Delta=(W\Pi)'(\mu+Px)=(PW\Pi)'(\mu+x)\), since
\(P\mu=\mu\). Applying Cauchy--Schwarz twice, with
\(\ngnorm{\mu+x}\le2\ngnorm x\), and using
\eqref{ng:eq:adaptive-projections}, we obtain
\[
\begin{aligned}
 |\E[(\Pi'v)(PW\Pi)'(\mu+x)]|
 &\le\left\{\E[(\Pi'v)^4]\right\}^{1/4}
    \left\{\E\ngnorm{\mu+x}^4\right\}^{1/4}
    \left\{\E\ngnorm{PW\Pi}^2\right\}^{1/2}\\
 &\le Cn\sqrt{s\log(pn)}.
\end{aligned}
\]
Dividing the second part by \(nh\) and combining it with the first, we
have
\[
 |\E[(\Pi'v)(T\Pi)'\nabla\Delta]|
 \le C\left\{1+\sqrt{n\mathcal A}+\sqrt{s\log(pn)}\right\},
\]
so that by \eqref{ng:eq:denominator-signed-identity},
\[
 |\E[\Delta L_n^2]|
 \le C\left\{\frac sn+\sqrt{\frac{\mathcal A}{n}}
             +\frac{\sqrt{s\log(pn)}}{n}+\frac1n\right\}.
\]
Finally, \(\mathcal A\le CS\) and \(s/n\le C\sqrt{S/n}\) by
\Cref{ng:lem:numerator}, while
\[
 \frac{\sqrt{s\log(pn)}}{n}=\frac sn\sqrt{\frac{\log(pn)}{s}}=o\!\left(\frac sn\right)
\]
by \eqref{ng:eq:growth}. This proves
\eqref{ng:eq:directional-signed}.
\end{proof}

\Needspace{4\baselineskip}
\begin{lemma}[Moments for the Optimal Instrument]
\label{ng:lem:oracle}
Under \Cref{ng:ass:sampling},
\begin{equation}
 \begin{aligned}
 \E[L_n^2]=\sigma_\varepsilon^2/H,\quad
 |\E[qL_n^2]|+\E[q^2L_n^2]\le C/n.
 \end{aligned}                                       \label{ng:eq:oracle-moment-rates}
\end{equation}
Consequently, \(C_0=(1-q)^2L_n^2\) satisfies
\begin{equation}
 \E[C_0]=\sigma_\varepsilon^2/H+O(n^{-1}).
                                                        \label{ng:eq:oracle-baseline-rate}
\end{equation}
\end{lemma}

\begin{proof}
By independence and centering, \(\E[L_n^2]=\sigma_\varepsilon^2/H\).
Equation \eqref{ng:eq:oracle-conditional} uses only
\eqref{ng:eq:structural}, so it holds under \Cref{ng:ass:sampling}.
Multiplying it by \(q=\Pi'v/(nH)\), which depends only on \(v\), and
using independence and \(\E v_i=0\), we have
\[
 \E[qL_n^2]
 =\frac{a^2\E[(\Pi'v)^3]}{n^2H^3}
 =\frac{a^2\E[v_i^3]\sum_i\Pi_i^3}{n^2H^3}
 =O(n^{-1}),
\]
where \(\sum_i|\Pi_i|^3\le K_\Pi nH\) by \eqref{ng:eq:identification}
and \(\E|v_i|^3\le C\) by \eqref{ng:eq:tails}. By the same tail bound,
\(\E[(\Pi'v)^4]\le Cn^2\), and by independence
\(\E[q^2]=\sigma_v^2/(nH)\). Therefore
\[
 \E[q^2L_n^2]
 =\frac{\sigma_u^2}{H}\E[q^2]
   +\frac{a^2\E[(\Pi'v)^4]}{n^3H^4}=O(n^{-1}).
\]
Expanding \(C_0=(1-2q+q^2)L_n^2\) proves
\eqref{ng:eq:oracle-baseline-rate}.
\end{proof}
{}
\Needspace{12\baselineskip}
\section{Supporting Results for \Cref{ng:thm:feasible}}
\label{ng:sec:feasible}

\Needspace{7\baselineskip}
\begin{lemma}[Sample Counterparts of \(\mathcal A\) and \(q\)]
\label{ng:lem:fitted-moment}
Under the hypotheses of \Cref{ng:thm:infeasible}, let
\(\widehat{\mathcal A}=\min_{b\in\mathbb R}
\En[(\widehat\Pi_i-b\Pi_i)^2]\), and let \(q\) and \(\widehat q\)
be as in the proofs of the main-text \Cref{thm:main,thm:feasible}. Then
\begin{align}
 \widehat{\mathcal A}-\mathcal A
 &=O_p\!\left(\sqrt{\mathcal A/n}+n^{-1}\right)
   =o_p(S),                                          \label{ng:eq:sample-approximation-rate}\\
 \widehat q-q
 &=O_p\!\left(\sqrt{\mathcal A/n}+s/n+n^{-1}\right)
   =o_p(S),                                          \label{ng:eq:fitted-score-rate}
\end{align}
where the second assertion holds on \(\mathcal H_n\).
\end{lemma}

\begin{proof}
\emph{Step 1: \(\widehat{\mathcal A}-\mathcal A
=O_p(\sqrt{\mathcal A/n}+1/n)\).}
Step~1 of the proof of \Cref{lem:fitted-moment-expansion} uses Gaussian
errors only through the Poincar\'e inequality, which it applies to
\(\En[\Pi_i\widehat\Pi_i]\) and \(\widehat{\mathcal A}\). By
\eqref{ng:eq:projection}, their gradients are again \(P\Pi/n\) and
\((2/n)P\{I-\Pi\Pi'/(nH)\}\widehat\Pi\). With the weighted Poincar\'e
inequality \eqref{ng:eq:poincare} in place of
\eqref{eq:gaussian-poincare}, the same argument shows that
\(0\le\mathcal A-\E[\widehat{\mathcal A}]\le C/n\) and
\(\Var(\widehat{\mathcal A})\le C\mathcal A/n\), and the claim follows
from Chebyshev's inequality.

\emph{Step 2: \(\widehat q-q
=O_p(\sqrt{\mathcal A/n}+s/n+1/n)\) on \(\mathcal H_n\).}
Since \(\widehat\Pi=(h/H)\Pi+m\) by
\eqref{ng:eq:centered-fit-definition}, we have
\(\En[v_i\widehat\Pi_i]=v'm/n+hq\), and
\(\E[(v'm-D_\tau)^2]\le C(n\mathcal A+s)\) by
\eqref{ng:eq:residual-stein}. Also, \(0\le D_\tau\le Cd=O_p(s)\) by
\Cref{ng:lem:kernel,ng:lem:sparse}. Since \(s\ge1\), it follows that
\(v'm/n=O_p(\sqrt{\mathcal A/n}+s/n)\). By independence,
\(\E[q^2]=\sigma_v^2/(nH)\), while \(\widehat h-h=O_p(n^{-1/2})\) by
\Cref{ng:lem:denominator}. On \(\mathcal H_n\), where
\(\widehat h\ge\underline H/4\),
\[
 \widehat q-q
 =\frac{v'm/n-q(\widehat h-h)}{\widehat h}
 =O_p\!\left(\sqrt{\mathcal A/n}+s/n+n^{-1}\right).
\]

\emph{Step 3: Both rates are \(o_p(S)\).}
By \Cref{ng:lem:numerator}, \(S\asymp\mathcal A+s^2/n\). The inequality
\(2s\sqrt{\mathcal A/n}\le\mathcal A+s^2/n\) and \(s^2/n\le CS\) imply
\[
 \frac{\sqrt{\mathcal A/n}+s/n+n^{-1}}{S}
 \le C(s^{-1}+s^{-2})\longrightarrow0,
\]
since \(s\to\infty\) by \eqref{ng:eq:growth}.
\end{proof}

\Needspace{4\baselineskip}
\begin{lemma}[Estimating the Error Moments]
\label{ng:lem:pilot}
Suppose the hypotheses of \Cref{ng:thm:infeasible} hold.  For any (possibly
data-dependent) \([J]\)-valued index \(\check c\),
\(\widehat\beta_{\check c}-\beta=O_p(n^{-1/2})\).  If the pilot
satisfies \(\check\beta-\beta=O_p(n^{-1/2})\), the moment estimators in
the main-text \eqref{eq:nuisance-estimators} satisfy
\begin{equation}
 \begin{aligned}
 \widehat\sigma_\varepsilon^2-\sigma_\varepsilon^2=O_p(n^{-1/2}),\quad
 \widehat\sigma_{\varepsilon v}-\sigma_{\varepsilon v}=O_p(n^{-1/2}).
 \end{aligned}
                                                        \label{ng:eq:pilot-rates}
\end{equation}
\end{lemma}

\begin{proof}
Step~1 of the proof of \Cref{ng:thm:infeasible} shows that
\(\sqrt n(\widehat\beta-\beta)=O_p(1)\), which proves the first claim
because \(\check c\in[J]\) and \(J\) is fixed. By the moment
conditions \eqref{ng:eq:conditional}, \eqref{ng:eq:tails}, and
\eqref{ng:eq:identification},
\(\sup_i\E[\varepsilon_i^4+x_i^4]<\infty\), and independence and
Chebyshev's inequality give
\[
 \En[\varepsilon_i^2]=\sigma_\varepsilon^2+O_p(n^{-1/2}),\qquad
 \En[\varepsilon_ix_i]=\sigma_{\varepsilon v}+O_p(n^{-1/2}),
\]
while \(\En[x_i^2]=O_p(1)\) by Markov's inequality. Substituting these bounds
and \(\check\beta-\beta=O_p(n^{-1/2})\) into the algebraic identities
\eqref{eq:residual-moment-expansion} proves \eqref{ng:eq:pilot-rates}.
\end{proof}
{}
{}
\startsupplementbibliography

\finishsupplementbibliography

\end{document}